\documentclass[11pt]{article}

\usepackage{graphicx} 
\usepackage{verbatim}

\usepackage{enumitem}

\usepackage{booktabs, multirow}

\usepackage{amsmath, amssymb, mathtools, bbm}
\usepackage{amsthm}
\usepackage{bm}
\allowdisplaybreaks

\usepackage[labelfont=bf]{caption}
\usepackage{subcaption}
\usepackage[table]{xcolor}
\usepackage{float}
\usepackage{tikz}
\usetikzlibrary{arrows.meta,positioning}

\usepackage[scaled=0.8]{beramono}

\usepackage{xcolor}
\definecolor{linkcolor}{HTML}{0645AD}
\definecolor{crimson}{HTML}{A41034}
\definecolor{darkgreen}{HTML}{6D8E42}
\definecolor{lightblue}{HTML}{ADD8E6}
\definecolor{brickred}{HTML}{AA4A44}

\usepackage{hyperref}
\hypersetup{
     colorlinks   = true,
     allcolors    = crimson,
}

\newcommand{\indep}{\!\perp\!\!\!\perp}

\newcommand{\E}{\mathbb{E}}

\def\*#1{\mathbf{#1}}
\newcommand{\argmax}{\mathop{\rm argmax}\limits}
\newcommand{\argmin}{\mathop{\rm argmin}\limits}
\newcommand{\bY}{\bm{Y}}
\newcommand{\by}{\bm{y}}
\newcommand{\bx}{\bm{x}}
\newcommand{\bX}{\bm{X}}

\newcommand{\bL}{\bm{L}}

\newcommand{\bu}{\bm{u}}
\newcommand{\bU}{\bm{U}}

\newcommand{\bS}{\bm{S}}
\newcommand{\bbone}{\mathbbm{1}}

\newcommand{\ROR}{\textsc{ror}}
\newcommand{\cash}{\textsc{cash}}
\newcommand{\crime}{\textsc{crime}}
\newcommand{\nocrime}{\neg\textsc{crime}}
\newcommand{\LCB}{\textsc{lcb}}
\newcommand{\HCB}{\textsc{hcb}}
\newcommand{\REM}{\textsc{rem}}
\newcommand{\NCA}{\textsc{nca}}
\newcommand{\NVCA}{\textsc{nvca}}
\newcommand{\cD}{\mathcal{D}}
\newcommand{\cR}{\mathcal{R}}
\newcommand{\cX}{\mathcal{X}}
\newcommand{\cY}{\mathcal{Y}}

\newtheorem{theorem}{Theorem}
\newtheorem{proposition}{Proposition}
\newtheorem{corollary}{Corollary}

\newtheorem{assumption}{Assumption}
\newtheorem{remark}{Remark}

\newtheorem{example}{Example}
\theoremstyle{definition}
\newtheorem{definition}{Definition}

\usepackage{natbib}

\newcommand{\blind}{0}
\newcommand{\tit}{Triage Score:\\ {\Large A Counterfactual Risk Assessment Instrument}}
\begin{document}

\date{\today}

\if0\blind

\title{{\bf\tit}\thanks{We acknowledge partial financial support from
    Arnold Ventures and the Impact Labs at Harvard Kennedy School.}}

  \author{Kosuke Imai\thanks{Professor, Department of Government and
      Department of Statistics, Harvard University.  1737 Cambridge
      Street, Institute for Quantitative Social Science, Cambridge MA
      02138.  Email: \href{mailto:imai@harvard.edu}{imai@harvard.edu}
      URL:
      \href{https://imai.fas.harvard.edu}{https://imai.fas.harvard.edu}}
    \and Sooahn Shin\thanks{Postdotcoral Associate, Department of Political Science, Massachusetts Institute of Technology. 30 Wadsworth St, Cambridge, MA 02142.
    Email:
      \href{mailto:sshin3@mit.edu}{sshin3@mit.edu} URL:
      \href{https://sooahnshin.com}{https://sooahnshin.com}} \and
    D. James Greiner\thanks{Honorable S.  William Green Professor of
      Public Law, Harvard Law School, 1525 Massachusetts Avenue,
      Griswold 504, Cambridge, MA 02138.} \and Ryan Halen\thanks{Data
      Analyst, Access to Justice Lab at Harvard Law School, 1607
      Massachusetts Avenue, Third Floor, Cambridge, MA 02138.}}

\fi

\if1\blind
\title{\bf \tit}

\fi

\maketitle

\begin{abstract}
  Risk assessment instruments, also known as ``risk scores,'' are
  widely used in high-stakes decision-making settings such as medicine
  and the criminal justice system. A risk score predicts the
  likelihood of an undesired outcome if no intervention is made.
  Thus, a sufficiently high score is often interpreted as a
  recommendation to intervene. However, risk scores fail to account
  for what would happen if a decision-maker does intervene.  This failure is problematic because effective decision making requires consideration of both (or multiple) potential outcomes.  
  We propose ``triage
  scores,'' which are based on additive counterfactual utilities and
  include risk scores as a special case.  Unlike risk scores, triage
  scores can incorporate counterfactual outcomes under alternative
  decisions, enabling decision makers to incorporate a wide range of
  ethical and practical factors. We illustrate the use of triage
  scores with an application to our own randomized controlled trial
  evaluating a pretrial risk score.  Our analysis
  demonstrates that triage scores are able to capture rich utility
  structures and yield substantively distinct results
  regarding policy evaluation and learning.
\end{abstract}

\section{Introduction}


Today, data-driven algorithms are deeply embedded in decision-making
systems. In high-stakes settings such as medicine and the criminal
justice system, human decision makers frequently rely on
recommendations produced by risk assessment instruments, commonly
referred to as ``risk scores'' \citep{stevenson2018assessing,
  chen2021probabilistic}. 
These scores typically classify the
predicted probability of an undesirable outcome (e.g., illness or
rearrest) under a baseline decision of no intervention. In practice,
individuals with sufficiently high scores are often flagged for
intervention, such as admission to intensive care or the imposition of
cash bail. This workflow implicitly treats prediction under no
intervention as a proxy for prescription. However, identifying
individuals who are ``high risk'' in the absence of intervention
provides no information about how those same individuals would respond
if an intervention were applied.

This one-sided focus of risk scores is fundamentally
misaligned with the objective of maximizing overall welfare in such
decision-making systems. Effective decision making requires evaluating
what would happen under alternative choices, rather than focusing
solely on outcomes under a single baseline of no intervention. In
pretrial settings, for example, a central question is not simply the
probability that a defendant will be rearrested if assigned no cash bail, but how that probability would change under cash versus no cash (perhaps with differing monitoring and support conditions). In addition, decision makers must weigh the societal
costs of recidivism together with the financial, ethical, and other
costs associated with cash bail and pretrial detention under each
possible course of action. Because risk scores are indexed to
a single baseline decision, they cannot distinguish between
individuals for whom an intervention would meaningfully change
outcomes and those for whom it would not. This failure to consider
counterfactual outcomes also limits our ability to reason about
ethical and practical tradeoffs, such as the regret associated with
unnecessarily detaining an individual who would not have been
rearrested if released.

In this paper, we propose a framework based on ``triage scores,'' a class of {\it
  counterfactual} risk assessment instruments designed to align more closely decision making with the objective of maximizing expected
utility. 
Although we refer in this paper to ``triage scores,'' our focus here is on the development of a counterfactual and evaluative framework for optimizing decision making that incorporates risk from alternative decisions; we do not here actually develop a triage score for any particular setting nor discuss how to do so.
Rather than summarizing risk under a single baseline
potential outcome, triage scores are constructed from counterfactual
utilities that depend on the joint distribution of potential outcomes
under alternative decisions. Our formulation builds on recent advances
in statistical decision theory with counterfactual utilities
\citep[e.g.,][]{coston2020counterfactual, benm:imai:jian:23,
  ben2024does, mueller2023personalized, christy2024starting,
  koch2025statistical}. By explicitly modeling both the realized
outcome under the chosen decision and the counterfactual outcomes
under alternative decisions, triage scores allow utilities to encode
considerations such as regret from unnecessary detention or failure to
prevent a crime. Under an additive counterfactual utility structure
and the standard assumption of unconfoundedness, we point identify the
expected utility of a decision-making system, enabling systematic
evaluation of existing policies as well as learning of new,
optimal decision rules.

Conceptually, triage scores generalize standard risk scores. In the
simplest binary setting with two decisions (e.g., no cash bail versus cash bail) and a binary outcome (e.g., at least one pretrial rearrest versus no arrest),
conventional risk scores depend solely on the baseline potential
outcome. In contrast, triage scores operate at the level of principal
strata defined by the joint potential outcomes \citep{fran:rubi:02}
and assign a utility to each possible decision within a stratum. This
framework enables decision makers to distinguish, for example, among
defendants who would not be rearrested under either no cash bail or
cash bail (\emph{safe}), those for whom cash bail would be
counterproductive (\emph{backlash}), those for whom cash bail would
prevent crime (\emph{preventable}), and those likely to reoffend
regardless of the decision (\emph{hopeless}).

We develop a general statistical framework for evaluating and learning from triage scores using data generated by human decision makers with or
without algorithmic recommendations. We apply this framework to assess
and improve the Public Safety Assessment (PSA), a pretrial risk assessment instrument designed to inform initial release/bail decisions, using data from our own randomized controlled trial
\citep[see][for a related RCT in Wisconsin]{imai2023experimental}.  In
this RCT, algorithmic recommendations were randomly assigned across arrested individuals under a single-blind design, such that defendants were unaware
of their assignment.  We also observe the full set of information
available to judges at the time of decision, rendering plausible the assumption that our model of judicial decisions is unconfounded.

Under these assumptions, we establish identification results for the
expected counterfactual utility of three decision-making systems:
human-alone, human-with-AI, and AI-alone. We construct
semiparametric estimators based on augmented inverse probability
weighting to estimate these utilities, adjusting for the full set of
information available to judges at the time of decision. In
particular, we incorporate GenAI-powered inference (GPI) to account
for unstructured confounding information contained in probable-cause
affidavits \citep{imai2025genai}. Finally, we show how to estimate an
optimal decision rule under the triage score utility framework by
solving an empirical utility maximization problem over a class of
feasible decision policies.

The remainder of the paper is organized as
follows. Section~\ref{sec:empirical} describes the Utah experiment and
the PSA and introduces a statistical
decision-theoretic framework for pretrial decision
making. Section~\ref{sec:triage} formally defines triage scores in settings with binary or multivalued decisional choices and clarifies their
relationship to standard risk scores. Section~\ref{sec:statistical}
presents our identification results and estimation strategy and
describes how to learn optimal decision rules from data under the
triage score utility framework.
Section~\ref{sec:analysis} reports empirical findings from the Utah
application under a range of utility specifications.  Finally,
Section~\ref{sec:conclusion} concludes by discussing broader
implications for the design and evaluation of algorithm-assisted
decision-making systems.

A short note regarding terminology: for simplicity and brevity, we sometimes refer to a judge's pretrial decision as `release' or `release own recognizance (ROR)'. More accurate phrasing would be `no cash bail required for release in this judicial proceeding.' We clarify this distinction because a judge's decision not to assign cash bail after a particular arrest does not always mean that the defendant will achieve release. Instead, a defendant might remain incarcerated because the present arrest violated terms of pretrial release on a different set of charges, or because the U.S. Immigration and Customs Enforcement may want to begin deportation proceedings, or because the defendant requires detoxing or a psychiatric evaluation. Nevertheless, `no cash bail required for release in this judicial proceeding' is cumbersome, so at the risk of some distortion we use the shorthand `release' for brevity.
The key point is that `release' in this setting refers to the judge's decision, not necessarily what the defendant experiences.

\subsection*{Related Literature}

Risk assessment tools are widely used across high-stakes decision-making domains. In the criminal
justice system, risk scores are routinely employed to inform pretrial and sentencing decisions
\citep[e.g.,][]{stevenson2018assessing, albright2019if}. In clinical medicine, risk prediction models
play a central role in diagnosis and treatment decisions \citep[e.g.,][]{dagostino2008general,
chen2021probabilistic}. Similar tools are also prevalent in consumer finance, where credit risk
models guide lending decisions \citep[e.g.,][]{einav2013impact, dobbie2021measuring}.

There has been a substantial body of research on the development and evaluation of risk scores across these
domains \citep[e.g.,][]{goel2016personalized, kleinberg2018human, coston2021characterizing,
angelova2025algorithmic}. A central methodological challenge in evaluating decision-making systems
is the selective labels problem \citep{lakkaraju2017selective}: outcomes are only observed for
individuals who receive a particular decision, complicating counterfactual evaluation under
alternative decisions. Existing approaches address this challenge using algorithmic thresholds and
staggered rollouts \citep[e.g.,][]{berk2021fairness, stevenson2022algorithmic, guerdan2023ground},
survey-based evaluations \citep[e.g.,][]{miller2013practitioner, skeem2020impact}, or quasi-random
assignment to decision-makers \citep[e.g.,][]{dobbie2018effects, arnold2022measuring}.

Most closely related to our work is \cite{ben2024does}, which
formulates the evaluation of decision-making systems using a
confusion-matrix representation grounded in the potential outcomes
framework. That approach considers experimental settings, in which the
provision of algorithmic recommendations is randomized, and can be
extended to observational settings under unconfoundedness between
recommendation provision and potential outcomes. Our work differs in
two key respects. First, while \cite{ben2024does} focuses on a
baseline potential outcome, consistent with standard risk assessment
frameworks, we consider joint potential outcomes under alternative
decisions, resulting in the construction of triage scores that
generalize standard risk scores.  Second, we consider an additive
counterfactual utility structure under the assumption of
unconfoundedness between human decisions and potential outcomes.

As noted earlier, the unconfoundedness assumption is plausible in our
application because we observe all information available to judges at
the time of decision making, including probable cause affidavits that
contain rich unstructured text describing the arrest. To adjust
flexibly for such high-dimensional text confounders, we draw on recent
advances in causal inference with texts \citep[see][for a
review]{feder2022causal}.  In particular, we apply GenAI-powered
inference methods
\citep{imai2024causalrepresentationlearninggenerative}, using internal
representations from open-source large language models \citep[e.g.,
Llama3 developed by][]{touvron2023llama} to estimate deconfounder
functions \citep{imai2025genai}.

Finally, our work builds on and contributes to a growing literature on
the evaluation and learning of decision-making systems using
counterfactual utilities \citep[e.g.,][]{coston2020counterfactual,
  rambachan2022robust, benm:imai:jian:23, mueller2023personalized,
  ben2024does, christy2024starting, koch2025statistical}. In
particular, we adapt the identification strategy based on additive
counterfactual utilities developed in \cite{koch2025statistical} to
algorithm-assisted decision-making settings, and further develop
semiparametric estimators for evaluation and optimal policy learning.

\section{Empirical Application}
\label{sec:empirical}

We now introduce an empirical application that motivates the proposed
methodology.  Our application is based on our own RCT in Utah, which was designed to evaluate the
value of a prominent risk score used in criminal justice system.
Below, we briefly explain the design of this RCT and present a basic
descriptive analysis of data.  Finally, we discuss questions based on
this RCT that motivate the development of the triage score utility framework.

\subsection{A Randomized Controlled Trial in Utah}

This field RCT was part of a series of experiments designed to assess
the effect of providing a risk score called the PSA to judges making bail and pretrial release condition
decisions shortly after arrest. 
All sites used similar RCT designs. In Utah, which provided the data for our application here, judges either did or did not receive the PSA when making the first post-arrest decision regarding release, bail, and monitoring conditions. Randomization was by defendant, meaning a defendant remained in either a judge-receives-PSA condition or a judge-does-not-receive-PSA condition for all their arrests.
Primary outcomes included new criminal activity (NCA), denoting a
$0$--$1$ variable for whether the defendant was arrested or cited for
an incarceration-eligible offense; new violent criminal activity
(NVCA), the same as NCA but for only violent offenses; and failure to
appear (FTA), denoting a $0$--$1$ variable for whether the court
issued a bench warrant stemming from a defendant's missing a required
court date. By definition, NCA, NVCA, and FTA could not occur during
time periods on which the defendant was incarcerated.

The Utah RCT took place in four counties, Davis County, Utah County,
Weber County, and Morgan County. As is true of all field operations,
the Utah sites had its own esoteric features. 
First, Utah's automated systems could produce the PSA only for arrestees who did not have records from other states that the Utah software could not machine read (and translate into PSA inputs), meaning that an unknown but not huge fraction of arrestees were not included in the study population. 
Second, Utah's
pretrial system required a judge to make the first consequential
decision regarding pretrial release and bail, not at a live hearing,
but rather upon a review of paper files only and concurrent with the
judge's determination of whether probable cause existed for the
arrest. Because there was no live hearing, there were no statements or
arguments from the defendant, from defense counsel, from the
prosecution, or from anyone else. The judge made all decisions based
on files that consisted of the charges, the law enforcement probable
cause (PC) affidavit, whatever online criminal history search the
judge decided to conduct, and the PSA (if the case was randomized to
the PSA-present condition).

The law enforcement PC affidavit was a sworn statement, almost always
from the arresting officer, describing in a paragraph or two the
circumstances of the arrest. The PC affidavit was supposed to
particularize the reasons why the officer believed that the defendant
had committed a criminal offense. Constitutional law (Gerstein
v. Pugh, 420 U.S. 103 (1975)) required a judicial officer to conduct
an independent review shortly after arrest of a law enforcement
officer's warrantless decision to arrest. Utah judges reviewed the
materials listed above by logging into an online system and indicating
their decisions electronically.

In cooperation with the Utah judiciary, we obtained the PC affidavits
as well as the other materials that judges observed when making their
decisions. Thus, we were able to observe everything that the judge
observed for each decision, rendering plausible the assumption,
described below, of unconfounded decisions conditional on observed
covariate information.

\subsection{Public Safety Assessment Instrument (PSA)}

The PSA is a set of three integer scores categorizing a defendant's
risk of FTA, NCA, and NVCA. The FTA and NCA scores run from $1$ to
$6$, while the NVCA metric takes the form of a $0$--$1$ flag. Higher
numbers corresponded to higher risk (according to the PSA). Eight
criminal history factors plus age serve as the PSA's inputs (see
\url{https://www.advancingpretrial.org/about-the-psa/}). The PSA
scores, which are the same for all jurisdictions, serve as inputs to a
jurisdiction-specific Decision Making Framework (DMF), which
incorporates local circumstances and values to transform the scores
into a recommendation for the judge regarding release, bail, and
monitoring conditions. The output of the PSA-DMF System (we refer in
this paper to the ``PSA'' for brevity) takes the form of a paper
printout or a computer file that reports the PSA scores, the values of
the nine inputs, the criminal history events that gave rise to those
values, and the DMF recommendation. A philanthropic foundation called
Arnold Ventures funded scientists to construct the PSA, and as of this
writing, dozens of jurisdictions across the United States provide it
to judges to guide initial release decisions.

\subsection{Data}

The dataset comprises a total of 9,855 cases, and we restrict the
sample to first-arrest cases (i.e., for arrestees who are arrested
multiple times in our dataset, we only consider their first arrest
during the study period). Among arrestees in this analytic sample,
$42\%$ are non-white males, $13\%$ are non-white females, $33\%$ are
white males, and $12\%$ are white females. The provision of the PSA
recommendation ($Z_i$) is randomized. The decision-maker in the
treated group is a judge who receives the PSA recommendation
($Z_i = 1$), whereas the decision-maker in the control group is the
same judge but without the PSA recommendation ($Z_i = 0$).  We
evaluate both the provision of the
PSA (which is what we randomized) as well as the DMF's recommendation,
defined as a dichotomized version of the PSA recommendation ($R_i$): whether it recommends release
on own recognizance (ROR) or not.

Table~\ref{table:decision} presents contingency tables comparing
decisions made under two different conditions. The left panel compares
decisions made by the judge without the PSA recommendation (control
group) to the PSA recommendation, and the right panel compares
decisions made by the same judge with the PSA recommendation (treated
group) to the PSA recommendation.  The judge is generally harsher than
the PSA. $52.7\%$ of cases in the control group and $46.1\%$ in the
treated group receive cash bail even when the PSA recommended ROR.  In contrast, when the PSA did not recommend ROR, the judge assigned ROR in about 5\% of cases in both treated and
control groups.

\begin{table}[t!] 
\centering\setlength{\tabcolsep}{3pt}
\begin{tabular}[t]{ll|ll}
\toprule
 & & \multicolumn{2}{c}{PSA} \\
 & & \multicolumn{1}{c}{ROR} & \multicolumn{1}{c}{Non-ROR}\\
\midrule
\multirow{2}{*}{\shortstack{Judge\\ without PSA}}&
ROR & 13.3\% (671) & \hspace{.025in} 5.0\% (253)\\
& Cash & 52.7\% (2655) & 29.0\% (1462)\\
\bottomrule
\end{tabular}
\begin{tabular}[t]{ll|ll}
\toprule
 & & \multicolumn{2}{c}{PSA} \\
 & & \multicolumn{1}{c}{ROR} & \multicolumn{1}{c}{Non-ROR}\\
\midrule
\multirow{2}{*}{\shortstack{Judge\\with PSA}}&
ROR & 18.0\% (867) & \hspace{.025in} 5.6\% (270)\\
& Cash & 46.1\% (2217) & 30.3\% (1460)\\
\bottomrule
\end{tabular}
\caption{Comparison of judge's decisions (rows) and PSA
  recommendations (columns).  ``ROR'' refers to release on own
  recognizance while ``Cash'' indicates the imposition of cash bail.
  The left panel is for the control group of cases where PSA was not
  provided, whereas the right panel is for the treated group where the
  judge was given PSA.}
\label{table:decision}
\end{table}

Table~\ref{table:outcome} reports the proportion of new criminal
activity (NCA) under each combination of the judge's decision and the
PSA recommendation.  For example, in the left panel, among cases in
which both the judge and the PSA agreed upon ROR, $9.8\%$ of
arrestees were rearrested for NCA within two years of
randomization. Overall, the NCA proportion is higher following cash
bail decisions or non-ROR recommendations than following ROR decisions or
recommendations. However, these raw associations do not account for
counterfactual outcomes and therefore cannot answer questions such as
how often judges make correct decisions, whether PSA recommendations
improve judicial decision making, or how accurate the PSA itself
is. In particular, we do not observe what would have happened in cash
bail cases had ROR decision been given instead. To address these
limitations, this paper proposes a statistical decision-theoretic
framework that enables researchers and policymakers to evaluate
decision quality and the value of algorithmic recommendations.

\begin{table}[t!]
\centering\setlength{\tabcolsep}{3pt}
\begin{tabular}[t]{ll|ll}
\toprule
 & & \multicolumn{2}{c}{PSA} \\
 & & \multicolumn{1}{c}{ROR} & \multicolumn{1}{c}{Non-ROR}\\
\midrule
\multirow{2}{*}{\shortstack{Judge\\ without PSA}}&
ROR & 9.8\% & 17.4\%\\
& Cash & 14.0\% & 23.1\%\\
\bottomrule
\end{tabular}
\begin{tabular}[t]{ll|ll}
\toprule
 & & \multicolumn{2}{c}{PSA} \\
 & & \multicolumn{1}{c}{ROR} & \multicolumn{1}{c}{Non-ROR}\\
\midrule
\multirow{2}{*}{\shortstack{Judge\\with PSA}}&
ROR & 11.8\% & 15.9\%\\
& Cash & 14.8\% & 23.9\%\\
\bottomrule
\end{tabular}
\caption{New criminal activity proportion by decision. Left: control group. Right: treated group.}
\label{table:outcome}
\end{table}

\subsection{Statistical Decision Theory with Counterfactual Utilities}

We formalize the evaluation and potential improvement of PSA by
applying the statistical decision theory based on counterfactual
utilities \citep{wald1950statistical,koch2025statistical}. We first
consider the {\it standard} utilities by specifying a utility for each
decision and its consequence without considering counterfactual
outcomes.  In our application, a judge chooses one among the following
four alternative decisions: release on own recognizance (\ROR), low
cash bail (\LCB), high cash bail (\HCB), and remand (\REM).
`Remand' here means that the judge decides that the defendant must remain incarcerated without the opportunity to post bail; in Utah, \REM decisions were uncommon but not so uncommon that we could ignore them. 
Our outcome variable of interest is binary, indicating whether or not a
defendant is rearrested for a new criminal activity upon release
(\NCA).  
The symbol $\neg$ denotes negation; for example, $\neg\NCA$
denotes no \NCA. 
Note that an arrestee may not be released immediately due to a
\REM{} decision or failure to pay cash bail. However, release may
subsequently occur if the initial decision is modified or overturned
at a later court hearing.

This setup implies that we have two possible outcomes under each of
four decisions, requiring the specification of eight utilities.  Here,
we provide an example of standard utilities.
\begin{example}[Additive standard utilities] \label{ex:standard} We
  may specify the utility for the \REM{} decision followed by the
  \NCA{} outcome as an additive function of two costs, which
  negatively contribute to the utility:
  \begin{equation*}
    \texttt{Utility}(\REM, \NCA) 
    \ = \ - ( \texttt{Cost}_{\REM} + \texttt{Cost}_{\NCA} ),
  \end{equation*}
  where $\texttt{Cost}_{\REM} \ge 0$ is the cost associated with the
  \REM{} decision, which may include the cost of detention on both the
  public and the defendant, and $\texttt{Cost}_{\NCA} \ge 0$ is the
  cost associated with the \NCA{} outcome, which may include the
  societal cost of new criminal activity. One can similarly define the
  other seven utilities.
\end{example}

In this paper, we apply the {\it counterfactual} statistical decision
theory \citep[e.g.,][]{coston2020counterfactual, benm:imai:jian:23,
  mueller2023personalized, ben2024does, christy2024starting,
  koch2025statistical}.  This framework generalizes the standard
statistical decision theory by incorporating counterfactual outcomes
under alternative decisions when specifying each utility.  Thus, the
counterfactual decision theory allows for the notion of ``regret'' by
comparing the outcome under each decision with counterfactual outcomes
under alternative decisions
\citep[e.g.,][]{bell1982regret,loomes1982regret}. For example, when
specifying a utility for the \REM{} decision, we consider the
counterfactual outcomes under different decisions (e.g., \ROR{}
decision) as well as the outcome that would be realized under the
\REM{} decision.

Thus, in our application, while there is only one standard utility
that can be specified for each of the eight decision-outcome pairs,
there are a total of 64 $(=2^6)$ counterfactual utilities that can
possibly be specified if one wishes to place no restriction.  However,
substantial simplification is required to facilitate interpretation
and practical use.  In this paper, we adopt the additive
counterfactual utility framework of \cite{koch2025statistical}, in
which the utility for each decision-outcome pair equals the sum of the
standard utility and separate counterfactual utilities under
alternative decisions.
\begin{example}[Additive counterfactual utilities] 
  Consider the utility for the \REM{} decision followed by the \NCA{}
  outcome.  Suppose that the counterfactual outcomes for all three
  alternative decisions are identical and are equal to the absence of
  \NCA{} event. Then, an additive counterfactual utility is given by,
  \begin{equation*}
    \texttt{Utility}(\REM, \NCA) + 
    \widetilde{\texttt{Utility}}(\ROR, \neg \NCA) + \widetilde{\texttt{Utility}}(\LCB, \neg \NCA) + \widetilde{\texttt{Utility}}(\HCB, \neg \NCA),
  \end{equation*}
  where the first term is the standard utility whose example is given
  in Example~\ref{ex:standard}, and the other three terms are
  counterfactual utilities associated with alternative decisions and
  counterfactual outcomes under those decisions.
\end{example}

To operationalize each counterfactual utility, we may use the notion
of regret, which weighs a cost of an alternative decision and another
cost of the counterfactual outcome that would have resulted under this
alternative decision.
\begin{align*}
  \widetilde{\texttt{Utility}}(\ROR, \neg \NCA)
  &= -\texttt{Regret}^\ROR_{\neg \NCA} \times (\texttt{Cost}_{\ROR} + \texttt{Cost}_{\neg \NCA}) \\
  \widetilde{\texttt{Utility}}(\LCB, \neg \NCA)
  &= -\textrm{Regret}^\LCB_{\neg \NCA} \times (\texttt{Cost}_{\LCB} + \texttt{Cost}_{\neg \NCA}) \\
  \widetilde{\texttt{Utility}}(\HCB, \neg \NCA)
  &= -\textrm{Regret}^\HCB_{\neg \NCA} \times (\texttt{Cost}_{\HCB} + \texttt{Cost}_{\neg \NCA})
\end{align*}
where
$\texttt{Regret}^\ROR_{\neg \NCA}, \textrm{Regret}^\LCB_{\neg \NCA},
\textrm{Regret}^\HCB_{\neg \NCA} \in [0, 1]$ are the weights (relative
to the standard utility) given to each counterfactual utility.  Here,
a greater regret implies a larger influence of counterfactual outcome
under an alternative decision and its associated counterfactual
outcome.  For example, if a judge could have released an arrestee and
achieved the best outcome (no \NCA), then the value of regret, i.e.,
$\texttt{Regret}^\ROR_{\neg \NCA}$, may be greater.

Once the utilities are fully specified, we can statistically evaluate
a different decision-making system by estimating its expected
utilities, which is the average utility across all arrestees in a
target population.  The expected utilities enable us to compare the
empirical performance of different decision-making systems. In our
application, such a system includes one in which a judge makes the
decision without help of PSA and the other in which a judge is
provided with PSA.  Beyond statistical evaluation, we can also learn
an optimal decision-making rule from the observed data by finding a
decision rule that maximizes the expected utility.  In the next two
sections, we will develop these methods.  In
Section~\ref{sec:analysis}, we will revisit this RCT and apply our methodology to evaluate the judge's decision with or without PSA. 
We will also derive an optimal decision rule under the triage score utility framework.

\section{Triage Score}
\label{sec:triage}

In this section, we formally develop the triage score framework.  Unlike risk
scores, this new counterfactual risk assessment instrument leverages
the full set of potential outcomes. To develop intuition, we begin by
introducing the simplest case with binary decisions and outcomes.  We
then generalize our formulation to handle multi-valued decisions and
outcomes.

\subsection{Binary Case}
\label{sec:binary}

We first consider settings with a binary decision $D \in \{0, 1\}$
and a binary outcome $Y \in \{0, 1\}$. In the context similar to that
of our application, $D = 0$ represents a decision to release an
arrestee on their own recognizance, while $D = 1$ represents a
decision to impose cash bail. The outcome $Y = 1$ indicates an
undesirable event, such as a rearrest for new criminal
activity. Lastly, let $Y(d)$ denote the potential outcome under
decision $D = d$, for $d \in \{0,1\}$.  For example, $Y(1)$ represents
the outcome that would be realized if the decision were $D = 1$. The
observed and potential outcomes are linked through the relation
$Y = D Y(1) + (1 - D) Y(0)$, which reflects the standard consistency
assumption \citep{rubi:90}. Relaxing this assumption, for instance by allowing for spillover effects, is beyond the scope of this paper.

We now formalize the proposed counterfactual risk assessment
instrument, which we call the ``triage score,'' within the framework
of statistical decision theory based on counterfactual utilities
\citep{koch2025statistical}. Our approach generalizes existing
counterfactual risk assessment instruments
\citep[e.g.,][]{coston2020counterfactual,ben2024does} by considering
the joint potential outcomes $(Y(0), Y(1))$ rather than focusing
solely on the baseline potential outcome $Y(0)$.

In the case of binary decisions and outcomes, the joint potential
outcomes define four principal strata, corresponding to all possible
combinations of $(Y(0), Y(1)) = (y_0, y_1)$ \citep{fran:rubi:02}. For
the purpose of exposition, we consider the context of our application
and refer to these four strata as follows:
\begin{itemize}
\item {\it Safe} \( (Y(0), Y(1)) = (0, 0) \): a defendant who would
  not be rearrested for a new crime under either decision
\item {\it Backlash} \( (Y(0), Y(1)) = (0, 1) \): a defendant who
  would be rearrested only if cash bail is imposed
\item {\it Preventable} \( (Y(0), Y(1)) = (1, 0) \): a defendant who
  would be rearrested only if released on their own recognizance
\item {\it Hopeless} \( (Y(0), Y(1)) = (1, 1) \): a defendant who would
  be rearrested regardless of the decision
\end{itemize}

To construct the triage score, we assign a utility to each possible
decision within every principal stratum, resulting in a total of eight
utility parameters to specify. However, in this fully general
formulation, the corresponding expected utility is not identifiable
because we do not observe two potential outcomes at the same time. To
address this, we impose an additivity assumption that restricts
counterfactual utilities to be additive in the potential outcomes
without interaction terms. \cite{koch2025statistical} show that this
additivity condition is both necessary and sufficient for the
identification of expected counterfactual utilities under the standard
assumption of unconfoundedness.

\begin{table}[t!]
\centering 
    \begin{tabular}{l|c|c|c|}
  \multicolumn{2}{c}{} & \multicolumn{2}{c}{\textbf{Decision}} \\\cline{3-4}
  \multicolumn{2}{c|}{} 
    & \multirow{2}{*}{Release $(D^\ast = 0)$} 
    & \multirow{2}{*}{Cash bail $(D^\ast = 1)$} \\
  \multicolumn{2}{c|}{} & & \\ \cline{2-4}

  & \multirow{2}{*}{\shortstack{Safe\\$(Y(0) = 0,\ Y(1) = 0)$}} 
    & \multirow{2}{*}{$u_{\nocrime}^{\ROR} + \tilde{u}_{\nocrime}^{\cash}$} 
    & \multirow{2}{*}{$u_{\nocrime}^{\cash} + \tilde{u}_{\nocrime}^{\ROR}$} \\
  & & & \\ \cline{2-4}

  \multirow{3}{*}{\textbf{Principal}}
  & \multirow{2}{*}{\shortstack{Backlash\\$(Y(0) = 0,\ Y(1) = 1)$}} 
    & \multirow{2}{*}{$u_{\nocrime}^{\ROR} + \tilde{u}_{\crime}^{\cash}$} 
    & \multirow{2}{*}{$u_{\crime}^{\cash} + \tilde{u}_{\nocrime}^{\ROR}$} \\
  & & & \\ \cline{2-4}

  \textbf{Strata}
  & \multirow{2}{*}{\shortstack{Preventable\\$(Y(0) = 1,\ Y(1) = 0)$}} 
    & \multirow{2}{*}{$u_{\crime}^{\ROR} + \tilde{u}_{\nocrime}^{\cash}$} 
    & \multirow{2}{*}{$u_{\nocrime}^{\cash} + \tilde{u}_{\crime}^{\ROR}$} \\
  & & & \\ \cline{2-4}

    & \multirow{2}{*}{\shortstack{Hopeless\\$(Y(0) = 1,\ Y(1) = 1)$}} 
    & \multirow{2}{*}{$u_{\crime}^{\ROR} + \tilde{u}_{\crime}^{\cash}$} 
    & \multirow{2}{*}{$u_{\crime}^{\cash} + \tilde{u}_{\crime}^{\ROR}$} \\
  & & & \\ \cline{2-4}
\end{tabular}
\caption{Additive counterfactual utilities in the case of binary
  decision and outcome. For each principal stratum, we specify additive
  utilities for the realized outcome under a given decision and the
  counterfactual outcome under the alternative decision. For example,
  in the safe stratum, the utility under the release decision is the
  sum of utility for the realized outcome (no crime)
  $u_{\nocrime}^{\ROR}$ and the utility for the counterfactual outcome
  (no crime) under the alternative decision (cash bail)
  $\tilde{u}_{\nocrime}^{\cash}$.}
\label{tbl:add_utility}
\end{table}

Table~\ref{tbl:add_utility} presents the additive counterfactual
utilities for the case of binary decisions and binary outcomes. These
utilities consist of two
components---$u_{y}^{d}$ and
$\tilde{u}_{y}^{d^\prime}$, which represent the standard and
counterfactual utilities, respectively, for $d, d^\prime \in \{\ROR,
\cash\}$, $d \ne d^\prime$, and $y \in \{\nocrime,
\crime\}$ where $\crime$ denotes a rearrest and
$\nocrime$ indicates no rearrest. We use
$D^\ast$ to denote a generic decision, to distinguish it from the
observed decision $D$ in the data. The term
$u_{y}^{d}$ represents the standard utility for decision $D^\ast =
d$ when the corresponding outcome is $Y(d) =
y$. In contrast,
$\tilde{u}_{y}^{d^\prime}$ represents the counterfactual utility that
would result under the same decision $D^\ast =
d$, but assuming the counterfactual outcome $Y(d^\prime) =
y$ would have occurred under the alternative decision $D^\ast =
d^\prime$ with $d \ne
d^\prime$. The additivity assumption implies that these two utility
components do not interact. A formal definition of additive
counterfactual utilities is presented in the next subsection, where we
further generalize this framework to accommodate non-binary decisions
and outcomes \citep[see also][for further
details]{koch2025statistical}.

As mentioned earlier, this counterfactual utility component can be
interpreted as capturing the notion of ``regret'' in decision making
\citep[e.g.,][]{bell1982regret,loomes1982regret}. Consider, for
example, the \emph{Safe} and \emph{Preventable} cases when a judge
imposes cash bail, $D^\ast = 1$. In both cases, the observed outcome
under cash bail is no rearrest, $Y(1) = 0$. However, in the
\emph{Safe} case, the counterfactual outcome under release would also
have been no rearrest $Y(0)=0$, whereas in the \emph{Preventable}
case, release would have resulted in a rearrest $Y(0)=1$. Standard
statistical decision theory does not distinguish between these two
cases, since the observed outcome is the same under the chosen
decision. In contrast, the proposed counterfactual risk assessment
framework may assign a lower utility to the \emph{Safe} case than to
the \emph{Preventable} case, because the counterfactual outcomes
differ. In this sense, imposing cash bail in the \emph{Safe} case
represents an unnecessarily harsh decision to avoid an outcome that
would not have occurred anyway.

Based on the above additive counterfactual utility formulation, we can
formally define the expected utility of any decision rule $D^\ast$ by
marginalizing over the joint distribution of principal strata and
decisions:
\begin{equation*}
    \overline{U}(u; D^\ast)
    := 
    \sum_{d=0}^{1} \sum_{y_{d}=0}^{1} \sum_{y_{1-d}=0}^{1}
    \E[(u_{y_{d}}^{d} + \tilde{u}_{y_{1-d}}^{1-d})
    \Pr(D^\ast = d, Y(d) = y_{d}, Y(1-d) = y_{1-d} \mid \bX)] 
\end{equation*}
where
$\bu = \{u_{y}^{d}, \tilde{u}_{y}^{1-d}\}_{y \in \{0,1\}, d \in
  \{0,1\}}$ denotes the utility parameters that define $u$, $\bX \in \cX$
represents pre-treatment covariates whose support is $\cX$,
and the expectation is taken over the distribution of $\bX$.

Unfortunately, even under the additivity assumption, the joint
distribution of principal strata and decisions is not identified due
to the {\it selective labels problem}: for the cases in which the
decision maker issues a cash bail decision, we do not observe the
counterfactual outcome under a release decision, and vice versa.
Nevertheless, we show that the expected utility can still be
identified---without the knowledge of the joint distribution---under
the unconfoundedness assumption with additive counterfactual
utilities.

Once the expected utility is identified, we can evaluate any decision
making system using the utility parameters $\bu$ specified by the
researcher.  As detailed in Section~\ref{sec:statistical}, our
empirical evaluation is based on a doubly robust estimator of the
expected utility with a rich set of pretreatment covariates $\bX$,
including prior criminal history, demographic information, and
PC affidavits.  Furthermore, we can learn the optimal
decision rule that maximizes the expected utility within this
framework. Under additive utilities, this can be done without
identifying the full conditional distribution of principal strata.
Specifically, we estimate the decision rule that maximizes the
expected utility under the decision it recommends within a specified
class of decision rules.

Before presenting the proposed statistical evaluation and learning
methodology, we briefly explain how this triage score differs from the
existing counterfactual risk assessment instruments.  We also
generalize the binary case presented above to the general categorical
case.

\subsection{Comparison with the Existing Counterfactual Risk
  Assessment Framework}
\label{sec:comparison}

\begin{table}[t!]
\centering 
    \begin{tabular}{l|c|c|c|c|}
  \multicolumn{3}{c}{} & \multicolumn{2}{c}{\textbf{Decision}} \\\cline{4-5}
  
  \multicolumn{3}{c|}{} 
    & \multirow{2}{*}{Release $(D^\ast = 0)$} 
    & \multirow{2}{*}{Cash bail $(D^\ast = 1)$} \\
  \multicolumn{3}{c|}{} & & \\ \cline{2-5}

  & \multirow{4}{*}{$Y(0) = 0$} 
    & \multirow{2}{*}{\shortstack{Safe\\$(Y(0) = 0,\ Y(1) = 0)$}} 
    & \multirow{4}{*}{\shortstack[l]{$u_{\ROR, \nocrime}$ 
    \\ $=  u_{\nocrime}^{\ROR} + \tilde{u}_{\nocrime}^{\cash}$ 
    \\ $=  u_{\nocrime}^{\ROR} + \tilde{u}_{\crime}^{\cash}$}} 
    & \multirow{4}{*}{\shortstack[l]{$u_{\cash, \nocrime}$ 
    \\ $= u_{\nocrime}^{\cash} + \tilde{u}_{\nocrime}^{\ROR}$ 
    \\ $= u_{\crime}^{\cash} + \tilde{u}_{\nocrime}^{\ROR}$}} \\
  & & & & \\ \cline{3-3}

  \multirow{3}{*}{\textbf{Baseline}}
  && \multirow{2}{*}{\shortstack{Backlash\\$(Y(0) = 0,\ Y(1) = 1)$}} 
    &
    & \\
  & & & & \\ \cline{2-5}

  \textbf{Outcome}
  & \multirow{4}{*}{$Y(0) = 1$} 
  & \multirow{2}{*}{\shortstack{Preventable\\$(Y(0) = 1,\ Y(1) = 0)$}} 
    & \multirow{4}{*}{\shortstack[l]{$u_{\ROR, \crime}$ 
    \\ $= u_{\crime}^{\ROR} + \tilde{u}_{\nocrime}^{\cash}$ 
    \\ $= u_{\crime}^{\ROR} + \tilde{u}_{\crime}^{\cash}$}} 
    & \multirow{4}{*}{\shortstack[l]{$u_{\cash, \crime}$ 
    \\ $= u_{\nocrime}^{\cash} + \tilde{u}_{\crime}^{\ROR}$
    \\ $= u_{\crime}^{\cash} + \tilde{u}_{\crime}^{\ROR}$}} \\
  & & & & \\ \cline{3-3}

  && \multirow{2}{*}{\shortstack{Hopeless\\$(Y(0) = 1,\ Y(1) = 1)$}} 
    & 
    & \\
  & & & & \\ \cline{2-5}
\end{tabular}
\caption{The existing counterfactual risk assessment framework based
  on the baseline potential outcome alone.  When compared to the
  proposed framework shown in Table~\ref{tbl:add_utility}, this
  framework assumes the equality of utilities between (1) the {\it
    Safe} and {\it Backlash} strata, and (2) the {\it Preventable} and
  {\it Hopeless} strata, i.e.,
  $\tilde{u}_{\nocrime}^{\cash} = \tilde{u}_{\crime}^{\cash}$ and
  $u_{\nocrime}^{\cash} = u_{\crime}^{\cash}$.}
\label{tbl:add_utility_special}
\end{table}

As mentioned earlier, the proposed formulation generalizes the
existing risk score framework by considering the joint potential
outcomes rather than the baseline potential outcome alone
\citep{coston2020counterfactual,ben2024does}.  The existing risk score
framework imposes additional constraints, implying that the utilities
must be equal between the {\it Safe} and {\it Backlash} strata and
between the {\it Preventable} and {\it Hopeless} strata because each pair
shares the same baseline potential outcome value, i.e.,
$\tilde{u}_{\nocrime}^{\cash} =
\tilde{u}_{\crime}^{\cash}$ and $u_{\nocrime}^{\cash} =
u_{\crime}^{\cash}$.  In other words, the assumption is that both
counterfactual and standard utilities under cash bail decision remain
identical regardless of their corresponding potential outcomes.  Under
these conditions, we can marginalize over the potential outcomes
$Y(1)$ and obtain a simplified version of the confusion matrix as
shown in Table~\ref{tbl:add_utility_special}.  

By considering the joint potential outcomes, our framework can
distinguish between two scenarios; one in which a cash bail decision
would have prevented a rearrest ($\tilde{u}_{\nocrime}^{\cash}$), and the other
in which a rearrest would have happened regardless of decision
($\tilde{u}_{\crime}^{\cash}$).  In addition, the proposed framework can
differentiate the utility that would result by imposing a cash bail
when the outcome under the decision is a rearrest ($u_{\crime}^{\cash}$) from
the utility of the cash bail decision when a rearrest does not occur
($u_{\nocrime}^{\cash}$).  We note that different parameterization of an additive
counterfactual utility is possible and can lead to different
restrictions.  Nevertheless, this comparison underscores the
importance of incorporating joint potential outcomes when developing a
counterfactual risk assessment instrument.

\subsection{General Case}

We now extend the additive counterfactual utility framework to a more
general setting that allows for multi-valued (i.e., non-binary)
decisions and outcomes.  Consider a setting in which the observed
decision $D$ and generic decision $D^\ast$ take on $K_{D}$ categories,
$D, D^\ast \in \cD := \{0, 1, \ldots, K_{D}-1\}$, and the
outcome variable takes on $K_{Y}$ categories,
$Y \in \cY := \{0, 1, \ldots, K_{Y}-1\}$.

This generalization is important both in theory and practice because,
in the binary case, any additive counterfactual utility specification
can be equivalently represented by a standard utility formulation
(though the interpretation may differ). However, in the non-binary
case, the standard utility framework is not sufficiently expressive to
capture additive counterfactual utilities \citep{koch2025statistical}.
In our empirical application, one could consider a multi-valued decision with
$K_{D}=4$; \ROR{} ($D^\ast=0$), \LCB{} ($D^\ast=1$), \HCB{}
($D^\ast=2$), and \REM{} ($D^\ast=3$).  Similarly, if we distinguish
two types of rearrest, one for new criminal activity and the other for
new {\it violent} criminal activity (\NVCA), we have a multi-valued
outcome with $K_{Y}=3$; no rearrest ($Y=0$), \NCA{} but not \NVCA{}
($Y=1$), and \NVCA{} ($Y=2$).

We use $Y(d)$ to denote the potential outcome under generic decision
$D^\ast = d$, for $d \in \{0, 1, \ldots, K_{D}-1\}$, and the observed
outcome is equal to $Y=Y(D)$ where $D$ is the actual decision.  Under
this generalized setup, we can define the principal strata as the set
of all possible combinations of potential outcomes,
\(\bY = (Y(0), Y(1), \ldots, Y(K_{D}-1))\).  We write
\(\by = (y_{0}, y_{1}, \ldots, y_{K_{D}-1})\), where
$y_{d} \in \cY$ for all $d \in \cD$, to denote a realization of this
principal stratum.  For
example, $(Y(0), Y(1), \ldots, Y(K_{D}-1)) = (0, 0, \ldots, 0)$
represents a case in which, regardless of the decision, the potential
outcome is always no new criminal activity.

Depending on the restrictions imposed on the utility function, the
framework may use either the full principal strata defined by the
joint potential outcomes or a coarsening of those strata based on the
baseline potential outcome alone.  In this paper, we call the former
the triage score framework, whereas the latter is referred to as the risk score
framework.  Under the triage score framework, utilities are assigned
based on both the full vector of potential outcomes and the decision.

\begin{definition}[Triage and Risk Score Frameworks] \label{def:triage_risk_framework}
The triage score framework is a statistical decision-theoretic framework based on counterfactual utilities in which the utility of a decision may depend on the full vector of potential outcomes, 
$\bY = (Y(0), Y(1), \ldots, Y(K_{D} - 1))$, or equivalently on the principal stratum defined by this vector. The risk score framework is a special case, in which utilities depend on the principal stratum only through the baseline potential outcome $Y(d_0)$, where $d_0 \in \cD$ denotes a baseline decision (e.g., release). That is, for each fixed decision $d$, all principal strata sharing the same value of $Y(d_0)$ receive the same utility.

A triage score is a function that summarizes the conditional distribution of the joint potential outcomes given covariates,
\begin{equation*}
  s^{\textsc{triage}}: \cX \to \Delta(\cY^{K_{D}}),
\end{equation*}
where $s^{\textsc{triage}}_{\by}(\bx)$ denotes the predicted probability of principal stratum $\by \in \cY^{K_{D}}$ for units with covariates $\bx$. The utility function assigns a utility to each possible decision within every principal stratum,
\begin{equation*}
    u:\cD \times \cY^{K_{D}} \to \mathbb{R}.
\end{equation*}
Under the risk score framework, this utility function is restricted so that $u(d,\by)=u(d,\by^\prime)$ whenever $y_{d_0}=y_{d_0}^\prime$.
\end{definition}

A construction of triage scores requires a
total of $K_{Y}^{K_{D}}\times K_{D}$ utility parameters as we need to specify a utility for each combination of principal stratum and decision.  To identify
the expected utility under the unconfoundedness assumption, which we
formally introduce in the next section, we again consider the additive
counterfactual utilities defined below.

\begin{definition}[Additive Counterfactual
  Utility] \label{def:add_utility} Additive counterfactual utility
  function is defined as
  $u:\cD \times \cY^{K_{D}} \to \mathbb{R}$ where for $d
  \in \cD$ and $\by \in \cY^{K_{D}}$, 
    \begin{equation*}
        u(d,\by)
        = u_{y_d}^{d}
        + \sum_{\substack{d^\prime=0 \\ d^\prime \ne d}}^{K_{D}-1}\tilde{u}_{y_{d^\prime}}^{d^\prime}.
    \end{equation*}
\end{definition}
In words, an additive counterfactual utility function
$u \in \mathcal{U}^{\textsc{Add}}$ assigns a utility to each decision
$d \in \cD$ within every principal stratum
$\by=(y_{0}, y_{1}, \ldots, y_{K_{D}-1})$, where $y_{d} \in \cY$, and this
utility consists of two components: standard utility $u_{y_d}^{d}$ and
counterfactual utility $\tilde{u}_{y_{d^\prime}}^{d^\prime}$. The term
$u_{y_d}^{d}$ is the portion of utility that is realized when the
decision $D^\ast = d$ is made and the corresponding observed outcome
is $Y = y_d$.  In contrast, $\tilde{u}_{y_{d^\prime}}^{d^\prime}$
represents the remaining portion of utility that is realized when the
decision $D^\ast = d$ is made but the counterfactual outcome is
$Y(d^\prime) = y_{d^\prime}$ under a given alternative decision
$d^\prime \in \cD$ and $d^\prime \ne d$.

This generalizes the additive counterfactual utility introduced in the
previous section for binary decision and binary outcome to settings
with multi-valued decisions and outcomes.  It is possible to further
generalize this by defining the additive counterfactual utilities
conditional on pre-treatment covariates $\bX = \bx$ as it is possible
to assign different utility values, depending on individual
characteristics, i.e., $u:\cD \times \cY^{K_{D}} \times \cX \to \mathbb{R}$.  
To simplify notation and focus on the core idea, we
do not condition on covariates in the utility function throughout this
paper.  Nevertheless, all the results presented below can be readily
extended to the cases in which utilities depend on the pre-treatment
covariates.

The additive counterfactual utility in
Definition~\ref{def:add_utility} is a special case of a more general
formulation in \cite{koch2025statistical}, in which we assume a zero
intercept for the utility function within each principal stratum. That
is, the utility for all possible interactions of potential outcomes
that do not depend on the realized decision is assumed to be zero.  By
applying the same proof as in Corollary 2 of
\cite{koch2025statistical}, it can be shown that the additivity
assumption imposed on the counterfactual utility in
Definition~\ref{def:add_utility} is both necessary and sufficient for
the point identification of the expected utility under the
unconfoundedness assumption.  For completeness, we provide the
definition of expected additive counterfactual utility,
\begin{equation*}
    \overline{U}(u; D^\ast) = \E[U(u; D^\ast\mid \bX)],
\end{equation*}
where the conditional expected utility of a decision $D^\ast$ given
the pre-treatment covariates $\bX$ is defined as,
\begin{align*}
    & U(u; D^\ast\mid \bX) \\
    := \ & 
    \E\left[
    \sum_{d=0}^{K_{D}-1} \sum_{y_{0}=0}^{K_{Y}-1} \ldots \sum_{y_{K_{D}-1}=0}^{K_{Y}-1}
    u(d,\by) \bbone \{D^\ast = d, Y(0) = y_{0}, \ldots Y(K_{D}-1) = y_{K_{D}-1}\}
    \ \Bigl | \ \bX
    \right] \\
    = \ &  
    \sum_{d=0}^{K_{D}-1} \sum_{y_{0}=0}^{K_{Y}-1} \ldots \sum_{y_{K_{D}-1}=0}^{K_{Y}-1}
    u(d,\by) \Pr(D^\ast = d, Y(0) = y_{0}, \ldots Y(K_{D}-1) = y_{K_{D}-1} \mid \bX).
\end{align*}

Under this additive utility structure, we can define the expected utility of a decision rule $D^\ast$
under the triage score framework as $\overline{U}(u; D^\ast)$ where $u(d,\by)
    =
    u_{y_d}^{d} + \sum_{d^\prime \ne d} \tilde{u}_{y_{d^\prime}}^{d^\prime}$ is allowed to vary across $\by,\by^\prime$
    even when $y_{d_0}=y_{d_0}^\prime$ with $\by=(y_0,\ldots,y_{K_{D}-1})$ and $d_0$ denoting the baseline decision.  In contrast, under the risk score framework, we have the restriction
    $u_{y_d}^{d} + \sum_{d^\prime \ne d} \tilde{u}_{y_{d^\prime}}^{d^\prime}
    =
    u_{y_d^\prime}^{d} + \sum_{d^\prime \ne d} \tilde{u}_{y_{d^\prime}^\prime}^{d^\prime}$ for all
    $d \in \cD$ and all  
    $\by,\by^\prime$ with
    $y_{d_0}=y_{d_0}^\prime$.

Accordingly, an optimal decision rule under the triage score framework is defined as
\begin{equation*}
    \pi^{\textsc{triage}} 
    \in 
    \argmax_{\pi \in \Pi} \overline{U}(u; D^\pi),
\end{equation*}
where $D^\pi=\pi(\bX)$ denotes the decision induced by a policy $\pi$ in the policy class $\Pi$. When the policy class is unrestricted so that decisions can be optimized pointwise in $\bx$, the corresponding optimal decision rule under the triage score framework is
\begin{equation*}
    \pi^{\textsc{triage}}(\bx)
    \in
    \argmax_{d \in \cD}
    \sum_{\by \in \cY^{K_{D}}}
    u(d,\by) \,
    s^{\textsc{triage}}_{\by}(\bx),
\end{equation*}
where $s^{\textsc{triage}}_{\by}(\bx)$ denotes the predicted probability of principal stratum $\by$ given covariates $\bx$. In the additive utility structure, this can be written as
\begin{equation*}
    \pi^{\textsc{triage}}(\bx)
    \in
    \argmax_{d \in \cD}
    \sum_{\by \in \cY^{K_{D}}}
    \left(
        u_{y_d}^{d} 
        + 
        \sum_{\substack{d^\prime=0 \\ d^\prime \ne d}}^{K_{D}-1}
        \tilde{u}_{y_{d^\prime}}^{d^\prime}
    \right)
    s^{\textsc{triage}}_{\by}(\bx).
\end{equation*}

\section{Statistical Evaluation and Learning}
\label{sec:statistical}

Given the additive counterfactual utility introduced above, we can now
statistically evaluate the expected utility of any decision rule
$D^\ast$, including the ones that have not been used in practice, and
learn the optimal decision rule under the triage score utility
framework.  In addition, we can
compare the expected utility of a given decision across different
additive utility specifications, allowing a decision-making system to
be evaluated under various sets of utility parameters.  In this
section, we develop such statistical evaluation and learning
methodology.

\subsection{Assumptions and Nonparametric Identification}

We first present the nonparametric identification result for the
expected utility of a decision under the counterfactual risk
assessment system described above. Consistent with the motivating
application, we consider a setting in which a human decision maker
receives an algorithmic recommendation for a randomly selected subset
of cases. For simplicity, we refer to these as ``AI recommendations,''
though our methodology is agnostic to how such recommendations are
generated. We also emphasize that the proposed methodology remains
valid whether or not such recommendations are provided.

As in the previous section, we use $D_{i} \in \cD$ to denote
the observed decision for case $i$.  Let $Z_{i} \in \{0, 1\}$ be an
indicator for the {\it provision} of AI recommendation, and
$R_{i} \in \cR := \{0, 1, \ldots, K_{R}-1\}$ be the {\it content} of
the recommendation.  We let $a:\cR \to \cD$ map each recommendation
category to the decision option it recommends.  In our
application, $Z_{i} = 1$ indicates that the judge has access to the
PSA, while $Z_{i} = 0$ means that the judge does not.

The potential decision, denoted by $D_{i}(z)$, represents the decision
the judge would make under the treatment assignment $Z_{i} = z$.
Thus, the observed decision is given by $D_{i} = D_{i}(Z_{i})$ under
the standard consistency assumption.  For simplicity, we assume that
the provision of AI recommendation to other cases does not influence
the judge's decision in the current case (see
\citealp{jiang2024longitudinal} who relaxes this assumption).  To
further increase the credibility of this assumption, we focus on the
first arrest cases and drop rearrest cases \citep[see][for a
justification of this approach]{imai2023experimental}.

Under this setup, our goal is to estimate the expected utilities of
three decision-making systems: the human-alone $D(0)$, the
human-with-AI $D(1)$, and the AI-alone $R$ decisions.  We follow
\cite{ben2024does} and assume a single-blinded treatment assignment
design, which is satisfied in our application.
\begin{assumption}[Single-blinded and unconfounded treatment
  assignment] \label{assum:single_blinded} 
  The
  treatment assignment $Z_i$, potential decisions $D_i(z)$,
  pre-treatment covariates $\bX_i$, and potential outcomes
  $Y_i(z, D_i(z))$ satisfy:
  \begin{enumerate}[label=(\alph*)]
  \item Single-blinded treatment assignment: $Y_i(z,D_i(z)) = Y_i(z',D_i(z'))$ for all $z,z'$ such that $D_i(z) = D_i(z')$
  \item Unconfounded treatment assignment:
    $Z_i \ \indep \ \{R_i, \{D_i(z), Y_i(d)\}_{z\in \{0,1\}, d\in\cD} \mid
    \bX_i$
  \item Overlap: 
  There exists a constant $\lambda_Z \in (0, \frac{1}{2})$ satisfying
  $\lambda_Z \le e(\bX_i):=\Pr(Z_i=1\mid \bX_i) \le 1-\lambda_Z$ almost surely.
  \end{enumerate}
\end{assumption}

Assumption~\ref{assum:single_blinded}(a) implies
$Y_{i}(0, d) = Y_{i}(1, d) = Y_{i}(d)$ for all $d \in \cD$.  In our
application, this means that an arrestee does not know whether the
judge has received the PSA, which is an especially plausible assumption in
Utah because there was no live hearing, so the arrestee had no way of
knowing what the judge considered.  In other words, we assume that the
provision or absence of an PSA can affect the outcome only through the
judge's decision. This assumption would be violated if the judge were
to inform the arrestee about the PSA, thereby directly influencing the
arrestee’s behavior independent of the judge’s decision. Here,
however, given the absence of a live hearing, such violation is
exceedingly unlikely.

In our study, Assumption~\ref{assum:single_blinded}(b) is satisfied by
the experimental design, in which the provision of the PSA to the
judge is randomized.  To allow for an extension to more general
settings, we present the unconfounded treatment assignment conditional
on the observed pre-treatment covariates $\bX_i$.  Lastly, through
Assumption~\ref{assum:single_blinded}(c), which also holds in our
application by design, we impose the overlap condition, assuming that
the treatment probability $e(\bx)$ is bounded away from both zero and
one.

In addition to the single-blinded treatment assignment, we assume the
unconfoundedness of the decision, implying that the potential outcomes
are independent of the decision, conditional on the observed
covariates, treatment assignment, and the AI recommendation.
\begin{assumption}[Unconfoundedness of decision]
    \label{assum:indep}
    \begin{equation*}
        \{Y_i(d)\}_{d\in \cD} \ \indep \ D_i \mid \bX_i, Z_{i}, R_{i} 
    \end{equation*}
\end{assumption}

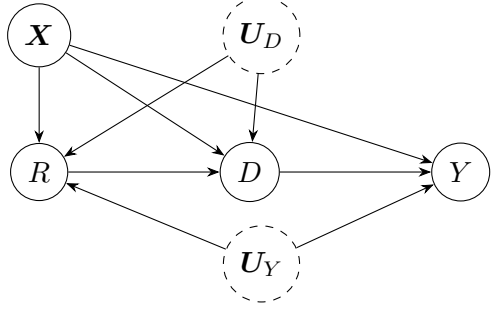
\begin{figure}[t!]
    \centering
        \begin{tikzpicture}[node distance=2cm, >={Stealth}]
        \node[circle, draw] (ai) at (0, 0) {$R$};
        \node[circle, draw] (x) [above=1cm of ai] {$\bX$};
        \node[circle, draw, dashed] (u_d) [right=2cm of x] {$\bU_D$};
        \node[circle, draw, dashed] (u_y) [below=2cm of u_d] {$\bU_Y$};
        \node[circle, draw] (decision) [right=2cm of ai] {$D$};
        \node[circle, draw] (outcome) [right=2cm of decision] {$Y$};

        \draw[->] (ai) -- (decision);
        \draw[->] (decision) -- (outcome);
        \draw[->] (x) -- (ai);
        \draw[->] (x) -- (decision);
        \draw[->] (x) -- (outcome);
        \draw[->] (u_d) -- (ai);
        \draw[->] (u_d) -- (decision);
        \draw[->] (u_y) -- (ai);
        \draw[->] (u_y) -- (outcome);

        \end{tikzpicture}
    \caption{An example of a causal diagram under which Assumptions~\ref{assum:single_blinded} and \ref{assum:indep} hold. When $Z = 0$, there is no edge $R \rightarrow D$.}
    \label{fig:dag}
\end{figure}

Assumption~\ref{assum:indep} implies the absence of confounders that
affect decision $D_i$ and the outcome $Y_i$.  However, as illustrated
in the directed acyclic graph (DAG) of Figure~\ref{fig:dag}, the
assumption allows for the presence of $\bU_D$, which confounds the
relationship between the recommendation and the decision, and $\bU_Y$,
which confounds the relationship between the recommendation and the
outcome.  Appendix~\ref{app:assum_indep} further discusses
Assumption~\ref{assum:indep} and present an alternative, stronger
assumption that is more aligned with the data-generating process in
our application.  In Appendix~\ref{app:assum_indep_validity}, we
explain that Assumption~\ref{assum:indep} is credible in our
application because we observe all the information a judge has when
making the cash bail decision.

\begin{assumption}[Decision positivity on
  support] \label{assum:decision_pos}
For each $(z,r)$ such that $\Pr(Z_i=z,R_i=r)>0$, there exists a constant
$\lambda_D \in (0,1)$ satisfying
$\min_{d\in\cD} m_d^{D}(z,r,\bX_i) \ge \lambda_D$ almost surely within
the $(Z_i=z,R_i=r)$ stratum.
\end{assumption}

Assumption~\ref{assum:decision_pos} is a positivity condition for the
decision model within each $(Z_i=z,R_i=r)$ stratum. It ensures that the conditional outcome probabilities
$\Pr(Y_i=y\mid D_i=d,R_i=r,Z_i=z,\bX_i=\bx)$ that appear in the identification formulas
are identified on the relevant covariate support: each decision
level $d\in\cD$ must occur with positive probability at covariate values that can
arise under $(Z=z,R=r)$. 

We now present the main identification result, which shows that the
expected utility of the counterfactual risk assessment system can be
identified under the stated assumptions.

\begin{theorem}[Identification of the expected additive counterfactual
  utility]
\label{thm:risk_additive}
Consider an additive counterfactual utility
$u \in \mathcal{U}^{\textsc{Add}}$ and a decision rule $D^\ast$ to be evaluated.  
Suppose that $D^\ast$ is either measurable with respect to $(R,\bX)$ (e.g., AI decision $a(R)$)
or satisfies $\{Y(k)\}_{k\in \cD} \indep D^\ast \mid \bX, Z, R$ (e.g., observed decision $D$).
Under
Assumptions~\ref{assum:single_blinded}, \ref{assum:indep}, and \ref{assum:decision_pos}, we can
identify the expected utility of the counterfactual risk assessment
system under the decision $D^\ast$ as follows: 
    \begin{align*}
    &\overline{U}(u; D^\ast)\\
    &=\E\Bigg[\sum_{y=0}^{K_{Y}-1} \sum_{d=0}^{K_{D}-1} 
    \sum_{r=0}^{K_{R}-1}  
    u_{y}^{d} 
    \Pr(Y = y \mid D = d, R = r, \bX = \bx)
    \Pr(D^\ast = d, R = r \mid \bX = \bx)
    \\
    &+
    \sum_{y=0}^{K_{Y}-1}
    \sum_{d=0}^{K_{D}-1} 
    \sum_{\substack{d^\prime=0 \\ d^\prime \ne d}}^{K_{D}-1}
    \sum_{r=0}^{K_{R}-1}
    \tilde{u}_{y}^{d^\prime}
    \Pr(Y = y \mid D = d^\prime, R = r, \bX = \bx) 
    \Pr(D^\ast = d, R = r \mid \bX = \bx)
    \Bigg]
    \end{align*}
\end{theorem}

The proof of Theorem~\ref{thm:risk_additive} is given in the
Appendix~\ref{proof:risk_additive}.  Theorem~\ref{thm:risk_additive}
shows that the expected utility of the counterfactual risk assessment
system can be point-identified under the stated assumptions.  This
identification strategy mirrors Corollary~2 of
\cite{koch2025statistical}, which shows that the additivity condition
is both necessary and sufficient for identifying expected
counterfactual utilities under the standard unconfoundedness
assumption (see Appendix~\ref{app:assum_indep} for a discussion of how
we adapt this assumption to our setting).  In certain cases,
Assumption~\ref{assum:indep} is not required for the identification of
the difference in expected utility between human decisions made with
and without recommendations, i.e., $D(1)$ and $D(0)$.  Specifically,
Remark~\ref{remark:utility_difference} in the Appendix shows that this
is possible either when the decision is binary or when a further
restriction is placed on additive utilities in the case of
multi-valued decisions.  This generalizes the findings of
\cite{ben2024does}.  In the following, we show how to estimate the
expected utility and construct the optimal decision rule based on this
identification result.

\subsection{Evaluating Decisions}

We now discuss the estimation strategy using an augmented
inverse propensity weighting (AIPW) estimator.  Throughout the rest of
the section, we assume that Assumptions~\ref{assum:single_blinded} and
\ref{assum:indep} hold in our study. We begin by defining two nuisance
components:
\begin{align*}
    &\text{(i) Decision model } m_{d}^{D}(z,r,\bx):= \Pr(D = d \mid Z = z, R = r, \bX = \bx), \\
    &\text{(ii) Outcome model } m_{y}^{Y}(z,r,d,\bx):= \Pr(Y = y \mid Z = z, R = r, D = d, \bX = \bx).
\end{align*}
We also  define the propensity score under the treatment assignment $z$ as:
\begin{equation*}
    e(z,\bx) := ze(\bx) + (1-z)(1-e(\bx))
\end{equation*}
where $e(\bx) := \Pr(Z = 1 \mid \bX = \bx)$. 

We evaluate the expected utility of a generic decision $D^\ast$.
By Theorem~\ref{thm:risk_additive}, the expected utility under a generic
decision rule $D^\ast$ can be written as
\begin{align*}
    &\overline{U}(u; D^\ast)\\
    &=\E\Bigg[\sum_{y=0}^{K_{Y}-1} \sum_{d=0}^{K_{D}-1} 
    \sum_{r=0}^{K_{R}-1}  
    \sum_{z=0}^{1}
    u_{y}^{d} 
    \Pr(Y = y \mid D = d, R = r, Z = z, \bX = \bx)
    \Pr(D^\ast = d, R = r, Z = z\mid \bX = \bx)
    \\
    &+
    \sum_{y=0}^{K_{Y}-1}
    \sum_{d=0}^{K_{D}-1} 
    \sum_{\substack{d^\prime=0 \\ d^\prime \ne d}}^{K_{D}-1}
    \sum_{r=0}^{K_{R}-1}
    \tilde{u}_{y}^{d^\prime}
    \Pr(Y = y \mid D = d^\prime, R = r, Z = z, \bX = \bx) 
    \Pr(D^\ast = d, R = r, Z = z \mid \bX = \bx)
    \Bigg]
\end{align*}
where the summation over $Z$ follows from Assumption~\ref{assum:single_blinded}(b) and consistency.

We propose an AIPW estimator using a two-decision-index
uncentered influence-function term.  The first decision index
identifies the potential outcome to be learned from the observed
human decision, while the second decision index identifies the
decision made by the evaluated rule $D^\ast$.  This distinction is
needed for the counterfactual utility terms, which involve
$\Pr(Y = y\mid D = d^\prime, R = r, Z = z, \bX = \bx)$ multiplied by
$\Pr(D^\ast = d, R = r, Z = z\mid \bX = \bx)$.
In this section, we assume that $D^\ast = f(r,\bx)$ is a deterministic decision rule that is known to the researcher.
The proposed AIPW estimator can be easily extended to accommodate a stochastic decision rule
by introducing an additional nuisance model
$m^{D^\ast}_{d}(r,\bx)
:= \Pr(D^\ast = d \mid R = r, \bX = \bx)$, which may need to be
estimated if it is unknown to the researcher. 

Let $O_i := (Y_i,D_i,R_i,Z_i,\bX_i)$.

\begin{equation*}
    \widehat{\overline{U}(u; D^\ast)} = 
    \frac{1}{n} \sum_{i=1}^{n}\sum_{y=0}^{K_{Y}-1}
    \sum_{d=0}^{K_{D}-1}
    \sum_{r=0}^{K_{R}-1}
    \sum_{z=0}^{1}
    \Bigg(
    u_{y}^{d} \widehat{\eta}_{yddrz}(O_i)
    + \sum_{\substack{d^\prime=0 \\ d^\prime \ne d}}^{K_{D}-1}
    \tilde{u}_{y}^{d^\prime}
    \widehat{\eta}_{yd^\prime drz}(O_i)
    \Bigg),
  \end{equation*}
where
\begin{align*}
    &\widehat{\eta}_{ykdrz}(Y,D,R,Z,\bX)\\
    &= \bbone\{D^\ast = d, R = r\} \Bigg\{
      \hat{m}_y^{Y}(z,r,k,\bX)\hat{e}(z, \bX)+ \frac{\bbone\{Z = z, D = k\}}{\hat{m}_{k}^{D}(z,r,\bX)}
      (\bbone\{Y = y\} - \hat{m}_y^{Y}(z,r,k,\bX)) \\
      &\hspace{1.5in}+ \hat{m}_y^{Y}(z,r,k,\bX)(\bbone\{Z = z\} - \hat{e}(z,\bX))\Biggl\}.
  \end{align*}

  The following theorem establishes the asymptotic normality of this
  AIPW estimator under the rate conditions presented and discussed in
  Appendix~\ref{app:rate_conditions}.
  \begin{theorem}
  [Asymptotic normality of the AIPW estimator for decision]\label{thm:utility_generic}
  Suppose that $D^\ast=f(R,\bX)$ is a fixed, known deterministic rule, or that
  the stochastic-rule extension is used with a consistently estimated
  $m_{d}^{D^\ast}(r,\bx)$. Under
  Assumptions~\ref{assum:single_blinded}, \ref{assum:indep},
  \ref{assum:decision_pos}, and \ref{assum:rate_condition}, 
  we have
  \begin{equation*}
    \sqrt{n}(\widehat{\overline{U}(u; D^\ast)} - \overline{U}(u; D^\ast)) \xrightarrow{d} N(0, V)
  \end{equation*}
  where
  \begin{equation*}
    V = \E\Bigg[\Bigg\{\sum_{y=0}^{K_{Y}-1}
    \sum_{d=0}^{K_{D}-1}
    \sum_{r=0}^{K_{R}-1}
    \sum_{z=0}^{1}
    u_{y}^{d} \eta_{yddrz}(O_i)
    + \sum_{\substack{d^\prime=0 \\ d^\prime \ne d}}^{K_{D}-1}
    \tilde{u}_{y}^{d^\prime}
    \eta_{yd^\prime drz}(O_i)
    - \overline{U}(u; D^\ast)\Bigg\}^2\Bigg].
  \end{equation*}
\end{theorem}
The  proof of Theorem~\ref{thm:utility_generic} is given in Appendix~\ref{proof:utility_generic}.
Appendices~\ref{utility_human} and \ref{utility_ai} provide additional discussion of AIPW estimators for human decisions and AI decisions.

\subsection{Optimizing Decisions}

Based on this framework, we now discuss how to derive optimal decision
rules that maximize expected utility given a set of utility
parameters.  Specifically, consider an additive counterfactual utility
$u \in \mathcal{U}^{\textsc{Add}}$. Analogous to
Theorem~\ref{thm:risk_additive}, the expected utility under a
covariate-dependent policy $\pi:\cX\to\cD$ is identified as follows
under Assumptions~\ref{assum:single_blinded} and \ref{assum:indep}.

\begin{align*}
    &\overline{U}(u; \pi) \\
    &:= \E\Bigg[\sum_{d=0}^{K_{D}-1} 
    \sum_{y_{0}=0}^{K_{Y}-1} 
    \ldots 
    \sum_{y_{K_{D}-1}=0}^{K_{Y}-1}
    u(d,\by) 
    \Pr(\pi(\bX) = d, Y(0) = y_{0}, \ldots Y(K_{D}-1) = y_{K_{D}-1} \mid \bX)\Bigg] \\
    &= \E\Bigg[\sum_{d=0}^{K_{D}-1} \sum_{y=0}^{K_{Y}-1} \sum_{z=0}^{1}\sum_{r=0}^{K_{R}-1}
    u_{y}^{d} \bbone\{\pi(\bX) = d\} 
    \Pr(Y = y \mid D = d, R = r, Z = z, \bX = \bx)\Pr(R = r, Z = z\mid \bX = \bx)\\
    &
    + 
    \sum_{d=0}^{K_{D}-1} \sum_{y=0}^{K_{Y}-1}
    \sum_{\substack{d^\prime=0 \\ d^\prime \ne d}}^{K_{D}-1} \sum_{z=0}^{1}\sum_{r=0}^{K_{R}-1}
    \widetilde{u}_{y}^{d^\prime}
    \bbone\{\pi(\bX) = d\}
    \Pr(Y = y \mid D = d^\prime, R = r, Z = z, \bX = \bx)\Pr(R = r, Z = z\mid \bX = \bx)
    \Bigg].
\end{align*}
Thus, we can estimate this optimal policy by solving the following empirical utility maximization problem with AIPW estimator:
\begin{align*}
    \hat{\pi}_{\bu}
    \in \argmax_{\pi \in \Pi} 
    \frac{1}{n}
    \sum_{i=1}^{n}
    \sum_{d=0}^{K_{D}-1} \sum_{y=0}^{K_{Y}-1}
    \sum_{z=0}^{1}\sum_{r=0}^{K_{R}-1}
    \bbone\{\pi(\bX_{i}) = d\}
    &\Big(
    u_{y}^{d}
    \widehat{\widetilde{\psi}}_{yrdz}(Y_{i},D_{i},R_{i},Z_{i},\bX_{i})\\
    &\hspace{2em}
    + \sum_{\substack{d^\prime=0 \\ d^\prime \ne d}}^{K_{D}-1}
    \tilde{u}_{y}^{d^\prime}
    \widehat{\widetilde{\psi}}_{yrd^\prime z}(Y_{i},D_{i},R_{i},Z_{i},\bX_{i})
    \Big)
\end{align*}
where $\Pi$ is a class of policies specified by the researcher, and $\widehat{\widetilde{\psi}}_{yrdz}$ is defined in Appendix~\ref{utility_ai}.

For binary decisions ($\cD = \{0,1\}$), this optimization problem can be reformulated as a weighted classification problem.
Define the arrestee-level utility contribution under decision $d$ as
\begin{align*}
    \widetilde{U}_{i}(d) := \sum_{y=0}^{K_{Y}-1} \sum_{z=0}^{1} \sum_{r=0}^{K_{R}-1}
    \Big(
    u_{y}^{d} \widehat{\widetilde{\psi}}_{yrdz,i}
    + \sum_{\substack{d^\prime=0 \\ d^\prime \ne d}}^{K_{D}-1}
    \tilde{u}_{y}^{d^\prime} \widehat{\widetilde{\psi}}_{yrd^\prime z,i}
    \Big).
\end{align*}
The empirical utility maximization problem becomes
\begin{align*}
    \hat{\pi}_{\bu} &\in \argmax_{\pi \in \Pi} \frac{1}{n} \sum_{i=1}^{n} \Big[ \bbone\{\pi(\bX_{i}) = 1\} \widetilde{U}_{i}(1) + \bbone\{\pi(\bX_{i}) = 0\} \widetilde{U}_{i}(0) \Big] \\
    &= \argmax_{\pi \in \Pi} \frac{1}{n} \sum_{i=1}^{n} \Big[ \pi(\bX_{i}) U^{\ast}_{i} + \widetilde{U}_{i}(0) \Big],
\end{align*}
where $U^{\ast}_{i} := \widetilde{U}_{i}(1) - \widetilde{U}_{i}(0)$ represents the utility gain from assigning cash bail ($d=1$) versus release ($d=0$) for arrestee $i$.
Since $\widetilde{U}_{i}(0)$ does not depend on $\pi$, the optimal policy maximizes $\sum_{i} \pi(\bX_{i}) U^{\ast}_{i}$, which is equivalent to weighted binary classification where the class label is $\bbone\{U^{\ast}_{i} > 0\}$ and the arrestee weight is $|U^{\ast}_{i}|$.
This reformulation enables the use of standard classification algorithms to learn optimal policies within a given policy class $\Pi$.
For multiclass decisions ($K_{D} > 2$), the problem generalizes naturally: finding the optimal policy requires computing $\argmax_{d \in \cD} \widetilde{U}_{i}(d)$ for each case $i$, assigning each arrestee to the decision with highest estimated utility.

\section{Empirical Analysis}
\label{sec:analysis}

We now illustrate the use of triage score by applying the proposed
methodology to the empirical application presented in
Section~\ref{sec:empirical}.  We begin by explaining how we specify
the utilities for our analysis and then present the empirical results
obtained through the proposed methodology.

\subsection{Utility Specification}

As illustrated in Section~\ref{sec:binary}, additive counterfactual
utilities in the case of binary decision and outcome consist of eight
utility parameters, $u_{y}^{d}$ and $\tilde{u}_{y}^{1-d}$ for
$d \in \{0, 1\}$ and $y \in \{0, 1\}$, which represent the standard
and counterfactual utilities, respectively.  In our application study,
$D = 1$ ($D = 0$) represents a decision to impose cash bail (release
an arrestee on their own recognizance), and $Y = 1$ ($Y = 0$)
indicates an undesirable event (absence of an undesirable event), such
as a rearrest for new criminal activity.  For ease of interpretation,
we impose the following restrictions on the additive counterfactual
utility parameters:
\begin{itemize}
\item cost of decision: the cost of a cash bail decision is
  $c^\cash \ge 0$.  Without loss of generality, we absorb this cost
  into $u_{\crime}^\cash$ and $u_{\nocrime}^\cash$.
    \item cost of outcome: the cost of an undesirable event under ROR is $c_{\crime}^\ROR$ (standardized to $1$), while that under cash bail is $c_{\crime}^\cash \ge 0$.
    We further assume a zero baseline cost for the absence of an undesirable event.
    \item regret of counterfactual outcome: regret for the absence of an
      undesirable event under ROR is $r_{\nocrime}^\ROR \ge 0$, while
      the same regret under cash bail is $r_{\nocrime}^\cash \ge 0$. No regret is
      incurred for an undesirable event in the counterfactual outcome
\end{itemize}
Together, these assumptions yield the specifications, illustrated below using rearrest as the outcome:
\begin{align*}
\begin{alignedat}{3}
    &u_{\nocrime}^{\ROR}  = 0, &\qquad &u_{\crime}^{\ROR}  = -c_{\crime}^\ROR = -1, \\
    &u_{\nocrime}^{\cash}  = -c^\cash, &\qquad &u_{\crime}^{\cash}  = -c^\cash - c_{\crime}^\cash, \\
    &\tilde{u}_{\nocrime}^{\ROR}  = -r_{\nocrime}^\ROR, &\qquad &\tilde{u}_{\crime}^{\ROR}  = 0, \\
    &\tilde{u}_{\nocrime}^{\cash}  = -r_{\nocrime}^\cash, &\qquad &\tilde{u}_{\crime}^{\cash}  = 0.
\end{alignedat}
\end{align*}

Based on this utility specification,
Table~\ref{tab:decision_parameters} presents four decision-theoretic
frameworks and their corresponding free parameters along with
associated constraints, where the standardized parameter is
$c_{\crime}^\ROR = 1$.  For example, as discussed in
Section~\ref{sec:comparison},
$u_{\nocrime}^{\cash} = u_{\crime}^{\cash}$ and
$\tilde{u}_{\nocrime}^{\cash} = \tilde{u}_{\crime}^{\cash}$ should
hold in the risk score framework with an additive utility in binary
decision case.  Under our utility specification, this corresponds to
the constraints $c_{\crime}^\cash = 0$ and $r_{\nocrime}^\cash = 0$,
leaving three free parameters, $c^\cash$, $c_{\crime}^\ROR$ and
$r_{\nocrime}^\ROR$.

\begin{table}[t!]
  \centering
  \begin{tabular}[t]{llll}
    \toprule
    \textbf{Utility} & \textbf{Score} & \textbf{Constraint} & \textbf{Free Parameters} \\ 
    \midrule \vspace{0.5em}
    Standard & Risk & 
    $\begin{cases}
        u_{\nocrime}^{\cash} = u_{\crime}^{\cash}\\
        \tilde{u}_{\nocrime}^{\cash} = \tilde{u}_{\crime}^{\cash} = 0\\
        \tilde{u}_{\nocrime}^{\ROR} = \tilde{u}_{\crime}^{\ROR} = 0
    \end{cases}$
    & $c^\cash$, $c_{\crime}^\ROR$ \\ \vspace{0.5em}
    Standard & Triage & 
    $\begin{cases}
        \tilde{u}_{\nocrime}^{\cash} = \tilde{u}_{\crime}^{\cash} = 0\\
        \tilde{u}_{\nocrime}^{\ROR} = \tilde{u}_{\crime}^{\ROR} = 0
    \end{cases}$
    & $c^\cash$, $c_{\crime}^\ROR$, $c_{\crime}^\cash$ \\ \vspace{0.5em}
    Counterfactual & Risk & 
      $\begin{cases}
        u_{\nocrime}^{\cash} = u_{\crime}^{\cash}\\
        \tilde{u}_{\nocrime}^{\cash} = \tilde{u}_{\crime}^{\cash}
      \end{cases}$
    & $c^\cash$, $c_{\crime}^\ROR$, $r_{\nocrime}^\ROR$ \\ \vspace{0.5em}
    Counterfactual & Triage & \textit{None} & $c^\cash$, $c_{\crime}^\ROR$, $c_{\crime}^\cash$, $r_{\nocrime}^\ROR$, $r_{\nocrime}^\cash$ \\
    \bottomrule
  \end{tabular}
  \caption{Four decision-theoretic frameworks with corresponding free parameters. We standardize $c_{\crime}^\ROR = 1$.}
  \label{tab:decision_parameters}
\end{table}

For the purpose of our illustration, we further simplify the utility
specification by reducing the number of parameters.  First, we assume
that the cost of undesirable outcome is also no less under ROR than
cash bail, i.e., $c_{\crime}^\ROR\ge c_{\crime}^\cash$.  This
represents the possibility that a negative outcome followed by the ROR
decision appears worse than the same outcome under the cash bail
decision.  Second, we assume that regret is greater for ROR than for
cash bail, i.e., $r_{\nocrime}^\ROR \ge r_{\nocrime}^\cash$.  The idea
is that a judge experiences greater regret when no arrest would have
occurred under the alternative decision of ROR than when no arrest
would have occurred under the alternative decision of cash bail.  The
final assumption is that when determining costs and regrets, a judge
applies the same relative weights to ROR vs. cash bail decisions if
the outcome under consideration is the same, i.e.,
\begin{equation} \label{eq:ratio}
  \frac{c_{\crime}^\cash}{c_{\crime}^\ROR}\\
= \frac{r_{\nocrime}^\cash}{r_{\nocrime}^\ROR} \le 1
\end{equation}
With our standardization scheme $c_{\crime}^\ROR = 1$, this implies
$r_{\nocrime}^\cash = r_{\nocrime}^\ROR c_{\crime}^\cash$, reducing
the number of parameters by one.  All together, we have the utility
structure summarized in Table~\ref{tab:binary_empirical} of
Appendix~\ref{additive_loss}.

Of course, many other utility specifications are possible.  For
example, Appendix~\ref{alternative_loss} provides an alternative
specification of the utility function that uses three parameters---the
cost of cash bail, the cost of an undesirable outcome, and a discount
factor for counterfactual utility.  In practice, decision-makers
should determine utilities to reflect their own value system and
incorporate ethical and other constraints.

\subsection{Results}

We now apply the proposed methodology to evaluate the expected utility
of different decision-making systems under alternative utility
structures.  Specifically, we present two main analysis results: the
first examines the difference in expected utilities between human
decisions with and without the PSA recommendation, and the second
examines how the proportion of cash bail would change under the
optimal decision tree policy that maximizes expected utility under
different utility parameters.  A primary goal of our analysis is to
demonstrate how one's conclusions depend on (i) whether the standard
or counterfactual decision-theoretic framework is applied, and (ii)
whether a risk score or triage score is used.

We use the AIPW estimator with the true propensity score,
$e(z, \bx) = 0.5$.  The decision model and outcome model are fitted
using case-level covariates $\bX_i$: gender (male or female), race
(white or non-white), age, PSA inputs including current and past
charges and prior convictions, three PSA risk scores, the overall PSA
recommendation, and probable-cause affidavits text.  Since
probable-cause affidavits text contains a high-dimensional and
unstructured textual information, we employ the GPI methodology
\citep{imai2025genai} using DragonNet \citep{shi2019adapting}, which
is a deep neural network architecture for causal inference.  See
Appendix~\ref{app:assum_indep_validity} for more details.  We use the
same covariates and nuisance components to learn the optimal decision
tree policy in the second part of our analysis.  We restrict the tree
to a maximum depth of two and a minimum leaf size of $10$ observations
to ensure interpretability and avoid overfitting.

Figure~\ref{fig:utility_nca} shows how the expected utility of different decision-making systems varies across utility parameters, using NCA as the main outcome. The figure compares decision-making systems (across rows) and decision-theoretic frameworks defined by utility parameters (across cells). 
Within each decision-making system, the figure varies three parameters: the
regret under release ($r_{\nocrime}^{\ROR}\in\{0,0.1,0.5,1\}$;
columns), the cost of an undesirable outcome under cash bail
($c_{\crime}^{\cash}\in[0,1]$; $x$-axis), and the cost of cash bail
itself ($c^{\cash}\in[0,1]$; $y$-axis). 
The risk
score system corresponds to the vertical slices at
$c_{\crime}^{\cash}=0$ (highlighted by the blue boxes), while the
leftmost column ($r_{\nocrime}^{\ROR}=0$) represents the standard
decision framework without counterfactual regret.

The result demonstrates that the expected utilities of these systems can vary substantially as the utility parameters change. 
For example, as the regret parameter increases across columns and the cost of cash bail increases along the $y$-axis, the expected utilities of the human-alone and human+PSA systems (first two rows) decline much more sharply than that of the optimal decision tree policy (last row).
A similar pattern is shown when comparing the first two rows: when both regret and the cost of cash bail are large, the difference between the human and human+PSA systems becomes more pronounced.

\begin{figure}[ht!]
    \centering
    \includegraphics[width=\linewidth, trim=0 0 0 0.4in, clip]{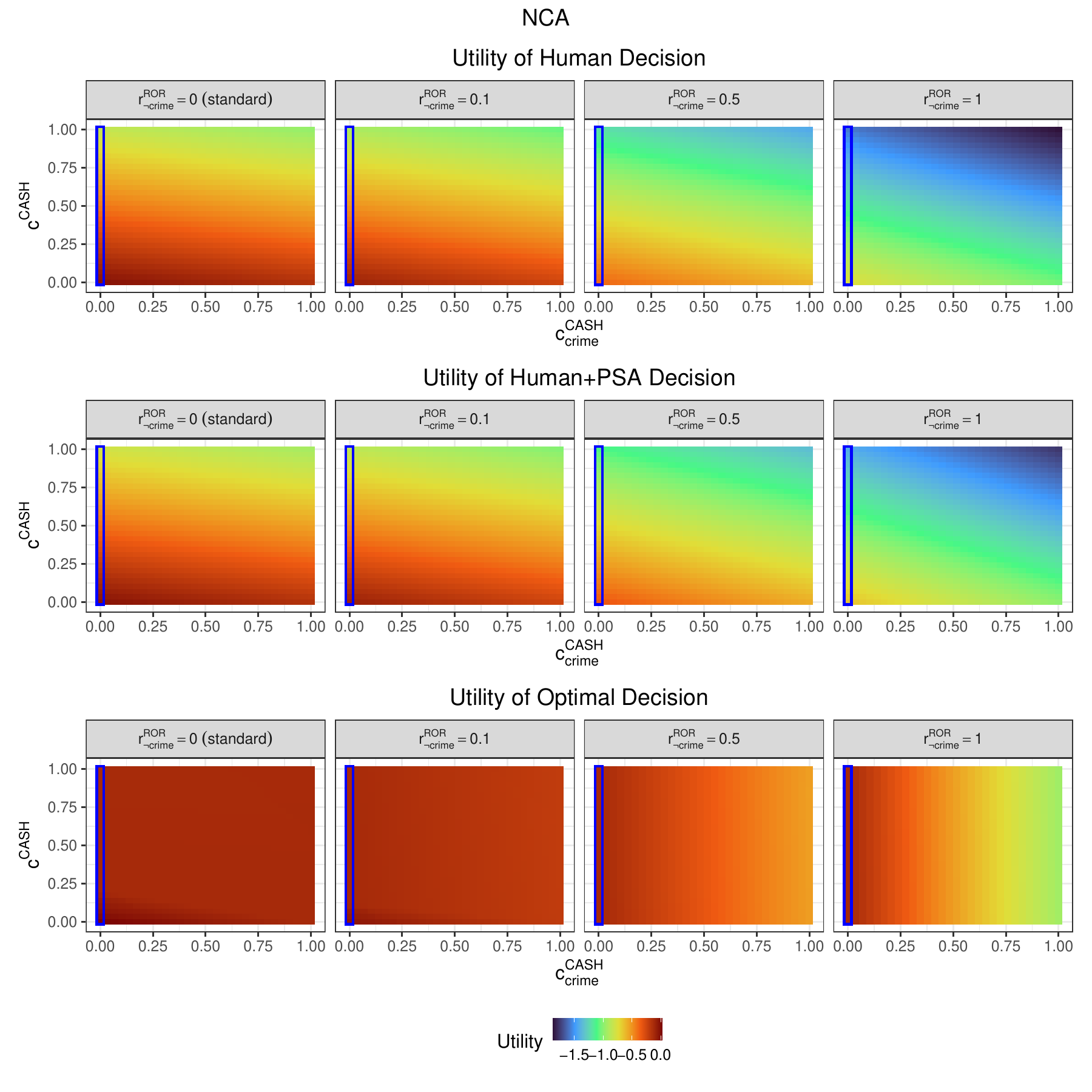}
    \caption{Estimated utility of different decision-making regimes (NCA outcome). 
    The optimal decision tree policy has maximum depth 2 and minimum leaf size 10.
    Each panel corresponds to a different regret parameter $r_{\nocrime}^{\ROR}$; $x$-axis = cost of crime under cash bail $c_{\crime}^{\cash}$; $y$-axis = cost of cash bail $c_d$; $c_{\crime}^{\ROR} = 1$ and $r_{\nocrime}^{\cash} = r_{\nocrime}^{\ROR} c_{\crime}^{\cash}$. Blue region at $c_{\crime}^{\cash} = 0$: risk score framework. First panel ($r_{\nocrime}^{\ROR} = 0$): standard decision framework.}
    \label{fig:utility_nca}
\end{figure}

We now further investigate the difference in expected utilities between human
decisions with and without PSA recommendations.  Specifically, we invert the following hypothesis test,
which gives us the region of utility parameters ($u$) where we can be
confident that decisions made with PSA recommendations would yield a
higher expected utility than those made without them:
\begin{equation*}
    H_{0}: \overline{U}(u; D(1)) \le \overline{U}(u; D(0)) \quad \text{vs} \quad H_{1}: \overline{U}(u; D(1)) > \overline{U}(u; D(0)).
\end{equation*}
Similarly, if we switch the roles of null and alternative hypotheses,
we can determine whether or not decisions made without PSA
recommendations outperform those made with them. The failure to reject
both hypotheses implies that the results are ambiguous. 

\begin{figure}[ht!]
    \centering
    \includegraphics[width=\linewidth]{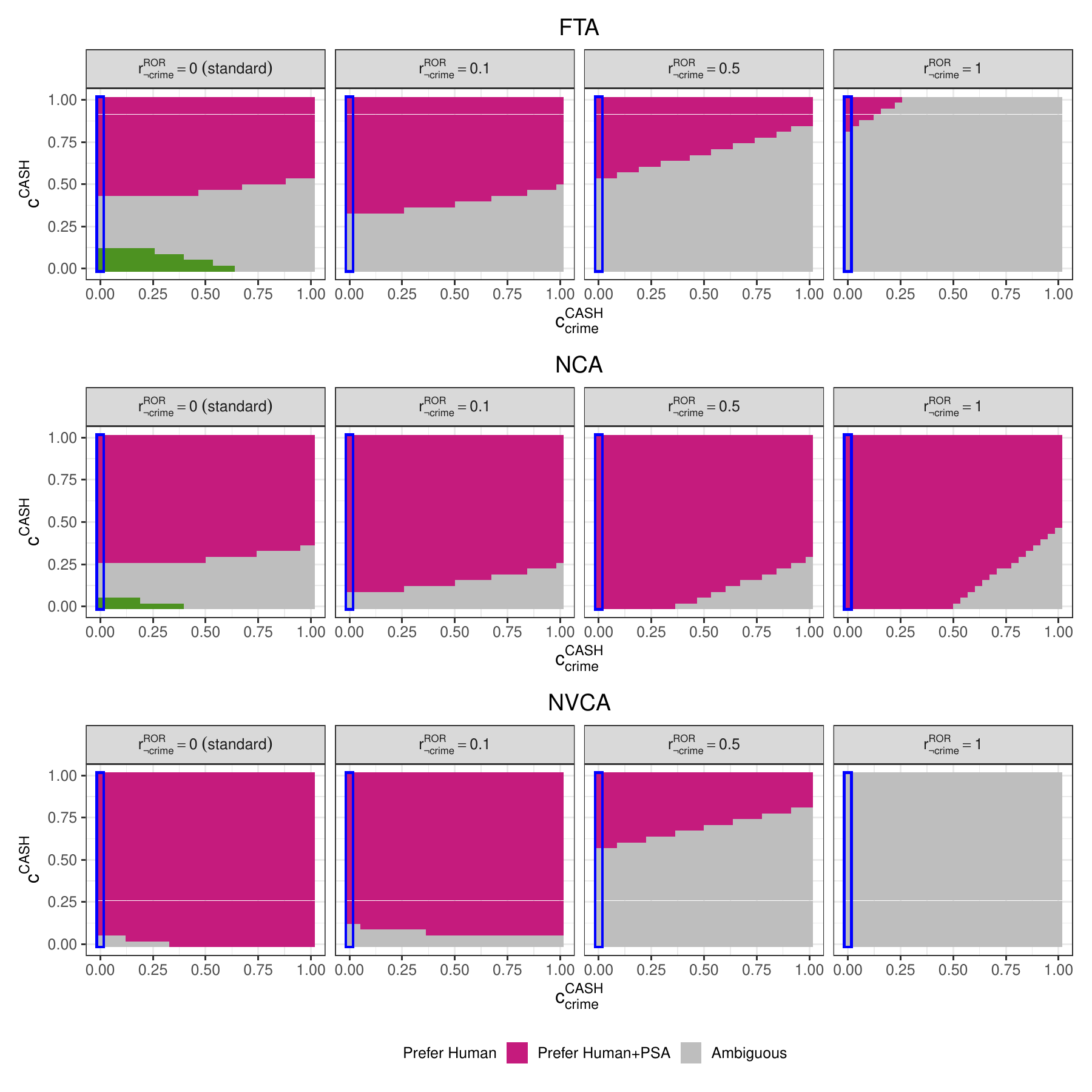}
    \caption{Estimated preference for human decisions over human+PSA recommendations. Each column = regret under release; $x$-axis = cost of new crime under cash; $y$-axis = cost of cash bail; $c_{\crime}^{\ROR} = 1$ and $r_{\nocrime}^{\cash} = r_{\nocrime}^{\ROR} c_{\crime}^{\cash}$. Blue boxes: risk score system (when $c_{\crime}^{\cash} = 0)$. First column: standard decision framework (when $r_{\nocrime}^{\ROR} = r_{\nocrime}^{\cash} =0)$.}
    \label{fig:pref_human_psa}
\end{figure}

Figure~\ref{fig:pref_human_psa} presents, under the specified utility
structure, when human-alone decisions are preferred to decisions made
with PSA recommendations. Results are shown separately
for FTA (top row), NCA (middle row), and NVCA (bottom row).  
By construction of the utility
specification, holding other parameters fixed, increases in the cost
of cash bail make release decisions more preferable.  Recall that
human decisions are, on average, harsher than those informed by PSA
recommendations (see Table~\ref{table:decision}).  Consistent with
these results, the figure shows that higher values of the cost of cash bail (the
$y$-axis) expand the region in which decisions made with PSA
recommendations are preferred (pink).

In contrast, whether cash bail is preferred is theoretically
indeterminate as either the regret under cash bail
($r_{\nocrime}^{\cash}$) or the cost of an undesirable outcome under
cash bail ($c_{\crime}^{\cash}$) increases.  Both parameters increase
the implied regret under cash bail due to the relation
$r_{\nocrime}^{\cash}=r_{\nocrime}^{\ROR}c_{\crime}^{\cash}$, which
can shift preferences toward either harsher or more lenient decisions.
Consequently, the preferred decision rule depends on the distribution
of principal strata in these cases (see
Table~\ref{tab:binary_empirical}).  Empirically,
Figure~\ref{fig:pref_human_psa} shows that the region favoring
decisions with PSA recommendations generally expands as
$c_{\crime}^{\cash}$ (the $x$-axis) increases, though the standard
framework with NVCA (leftmost bottom panel) exhibits the opposite
pattern.

Next, Figure~\ref{fig:pref_human_psa} shows that conclusions about the preferred decision rule can differ depending on both the decision-theoretic framework and the type of score used. 
Comparing the standard framework with the counterfactual framework (the leftmost column \textit{versus} the others), we find that the region favoring decisions made with PSA recommendations tends to expand under the counterfactual framework up to a certain level as both regret parameters increase (e.g., $r_{\nocrime}^{\ROR}$ from $0$ to $0.1$ in the second column), though this effect attenuates at higher levels.
Similarly, comparing the risk score and triage score frameworks
(vertical slices at $c_{\crime}^{\cash}=0$ \textit{versus} the
remaining panels), we observe that the preferred decision rule can
differ across different frameworks. 

Intuitively, moving from the standard to the counterfactual framework
amounts to turning on the regret terms, so the analysis penalizes
unnecessary detention. Because human-alone decisions are harsher than
decisions made with PSA recommendations on average, this initially
expands the region favoring human+PSA recommendations.  However, the
effect attenuates at higher regret levels because our parameterization
also raises the penalty for release in cases where cash bail would
have prevented the undesirable outcome. By contrast, moving from risk
scores to triage scores relaxes the restrictions that collapse cases
with the same baseline risk ($Y(0)$), allowing the utility of cash
bail to depend on whether it was unnecessary or genuinely
preventive. In other words, cases with the same baseline risk can be
assigned different utilities according to their joint potential
outcomes. Thus, the welfare ranking of human-alone versus human+PSA
depends on their mix of case types rather than on overall harshness
alone, making the preferred decision rule more ambiguous.


\begin{figure}[p]
    \centering
    \includegraphics[width=\linewidth]{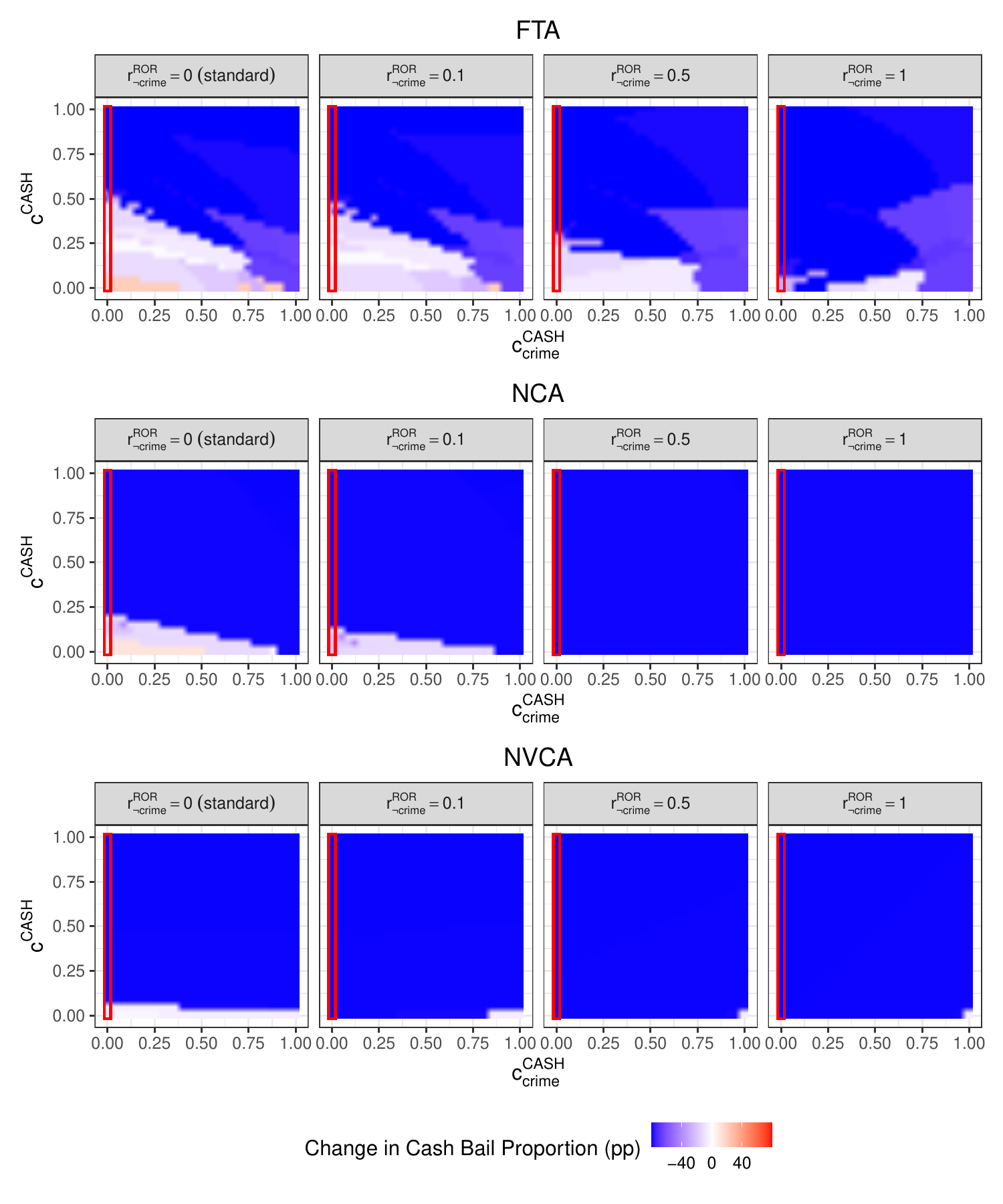}
    \caption{Estimated change in cash bail proportion under optimal decision tree. Each column = regret under release; $x$-axis = cost of new crime under cash; $y$-axis = cost of cash bail; $c_{\crime}^{\ROR} = 1$ and $r_{\nocrime}^{\cash} = r_{\nocrime}^{\ROR} c_{\crime}^{\cash}$. Red boxes: risk score system (when $c_{\crime}^{\cash} = 0)$. First column: standard decision framework (when $r_{\nocrime}^{\ROR} = r_{\nocrime}^{\cash} =0)$.}
    \label{fig:cash_prop}
\end{figure}

Figure~\ref{fig:cash_prop} summarizes how the proportion of cash bail
decisions would change under the estimated optimal policy relative to
the current decision. Each panel varies the same three utility
parameters as before. Darker blue regions correspond to parameter
combinations for which the optimal policy recommends a lower
proportion of cash bail relative to the current decision, while darker
red regions indicate that the optimal policy would increase the use of
cash bail. Whiter regions correspond to parameter combinations for
which the optimal policy does not differ from the current practice in
terms of cash bail proportion.

The main takeaways from this figure are similar to those in
Figure~\ref{fig:pref_human_psa}.  Here, we highlight the findings that
illustrate the practical difference between the conventional
\emph{risk score} framework and the proposed \emph{triage score}
framework.  The red boxes correspond to the risk score system, which
evaluates decisions solely based on the baseline potential outcome.
In contrast, the triage score framework explicitly incorporates
counterfactual outcomes under alternative decisions, allowing the
policy to account for cases where cash bail may be unnecessary or
counterproductive.  As a result, the optimal policy can differ
substantially between the two frameworks even under the same cash bail
cost parameters (i.e., the $y$-axis).  For example, as we increase the
cost of FTA under cash bail (i.e., the $x$-axis) in the second column
of the top row, the optimal policy recommends a substantially lower
proportion of cash bail under the triage score framework (darker
blue).  Overall, the counterfactual triage score framework provides a
more flexible policy evaluation by accommodating richer and more
tailored utility specifications.  The empirical results show that
policy recommendations derived from this framework can differ
substantially from those implied by conventional risk score
approaches.

\section{Concluding Remarks}
\label{sec:conclusion}

This paper proposes a new framework of triage score that can
accommodate counterfactual utilities while incorporating standard
utilities.  Unlike popular risk score approaches that focus on the
baseline potential outcome alone, triage score considers all potential
outcomes, allowing decision makers to choose the best decision among
alternatives. While we applied the proposed methodology to an RCT in
criminal justice, triage score represents a general decision-support
framework that can be applied to a variety of settings.

Future work should consider the real-world application of triage score
framework by directly collaborating with policy makers.  Doing so
requires careful elicitation of utility parameters.  Moreover, the
proposed framework can be generalized to dynamic settings, where
decisions are made sequentially over time, generating carryover
effects and other temporal dependencies.  An ambitious goal is to
develop a dynamic triage score that are dynamically updated as more
decisions are made and additional data become available over time.

\bigskip
\bibliographystyle{chicago}
\bibliography{bibliography}

@article{imai2025genai,
  title={{GenAI}-powered inference},
  author={Imai, Kosuke and Nakamura, Kentaro},
  journal={arXiv preprint arXiv:2507.03897},
  year={2025}
}

@article{feder2022causal,
  title={Causal inference in natural language processing: Estimation, prediction, interpretation and beyond},
  author={Feder, Amir and Keith, Katherine A and Manzoor, Emaad and Pryzant, Reid and Sridhar, Dhanya and Wood-Doughty, Zach and Eisenstein, Jacob and Grimmer, Justin and Reichart, Roi and Roberts, Margaret E and others},
  journal={Transactions of the Association for Computational Linguistics},
  volume={10},
  pages={1138--1158},
  year={2022},
  publisher={MIT Press One Broadway, 12th Floor, Cambridge, Massachusetts 02142, USA~…}
}

@article{stevenson2018assessing,
  title={Assessing risk assessment in action},
  author={Stevenson, Megan},
  journal={Minnesota Law Review},
  volume={103},
  pages={303--384},
  year={2018}
}

@article{dagostino2008general,
  title={General cardiovascular risk profile for use in primary care},
  author={D'Agostino, Ralph B and Vasan, Ramachandran S and Pencina, Michael J and Wolf, Philip A and Cobain, Mark and Massaro, Joseph M and Kannel, William B},
  journal={Circulation},
  volume={117},
  number={6},
  pages={743--753},
  year={2008}
}

@article{chen2021probabilistic,
  title={Probabilistic machine learning for healthcare},
  author={Chen, Irene Y and Joshi, Shalmali and Ghassemi, Marzyeh and Ranganath, Rajesh},
  journal={Annual review of biomedical data science},
  volume={4},
  number={1},
  pages={393--415},
  year={2021},
  publisher={Annual Reviews}
}

@article{einav2013impact,
  title={The impact of credit scoring on consumer lending},
  author={Einav, Liran and Jenkins, Mark and Levin, Jonathan},
  journal={The RAND Journal of Economics},
  volume={44},
  number={2},
  pages={249--274},
  year={2013},
  publisher={Wiley Online Library}
}

@article{dobbie2021measuring,
  title={Measuring bias in consumer lending},
  author={Dobbie, Will and Liberman, Andres and Paravisini, Daniel and Pathania, Vikram},
  journal={The Review of Economic Studies},
  volume={88},
  number={6},
  pages={2799--2832},
  year={2021},
  publisher={Oxford University Press}
}

@article{kleinberg2018human,
  title={Human decisions and machine predictions},
  author={Kleinberg, Jon and Lakkaraju, Himabindu and Leskovec, Jure and Ludwig, Jens and Mullainathan, Sendhil},
  journal={The quarterly journal of economics},
  volume={133},
  number={1},
  pages={237--293},
  year={2018},
  publisher={Oxford University Press}
}

@inproceedings{coston2021characterizing,
  title={Characterizing fairness over the set of good models under selective labels},
  author={Coston, Amanda and Rambachan, Ashesh and Chouldechova, Alexandra},
  booktitle={International Conference on Machine Learning},
  pages={2144--2155},
  year={2021},
  organization={PMLR}
}

@article{angelova2025algorithmic,
  title={Algorithmic recommendations and human discretion},
  author={Angelova, Victoria and Dobbie, Will and Yang, Crystal S},
  journal={Review of Economic Studies},
  pages={rdaf084},
  year={2025},
  publisher={Oxford University Press UK}
}

@incollection{wald1950statistical,
  title={Statistical decision functions},
  author={Wald, Abraham},
  booktitle={Breakthroughs in Statistics: Foundations and Basic Theory},
  pages={342--357},
  year={1950},
  publisher={Springer}
}

@article{jiang2024longitudinal,
  title={Longitudinal Causal Inference with Selective Eligibility},
  author={Jiang, Zhichao and Ben-Michael, Eli and Greiner, D James and Halen, Ryan and Imai, Kosuke},
  journal={arXiv preprint arXiv:2410.17864},
  year={2024}
}

@article{mueller2023personalized,
  title={Personalized decision making--a conceptual introduction},
  author={Mueller, Scott and Pearl, Judea},
  journal={Journal of Causal Inference},
  volume={11},
  number={1},
  pages={20220050},
  year={2023},
  publisher={De Gruyter}
}

@article{christy2024starting,
  title={Starting small: Prioritizing safety over efficacy in randomized experiments using the exact finite sample likelihood},
  author={Christy, Neil and Kowalski, Amanda Ellen},
  journal={arXiv preprint arXiv:2407.18206},
  year={2024}
}

@article{bell1982regret,
  author    = {David E. Bell},
  title     = {Regret in Decision Making under Uncertainty},
  journal   = {Operations Research},
  year      = {1982},
  volume    = {30},
  number    = {5},
  pages     = {961--981},
  doi       = {10.1287/opre.30.5.961},
}

@article{loomes1982regret,
  author    = {Graham Loomes and Robert Sugden},
  title     = {Regret Theory: An Alternative Theory of Rational Choice Under Uncertainty},
  journal   = {The Economic Journal},
  volume    = {92},
  number    = {368},
  pages     = {805--824},
  year      = {1982},
  publisher = {Oxford University Press},
  jstor     = {2232669},
  url       = {https://www.jstor.org/stable/2232669}
}

@article{fran:rubi:02,
	Author = {Frangakis, Constantine E. and Rubin, Donald B.},
	Journal = {Biometrics},
	Number = 1,
	Pages = {21--29},
	Title = {Principal Stratification in Causal Inference},
	Volume = 58,
	Year = 2002}

@article{rubi:90,
	Author = {Rubin, Donald B.},
	Journal = {Statistical Science},
	Pages = {472--480},
	Title = {Comments on ``{O}n the Application of Probability Theory to Agricultural Experiments. {E}ssay on Principles. {S}ection 9'' by {J.} {S}plawa-{N}eyman translated from the {P}olish and edited by {D.} {M.} {D}abrowska and {T.} {P.} {S}peed},
	Volume = 5,
	Year = 1990}

@article{koch2025statistical,
  title={Statistical Decision Theory with Counterfactual Loss},
  author={Koch, Benedikt and Imai, Kosuke},
  journal={arXiv preprint arXiv:2505.08908},
  year={2025}
}

@article{ben2024does,
  title={Does {AI} help humans make better decisions? A statistical evaluation framework for experimental and observational studies},
  author={Ben-Michael, Eli and Greiner, D James and Huang, Melody and Imai, Kosuke and Jiang, Zhichao and Shin, Sooahn},
  journal={Proceedings of the National Academy of Sciences},
  volume={122},
  number={38},
  pages={e2505106122},
  year={2025}
}

@article{imai2024causalrepresentationlearninggenerative,
      title={Causal Representation Learning with Generative Artificial Intelligence: Application to Texts as Treatments},
  author={Imai, Kosuke and Nakamura, Kentaro},
  journal={arXiv preprint arXiv:2410.00903},
  year={2024}
}

@article{touvron2023llama,
  title={Llama: Open and efficient foundation language models},
  author={Touvron, Hugo and Lavril, Thibaut and Izacard, Gautier and Martinet, Xavier and Lachaux, Marie-Anne and Lacroix, Timoth{\'e}e and Rozi{\`e}re, Baptiste and Goyal, Naman and Hambro, Eric and Azhar, Faisal and others},
  journal={arXiv preprint arXiv:2302.13971},
  year={2023}
}

@article{benm:imai:jian:23,
  author = 	 {Ben-Michael, Eli and Imai, Kosuke and Jiang, Zhichao},
  title = 	 {Policy learning with asymmetric counterfactual utilities},
  journal =      {Journal of the American Statistical Association},
  year = 	 {2024},
  volume = 	 {119},
  OPTtype = 	 {},
  number = 	 {548},
  pages = 	 {3045--3058},
  OPTmonth = 	 {},
  OPTnote = 	 {},
  OPTannote = 	 {}
}

@article{rambachan2022robust,
  title={Robust design and evaluation of predictive algorithms under unobserved confounding},
  author={Rambachan, Ashesh and Coston, Amanda and Kennedy, Edward},
  journal={arXiv preprint arXiv:2212.09844},
  year={2022}
}

@inproceedings{lakkaraju2017selective,
  title={The selective labels problem: Evaluating algorithmic predictions in the presence of unobservables},
  author={Lakkaraju, Himabindu and Kleinberg, Jon and Leskovec, Jure and Ludwig, Jens and Mullainathan, Sendhil},
  booktitle={Proceedings of the 23rd ACM SIGKDD International Conference on Knowledge Discovery and Data Mining},
  pages={275--284},
  year={2017}
}

@article{dobbie2018effects,
  title={The effects of pre-trial detention on conviction, future crime, and employment: Evidence from randomly assigned judges},
  author={Dobbie, Will and Goldin, Jacob and Yang, Crystal S},
  journal={American Economic Review},
  volume={108},
  number={2},
  pages={201--240},
  year={2018},
  publisher={American Economic Association 2014 Broadway, Suite 305, Nashville, TN 37203}
}

@article{imai2023experimental,
  title={Experimental evaluation of algorithm-assisted human decision-making: Application to pretrial public safety assessment},
  author={Imai, Kosuke and Jiang, Zhichao and Greiner, D James and Halen, Ryan and Shin, Sooahn},
  journal={Journal of the Royal Statistical Society Series A: Statistics in Society},
  volume={186},
  number={2},
  pages={167--189},
  year={2023},
  publisher={Oxford University Press US}
}

@article{stevenson2022algorithmic,
  title={Algorithmic risk assessment in the hands of humans},
  author={Stevenson, Megan T and Doleac, Jennifer L},
  journal={Available at SSRN 3489440},
  year={2022}
}

@article{goel2016personalized,
  title={Personalized risk assessments in the criminal justice system},
  author={Goel, Sharad and Rao, Justin M and Shroff, Ravi},
  journal={American Economic Review},
  volume={106},
  number={5},
  pages={119--123},
  year={2016},
  publisher={American Economic Association 2014 Broadway, Suite 305, Nashville, TN 37203}
}

@article{berk2021fairness,
  title={Fairness in criminal justice risk assessments: The state of the art},
  author={Berk, Richard and Heidari, Hoda and Jabbari, Shahin and Kearns, Michael and Roth, Aaron},
  journal={Sociological Methods \& Research},
  volume={50},
  number={1},
  pages={3--44},
  year={2021},
  publisher={Sage Publications Sage CA: Los Angeles, CA}
}

@article{miller2013practitioner,
  title={Practitioner compliance with risk/needs assessment tools: A theoretical and empirical assessment},
  author={Miller, Joel and Maloney, Carrie},
  journal={Criminal Justice and Behavior},
  volume={40},
  number={7},
  pages={716--736},
  year={2013},
  publisher={Sage Publications Sage CA: Los Angeles, CA}
}

@article{albright2019if,
  title={If you give a judge a risk score: evidence from Kentucky bail decisions},
  author={Albright, Alex},
  journal={Law, Economics, and Business Fellows’ Discussion Paper Series},
  volume={85},
  year={2019}
}

@article{skeem2020impact,
  title={Impact of risk assessment on judges’ fairness in sentencing relatively poor defendants.},
  author={Skeem, Jennifer and Scurich, Nicholas and Monahan, John},
  journal={Law and human behavior},
  volume={44},
  number={1},
  pages={51},
  year={2020},
  publisher={Educational Publishing Foundation}
}

@inproceedings{guerdan2023ground,
  title={Ground (less) Truth: A Causal Framework for Proxy Labels in Human-Algorithm Decision-Making},
  author={Guerdan, Luke and Coston, Amanda and Wu, Zhiwei Steven and Holstein, Kenneth},
  booktitle={Proceedings of the 2023 ACM Conference on Fairness, Accountability, and Transparency},
  pages={688--704},
  year={2023}
}

@article{arnold2022measuring,
  title={Measuring racial discrimination in bail decisions},
  author={Arnold, David and Dobbie, Will and Hull, Peter},
  journal={American Economic Review},
  volume={112},
  number={9},
  pages={2992--3038},
  year={2022},
  publisher={American Economic Association 2014 Broadway, Suite 305, Nashville, TN 37203}
}

@inproceedings{coston2020counterfactual,
  title={Counterfactual risk assessments, evaluation, and fairness},
  author={Coston, Amanda and Mishler, Alan and Kennedy, Edward H and Chouldechova, Alexandra},
  booktitle={Proceedings of the 2020 conference on fairness, accountability, and transparency},
  pages={582--593},
  year={2020}
}

@article{kennedy_semiparametric_2023,
	title={Semiparametric doubly robust targeted double machine learning: a review},
  author={Kennedy, Edward H},
  journal={Handbook of Statistical Methods for Precision Medicine},
  pages={207--236},
  year={2024},
  publisher={Chapman and Hall/CRC}
}

@article{shi2019adapting,
  title={Adapting neural networks for the estimation of treatment effects},
  author={Shi, Claudia and Blei, David and Veitch, Victor},
  journal={Advances in neural information processing systems},
  volume={32},
  year={2019}
}

\newpage
\appendix

\setcounter{figure}{0}
\setcounter{table}{0}
\setcounter{theorem}{0}

\renewcommand{\thefigure}{\thesection.\arabic{figure}}
\renewcommand{\thetable}{\thesection.\arabic{table}}
\renewcommand{\thetheorem}{\thesection.\arabic{theorem}}
\renewcommand{\theequation}{\thesection.\arabic{equation}}
\renewcommand{\theassumption}{\thesection.\arabic{assumption}}
\renewcommand{\theremark}{\thesection.\arabic{remark}}
\renewcommand{\theproposition}{\thesection.\arabic{proposition}}
\renewcommand{\thecorollary}{\thesection.\arabic{corollary}}
\renewcommand{\thelemma}{\thesection.\arabic{lemma}}

\section{Additional Discussions about Assumptions}

\subsection{Assumption~\ref{assum:indep} and Alternative Assumptions}
\label{app:assum_indep}

As an alternative to Assumption~\ref{assum:indep}, we could adopt the following set of assumptions:
\begin{align}
    \label{assum:indep2}
    \left\{\{Y_i(d)\}_{d \in [0:K_{D}-1]}, R_i\right\} &\indep D_i \mid \bX_i, Z_i = 0 \\
    \label{assum:indep2_2}
    \{Y_i(d)\}_{d \in [0:K_{D}-1]} &\indep D_i \mid \bX_i, R_i, Z_i = 1
\end{align}
which are adapted from Assumption 3(a) of \cite{koch2025statistical}.
Assumption 3(a) in \cite{koch2025statistical} states:
$\left\{\{Y_i(d)\}_{d \in [0:K_{D}-1]}, R_i\right\} \indep D_i \mid \bX_i,$
which implies that $D_i$ must be independent of $R_i$ conditional on $\bX_i$. 
This assumption does not hold in our setting where judge's decision may be affected 
by PSA when it is provided (i.e., $Z_i = 1$), so we instead adopt the modified version
given in \eqref{assum:indep2} and \eqref{assum:indep2_2}.

However, we can show that this is a stronger assumption than Assumption~\ref{assum:indep}; 
it implies Assumption~\ref{assum:indep} (see Remark~\ref{remark1}), but the converse does not hold, 
as demonstrated by a counterexample in Remark~\ref{remark2}.

\begin{remark}\label{remark1}
    Suppose the conditional independence conditions in \eqref{assum:indep2} and
    \eqref{assum:indep2_2} hold. For the $z=0$ case, we have
    \begin{align*}
       \Pr(Y_{i}(d) = y \mid R_{i} = r, Z_{i} = 0, \bX_{i} = \bx)
        &=
        \frac{\Pr(Y_{i}(d) = y, R_{i} = r \mid Z_{i} = 0, \bX_{i} = \bx)}{\Pr(R_{i} = r \mid Z_{i} = 0, \bX_{i} = \bx)} \\
        &=
        \frac{\Pr(Y_{i}(d) = y, R_{i} = r \mid D_{i} = d, Z_{i} = 0, \bX_{i} = \bx)}{\Pr(R_{i} = r \mid D_{i} = d, Z_{i} = 0, \bX_{i} = \bx)} \\
        &= \Pr(Y_{i}(d) = y \mid D_{i} = d, R_{i} = r, Z_{i} = 0, \bX_{i} = \bx)
    \end{align*}
    for each $d$ and all $y,r,\bx$. Thus Assumption~\ref{assum:indep} holds under $z = 0$.
    The case $z=1$ follows analogously from \eqref{assum:indep2_2}.
\end{remark}

\begin{remark}\label{remark2}
    In the following, we provide an example of a structural equation model in which Assumption~\ref{assum:indep} holds, but the conditional independence in \eqref{assum:indep2} does not.
    Let
    \begin{align*}
       Z_{i} &= f_{Z}(\bX_{i}, \epsilon_{Z,i}) \\
       D_{i} &= f_{D_{1}}(A_{i}, \bX_{i}, \bU_{D,i}, \epsilon_{D,i})Z_{i} 
              + f_{D_{0}}(\bX_{i}, \bU_{D,i}, \epsilon_{D,i})(1-Z_{i}) \\
       A_{i} &= f_{A}(\bX_{i}, \bU_{D,i}, \bU_{Y,i}, \epsilon_{A,i}) \\
       Y_{i} &= f_{Y}(D_{i}, \bX_{i}, \bU_{Y,i}, \epsilon_{Y,i})
    \end{align*}
where $f_{Z}, f_{D_{1}}, f_{D_{0}}, f_{A}, f_{Y}$ are some functions, 
$\epsilon_{Z,i}, \epsilon_{D,i}, \epsilon_{A,i}, \epsilon_{Y,i}$ are error terms, 
and $\bX_{i}$, $\bU_{D,i}$, and $\bU_{Y,i}$ are exogenous. A corresponding causal diagram is shown in Figure~\ref{fig:dag}.
We can show that Assumption~\ref{assum:indep} holds using the backdoor criterion for each $Z_{i} = 0, 1$ cases.
However, $A_{i} \not\indep D_{i} \mid \bX_{i}, Z_{i} = 0$ because of unobserved confounder $U_{D,i}$.
\end{remark}

\subsection{Validity of Assumption~\ref{assum:indep} in the Application}
\label{app:assum_indep_validity}

In the Utah RCT, we have access to all the information that judges had at the time of their pretrial decisions, including structured case covariates and unstructured text from probable-cause (PC) affidavits. 
The PC affidavit contains rich, case-specific facts
about the arrest as well as charging table information.
Therefore, conditioning on the information in
these affidavits in addition to the structured case covariates $\bX_i$ and the PSA-related variables
$(Z_i, R_i)$ renders Assumption~\ref{assum:indep} plausible in the application.

Because PC affidavits are high-dimensional unstructured text, we adjust for them using
GenAI-powered inference (GPI) \citep{imai2025genai}.  The role of GPI in this project is to
construct a low-dimensional \emph{deconfounder} from affidavit text that can be incorporated into
the nuisance functions $(m_d^D, m_y^Y)$ used in our AIPW estimators in
Section~\ref{sec:statistical}.  Figure~\ref{fig:dag_application} summarizes the assumed
data-generating process.

\begin{figure}[ht!]
    \centering
        \begin{tikzpicture}[node distance=2cm, >={Stealth}]

            \node[circle, draw] (ai) at (0, 0) {$R$};
            \node[circle, draw] (x) [above=1cm of ai] {$\bX$};
            \node[circle, draw] (text) [above=2.5cm of ai] {$\bS$};
            \node[circle, draw] (u) [right=2.5cm of x] {$g_{U}(\bS)$};
            \node[circle, draw] (decision) [right=2cm of ai] {$D$};
            \node[circle, draw] (outcome) [right=2cm of decision] {$Y$};

        \draw[->] (ai) -- (decision);
        \draw[->] (decision) -- (outcome);
        \draw[->] (x) -- (text);
        \draw[->] (x) -- (ai);
        \draw[->] (x) -- (decision);
        \draw[->] (x) -- (outcome);
        \draw[<-] (u) -- (text);
        \draw[->] (u) -- (decision);
        \draw[->] (u) -- (outcome);

        \end{tikzpicture}
    \caption{A causal diagram illustrating the application study. When $Z = 0$, there is no edge $R \rightarrow D$. $\bS$ represents the PC affidavit text, and $g_{U}(\bS)$ is a low-dimensional affidavit feature vector extracted from the text.}
    \label{fig:dag_application}
\end{figure}
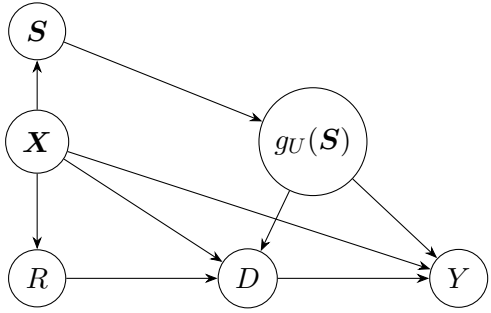

\subsubsection{Setup}

For each case $i$, in addition to $(Y_i, D_i, Z_i, R_i, \bX_i)$ defined in the main text, we also
observe the PC affidavit text $\bS_i$.
We maintain the potential outcomes notation $\{Y_i(d)\}_{d\in\cD}$ and consistency,
$Y_i = Y_i(D_i)$.
Throughout this appendix, we use the boldface notation $\bL_i$ for a (vector-valued) internal representation extracted
from a large language model (LLM) that processes the affidavit text.

The key substantive idea is that the affidavit contains latent, low-dimensional case features that
confound the relationship between the judge's decision and the outcome.  We formalize this via a
latent ignorability condition adapted to the affidavit setting.

\begin{assumption}[Latent ignorability given low-dimensional affidavit features]
\label{assum:latent_ignorability_affidavit}
There exists a deterministic function $g_U(\cdot)$ mapping affidavit text to a low-dimensional
vector of confounding features,
\begin{equation*}
    \bU_i := g_U(\bS_i) \in \mathbb{R}^{p},
\end{equation*}
such that the potential outcomes are conditionally independent of the observed decision given
$(Z_i,R_i,\bX_i)$ and these features:
\begin{equation*}
    \{Y_i(d)\}_{d\in\cD} \ \indep \ D_i \mid Z_i, R_i, \bX_i, \bU_i.
\end{equation*}
In addition, for all $d\in\cD$ on the relevant support, there exists
$\epsilon > 0$ such that
$\Pr(D_i=d \mid Z_i, R_i, \bX_i, \bU_i) > \epsilon$.
\end{assumption}

Assumption~\ref{assum:latent_ignorability_affidavit} is a strengthening of the substantive claim
that ``we observe all information the judge uses,'' allowing that the relevant portion of the
affidavit can be represented by an unknown but low-dimensional summary $\bU_i$ that satisfies 

\subsubsection{Identification via GPI}

Directly conditioning on the raw text $\bS_i$ is undesirable in practice because it is
high-dimensional and can induce severe sparsity and near-deterministic prediction of the decision,
which undermines effective adjustment.  GPI \citep{imai2025genai} addresses this by leveraging a
deep generative model (e.g., an LLM) to obtain an internal representation $\bL_i$ and then learning
a lower-dimensional deconfounder $f(\bL_i)$.

Following \cite{imai2025genai}, we view the affidavit text as being generated by a deep generative
model with an internal representation.  Formally, let $P_i$ be a prompt, and let
$\bL_i \in \mathcal{L}\subseteq\mathbb{R}^{r}$ be an internal representation produced by the model.
The model generates $\bS_i$ through a decoder indexed by parameters $\gamma$,
\begin{equation*}
    \Pr(\bS_i \mid h_{\gamma}(\bL_i)) \ \Pr(\bL_i \mid P_i).
\end{equation*}

\begin{assumption}[Deterministic decoding]
\label{assum:deterministic_decoding_affidavit}
Conditional on $\bL_i$, the affidavit text $\bS_i$ is deterministically generated; equivalently,
$\Pr(\bS_i = h_{\gamma}(\bL_i)\mid \bL_i) = 1$.
\end{assumption}

In our application, where affidavits are \emph{observed} rather than generated for the study, we
use a ``text-reuse'' prompting strategy: for each observed affidavit $\bS_i$, we prompt an LLM to
output the same text under deterministic decoding, and we extract the
associated internal representation $\bL_i$ from the model.

Under Assumption~\ref{assum:deterministic_decoding_affidavit}, any low-dimensional feature
$\bU_i=g_U(\bS_i)$ is also a deterministic function of $\bL_i$ because $\bS_i$ is a deterministic
function of $\bL_i$.  Thus, there exists a (not necessarily unique) lower-dimensional function of
$\bL_i$ that suffices for deconfounding.

\begin{proposition}[Existence of a deconfounder and identification]
\label{prop:gpi_affidavit_id}
Suppose Assumption~\ref{assum:latent_ignorability_affidavit} holds and the affidavit text satisfies
Assumption~\ref{assum:deterministic_decoding_affidavit}. Then there exists a (possibly non-unique)
deconfounder function $f:\mathcal{L}\to\mathbb{R}^{q}$ with $q\le r$ such that
\begin{equation*}
    \{Y_i(d)\}_{d\in\cD} \ \indep \ \bL_i \mid D_i, Z_i, R_i, \bX_i, f(\bL_i).
\end{equation*}
Moreover, for any $d\in\cD$ and $y\in\cY$, the marginal distribution of the potential outcome is
identified by
\begin{align*}
    \Pr(Y_i(d)=y)
    &=
    \E\left[\Pr\left(Y_i=y \mid D_i=d, Z_i, R_i, \bX_i, f(\bL_i)\right)\right].
\end{align*}
\end{proposition}

The proof of Proposition~\ref{prop:gpi_affidavit_id} follows from that of Proposition 1 in \cite{imai2025genai} and is provided in Appendix S4.1 of their paper.

\paragraph{Implication for Assumption~\ref{assum:indep}.}
Let the augmented covariate vector be
\begin{equation*}
    \widetilde{\bX}_i := (\bX_i, f(\bL_i)).
\end{equation*}
Then Proposition~\ref{prop:gpi_affidavit_id} implies that conditioning on $\widetilde{\bX}_i$
renders the decision unconfounded with respect to potential outcomes in the sense required by
Assumption~\ref{assum:indep}.  For notational simplicity, in the empirical analysis we treat the
estimated deconfounder as part of the covariate vector and write $\bX_i$ to include it.

\subsubsection{Estimation and Implementation}

Our implementation of GPI in this paper is designed to estimate the nuisance functions
\begin{align*}
    m_{d}^{D}(z,r,\bx)
    &:= \Pr(D = d \mid Z = z, R = r, \bX = \bx), \\
    m_{y}^{Y}(z,r,d,\bx)
    &:= \Pr(Y = y \mid Z = z, R = r, D = d, \bX = \bx),
\end{align*}
where $\bX$ includes structured covariates and an estimated affidavit-based deconfounder.

We proceed in three steps.

\begin{enumerate}[label=(\arabic*)]
\item Extract internal representations from affidavits.
For each observed affidavit $\bS_i$, we obtain an internal representation $\bL_i$ from an
open-source LLM under deterministic decoding.  We then construct a fixed-length representation
(e.g., by pooling token-level hidden states) to use as the input $\bL_i$ for GPI.

\item Estimate a deconfounder and nuisance models.
We parameterize a deconfounder map $f(\bL;\lambda)$ as a neural network that outputs a
$q$-dimensional bottleneck representation.  We then estimate $f(\cdot)$ jointly with prediction
models for the decision and outcome.  In particular, we use DragonNet \citep{shi2019adapting} to
learn a shared representation and task-specific heads, yielding estimates
$\widehat{f}(\bL_i)$, $\hat{m}_{d}^{D}(z,r,\bX_i)$, and
$\hat{m}_{y}^{Y}(z,r,d,\bX_i)$.

\item Plug into AIPW evaluation.
Specifically, we plug $\hat{m}_{d}^{D}$ and $\hat{m}_{y}^{Y}$ into the AIPW estimators defined in
Section~\ref{sec:statistical}.  In the RCT, the treatment propensity score is known,
$e(z,\bx)=0.5$, so we do not estimate $e(\cdot)$.
\end{enumerate}

This procedure allows us to adjust for the rich information contained in PC affidavits while retaining a low-dimensional conditioning set. This supports both the plausibility of Assumption~\ref{assum:indep} and the practical estimation of the nuisance components required for our semiparametric estimators.

Formally, the validity of this procedure relies on the assumptions stated above together with the consistency of the learned deconfounder. In particular, we assume that the learned representation $\widehat{f}(\bL_i)$ converges to a population-level representation $f(\bL_i)$ such that $\{Y_i(d)\}_{d\in\cD} \ \indep \ \bL_i \mid D_i, Z_i, R_i, \bX_i, f(\bL_i)$, described in Proposition~\ref{prop:gpi_affidavit_id}.

In practice, we assume that the DragonNet architecture and training procedure are sufficiently flexible to learn a representation satisfying this condition. The architecture maps the high-dimensional representation $\bL$ into a shared representation $f(\bL)$, which is then used by both the outcome and decision heads,
\begin{equation*}
\bL \ \longrightarrow\ f(\bL) \ \longrightarrow\ \{h_d(f(\bL)),\, p_d(f(\bL))\}_{d=0}^{K_{D}-1},
\end{equation*}
where $f$ denotes the final layer of the shared representation network, $h_d$ represents the outcome head for the outcome model under decision $d$, and $p_d$ denotes the decision model.

In the binary decision case, for example, DragonNet jointly optimizes the following objective:
\begin{align*}
&\{\widehat{\lambda}, \widehat{\theta}_0, \widehat{\theta}_1, \widehat{\phi}\}_{\text{DragonNet}} \\
&\qquad =
\argmin_{\lambda,\theta_0,\theta_1,\phi}
\frac{1}{n}\sum_{i=1}^n
\Bigl\{
Y_i - h_{D_i}\bigl(f(\bL_i)\bigr)
\Bigr\}^2 
- \alpha
\underbrace{
\frac{1}{n}\sum_{i=1}^n
\Bigl[
D_i \log\bigl(p(f(\bL_i))\bigr)
+
(1-D_i)\log\bigl(1-p(f(\bL_i))\bigr)
\Bigr]
}_{\text{cross-entropy loss for the decision model}},
\end{align*}
where $\lambda, \theta_0, \theta_1,$ and $\phi$ denote the parameters of the networks $f$, $h_0$, $h_1$, and $p$, respectively.

We assume that these two objectives do not conflict. Formally, let $\mathcal{S}_Y = \sigma(\mu_0(\bL), \mu_1(\bL))$ denote the $\sigma$-algebra generated by the outcome mean functions and let $\mathcal{S}_D = \sigma(p(\bL))$ denote the $\sigma$-algebra generated by the propensity score. We assume that there exists a representation $f$ such that $\mathcal{S}_Y \vee \mathcal{S}_D \subseteq \sigma\bigl(f(\bL)\bigr)$.
Under this assumption, we employ a DragonNet architecture with sufficient width and depth so that the learned representation $f(\bL)$ can serve as a deconfounder satisfying the required conditional independence.





\subsection{Rate Conditions}
\label{app:rate_conditions}

\begin{assumption}[Rate conditions] \label{assum:rate_condition} 
  For each $z = 0,1$, $r \in \cR$, $d,d^\prime \in \cD$, and $y \in \cY$,
  we have:
  \begin{align*}
    &\lVert m_d^{D}(z,r,\cdot) - \hat{m}_d^{D}(z,r,\cdot) \rVert_2
    \times 
   \lVert e(z,\cdot) - \hat{e}(z,\cdot) \rVert_2 = o_p(n^{-\frac{1}{2}}), \\
    &\lVert m_y^{Y}(z,r,d,\cdot) - \hat{m}_y^{Y}(z,r,d,\cdot) \rVert_2
    \times 
   \lVert e(z,\cdot) - \hat{e}(z,\cdot) \rVert_2 = o_p(n^{-\frac{1}{2}}), \\
   &\lVert m_y^{Y}(z,r,d,\cdot) - \hat{m}_y^{Y}(z,r,d,\cdot) \rVert_2
   \times \lVert m_d^{D}(z,r,\cdot) - \hat{m}_d^{D}(z,r,\cdot) \rVert_2 = o_p(n^{-\frac{1}{2}}), \\
   &\lVert m_y^{Y}(z,r,d^\prime,\cdot) - \hat{m}_y^{Y}(z,r,d^\prime,\cdot) \rVert_2
   \times \lVert m_d^{D}(z,r,\cdot) - \hat{m}_d^{D}(z,r,\cdot) \rVert_2 = o_p(n^{-\frac{1}{2}}), \\
    &\lVert m_y^{Y}(z,r,d,\cdot) - \hat{m}_y^{Y}(z,r,d,\cdot) \rVert_\infty = o_p(1), \\
    &\lVert m_d^{D}(z,r,\cdot) - \hat{m}_d^{D}(z,r,\cdot) \rVert_\infty = o_p(1), \\
    &\lVert e(z,\cdot) - \hat{e}(z,\cdot) \rVert_\infty = o_p(1),
  \end{align*}
  where for a given function $f$, $\lVert f \rVert_2 := (\E[f(\bX)^2])^{1/2}$ and $\lVert f \rVert_\infty := \sup_{\bx \in \cX} |f(\bx)|$.
\end{assumption}

Assumption~\ref{assum:rate_condition} imposes the standard
product--rate condition for AIPW estimators.  It includes both the
decision and outcome models, since our estimand involves the product
of these two models.  According to this assumption, to establish the
asymptotic normality of the proposed AIPW estimator, it is sufficient
(though not necessary) to require, for example, the quarter--root rate
$\lVert m_d^{D}(z,r,\cdot)-\hat
m_d^{D}(z,r,\cdot)\rVert_2=o_p(n^{-1/4})$,
$\lVert m_y^{Y}(z,r,d,\cdot)-\hat
m_y^{Y}(z,r,d,\cdot)\rVert_2=o_p(n^{-1/4})$, and
$\lVert e(z,\cdot)-\hat e(z,\cdot)\rVert_2=o_p(n^{-1/4})$, which
guarantees that all relevant products of errors are $o_p(n^{-1/2})$.
In a randomized experiment, the propensity score $e(z,\bx)$ is known,
and thus all rate conditions involving
$\lVert e(z,\cdot)-\hat e(z,\cdot)\rVert_2$ can be dropped.

\section{Main Identification Results}

\subsection{Proof of Theorem~\ref{thm:risk_additive}}
\label{proof:risk_additive}
With additive utility, we have
\begin{align*}
    &U(u; D^\ast \mid \bX = \bx)\\
    &=
    \sum_{d=0}^{K_{D}-1} \sum_{y_{0}=0}^{K_{Y}-1} \ldots \sum_{y_{K_{D}-1}=0}^{K_{Y}-1} 
    u(d,\by) \Pr(D^\ast = d, Y(0) = y_{0}, \ldots Y(K_{D}-1) = y_{K_{D}-1} \mid \bX = \bx) \\
    &=
    \sum_{d=0}^{K_{D}-1} \sum_{y_{0}=0}^{K_{Y}-1} \ldots \sum_{y_{K_{D}-1}=0}^{K_{Y}-1} 
    \Bigg(u_{y_d}^{d} 
        + \sum_{\substack{d^\prime=0 \\ d^\prime \ne d}}^{K_{D}-1}\tilde{u}_{y_{d^\prime}}^{d^\prime}
    \Bigg)
    \Pr(D^\ast = d, Y(0) = y_{0}, \ldots Y(K_{D}-1) = y_{K_{D}-1} \mid \bX = \bx) \\
    &=
    \sum_{d=0}^{K_{D}-1} \sum_{y=0}^{K_{Y}-1}
    u_{y}^{d} 
    \Pr(D^\ast = d, Y(d) = y \mid \bX = \bx)
    +
    \sum_{d=0}^{K_{D}-1} 
    \sum_{\substack{d^\prime=0 \\ d^\prime \ne d}}^{K_{D}-1}
    \sum_{y=0}^{K_{Y}-1}
    \tilde{u}_{y}^{d^\prime}
    \Pr(D^\ast = d, Y(d^\prime) = y \mid \bX = \bx).
\end{align*}
By the law of total probability,
\begin{align*}
    &\Pr(D^\ast = d, Y(d^\prime) = y \mid \bX = \bx)\\
    &= 
    \sum_{r=0}^{K_{R}-1}  \sum_{z=0}^{1} \Pr(D^\ast = d, Y(d^\prime) = y \mid \bX = \bx, R = r, Z = z) 
    \Pr(R = r, Z = z \mid \bX = \bx)\\
    &= 
    \sum_{r=0}^{K_{R}-1}  \sum_{z=0}^{1} \Pr(D^\ast = d \mid \bX = \bx, R = r, Z = z) \Pr(Y(d^\prime) = y \mid \bX = \bx, R = r, Z = z) 
    \Pr(R = r, Z = z \mid \bX = \bx)\\
    &= 
    \sum_{r=0}^{K_{R}-1} \Pr(Y(d^\prime) = y \mid \bX = \bx, R = r) \sum_{z=0}^{1}
    \Pr(D^\ast = d \mid \bX = \bx, R = r, Z = z) \Pr(R = r, Z = z \mid \bX = \bx)\\
    &= 
    \sum_{r=0}^{K_{R}-1} \Pr(D^\ast = d, R = r\mid \bX = \bx) \Pr(Y(d^\prime) = y \mid \bX = \bx, R = r) \\
    &= 
    \sum_{r=0}^{K_{R}-1}  \Pr(D^\ast = d, R = r \mid \bX = \bx) \Pr(Y = y \mid D = d^\prime, \bX = \bx, R = r).
\end{align*}
The second equality follows because, conditional on $(\bX,R,Z)$, $D^\ast$ is either degenerate or satisfies
$D^\ast \indep Y(d^\prime)\mid \bX,R,Z$.
For the third equality, Assumption~\ref{assum:single_blinded} (b) implies $Y(d^\prime)\indep Z\mid \bX,R$.
The final equality follows from consistency and
Assumption~\ref{assum:indep}.

Accordingly,
\begin{align*}
    &\overline{U}(u; D^\ast)\\
    &=\E\Bigg[\sum_{y=0}^{K_{Y}-1} \sum_{d=0}^{K_{D}-1} 
    \sum_{r=0}^{K_{R}-1}  
    u_{y}^{d} 
    \Pr(Y = y \mid D = d, R = r, \bX = \bx)
    \Pr(D^\ast = d, R = r \mid \bX = \bx)
    \\
    &+
    \sum_{y=0}^{K_{Y}-1}
    \sum_{d=0}^{K_{D}-1} 
    \sum_{\substack{d^\prime=0 \\ d^\prime \ne d}}^{K_{D}-1}
    \sum_{r=0}^{K_{R}-1}
    \tilde{u}_{y}^{d^\prime}
    \Pr(Y = y \mid D = d^\prime, R = r, \bX = \bx) 
    \Pr(D^\ast = d, R = r \mid \bX = \bx)
    \Bigg]
    \end{align*}

\subsection{Corollary of Theorem~\ref{thm:risk_additive}}
\begin{corollary}
\label{coro:risk_additive}
Consider an additive counterfactual utility
$u \in \mathcal{U}^{\textsc{Add}}$.  Under
Assumptions~\ref{assum:single_blinded}, \ref{assum:indep}, and \ref{assum:decision_pos}, we can
identify the expected utility of the counterfactual risk assessment
system under the human-alone decision $D(0)$,
human-with-recommendation decision $D(1)$, and recommendation-alone
decision $R$ as follows: 
    \begin{align*}
    &\overline{U}(u; D(z))\\
    &=\E\Bigg[\sum_{y=0}^{K_{Y}-1} \sum_{d=0}^{K_{D}-1} 
    u_{y}^{d} 
    \Pr(Y = y, D = d \mid Z = z, \bX = \bx)\\
    &+
    \sum_{y=0}^{K_{Y}-1}
    \sum_{d=0}^{K_{D}-1} 
    \sum_{\substack{d^\prime=0 \\ d^\prime \ne d}}^{K_{D}-1}
    \sum_{r=0}^{K_{R}-1}
    \tilde{u}_{y}^{d^\prime}
    \Pr(Y = y \mid D = d^\prime, R = r, Z = z, \bX = \bx) \Pr(D = d, R = r \mid Z = z, \bX = \bx)\Bigg]
    \end{align*}
    for $z = 0,1$ and
    \begin{align*}
    &\overline{U}(u; R)\\
    &=
    \E\Bigg[\sum_{y=0}^{K_{Y}-1} \sum_{r=0}^{K_{R}-1}
    u_{y}^{a(r)} \Pr(Y = y \mid D = a(r), R = r, \bX = \bx) \Pr(R = r \mid \bX = \bx)\\
    &+
    \sum_{y=0}^{K_{Y}-1} \sum_{r=0}^{K_{R}-1} 
    \sum_{\substack{d^\prime=0 \\ d^\prime \ne a(r)}}^{K_{D}-1}
    \tilde{u}_{y}^{d^\prime}
    \Pr(Y = y \mid D = d^\prime, R = r, \bX = \bx) \Pr(R = r  \mid \bX = \bx)\Bigg].
\end{align*}
\end{corollary}

\begin{proof}
    With additive utility, we have
\begin{align*}
    &U(u; D(z)\mid \bX = \bx)\\
    &=
    \sum_{d=0}^{K_{D}-1} \sum_{y_{0}=0}^{K_{Y}-1} \ldots \sum_{y_{K_{D}-1}=0}^{K_{Y}-1} 
    u(d,\by) \Pr(D(z) = d, Y(0) = y_{0}, \ldots Y(K_{D}-1) = y_{K_{D}-1} \mid \bX = \bx) \\
    &=
    \sum_{d=0}^{K_{D}-1} \sum_{y_{0}=0}^{K_{Y}-1} \ldots \sum_{y_{K_{D}-1}=0}^{K_{Y}-1} 
    \Bigg(u_{y_d}^{d} 
        + \sum_{\substack{d^\prime=0 \\ d^\prime \ne d}}^{K_{D}-1}\tilde{u}_{y_{d^\prime}}^{d^\prime}
    \Bigg)
    \Pr(D(z) = d, Y(0) = y_{0}, \ldots Y(K_{D}-1) = y_{K_{D}-1} \mid \bX = \bx) \\
    &=
    \sum_{d=0}^{K_{D}-1} \sum_{y=0}^{K_{Y}-1}
    u_{y}^{d} 
    \Pr(D(z) = d, Y(d) = y \mid \bX = \bx)
    +
    \sum_{d=0}^{K_{D}-1} 
    \sum_{\substack{d^\prime=0 \\ d^\prime \ne d}}^{K_{D}-1}
    \sum_{y=0}^{K_{Y}-1}
    \tilde{u}_{y}^{d^\prime}
    \Pr(D(z) = d, Y(d^\prime) = y \mid \bX = \bx).
\end{align*}
By Assumption~\ref{assum:single_blinded} and consistency,
\begin{align*}
    \Pr(D(z) = d, Y(d) = y \mid \bX = \bx)
    &= \Pr(D(z) = d, Y(d) = y \mid \bX = \bx, Z = z)\\
    &= \Pr(D = d, Y = y \mid \bX = \bx, Z = z).
\end{align*}
We also have
\begin{align*}
&\Pr(D(z) = d, Y(d^\prime) = y \mid \bX = \bx) \\
&=\Pr(D(z) = d, Y(d^\prime) = y \mid \bX = \bx, Z = z) \\
&=\sum_{r=0}^{K_{R}-1}\Pr(R = r, D(z) = d, Y(d^\prime) = y \mid \bX = \bx, Z = z)\\
&=\sum_{r=0}^{K_{R}-1}\Pr(R = r, D = d, Y(d^\prime) = y \mid \bX = \bx, Z = z)\\
&=\sum_{r=0}^{K_{R}-1}\Pr(Y(d^\prime) = y \mid R = r, D = d, \bX = \bx, Z = z) \Pr(R = r, D = d \mid \bX = \bx, Z = z)\\
&=\sum_{r=0}^{K_{R}-1}\Pr(Y(d^\prime) = y \mid R = r, D = d^\prime, \bX = \bx, Z = z) \Pr(R = r, D = d \mid \bX = \bx, Z = z)\\
&=\sum_{r=0}^{K_{R}-1}\Pr(Y = y \mid R = r, D = d^\prime, \bX = \bx, Z = z) \Pr(R = r, D = d \mid \bX = \bx, Z = z)
\end{align*}
where the first equality follows from Assumption~\ref{assum:single_blinded} and the second to the last equality follows from Assumption~\ref{assum:indep}.

Combining these results,
\begin{align*}
    &U(u; D(z)\mid \bX = \bx)\\
    &=\sum_{d=0}^{K_{D}-1} \sum_{y=0}^{K_{Y}-1}
    u_{y}^{d} 
    \Pr(D = d, Y = y \mid \bX = \bx, Z = z)\\
    &+
    \sum_{d=0}^{K_{D}-1} 
    \sum_{\substack{d^\prime=0 \\ d^\prime \ne d}}^{K_{D}-1}
    \sum_{r=0}^{K_{R}-1}
    \sum_{y=0}^{K_{Y}-1}
    \tilde{u}_{y}^{d^\prime}
    \Pr(Y = y \mid R = r, D = d^\prime, \bX = \bx, Z = z) \Pr(R = r, D = d \mid \bX = \bx, Z = z).
\end{align*}

Similarly, for the utility under $R$ with additive utility:
\begin{align*}
    &U(u; R \mid \bX = \bx)\\
    &=
    \sum_{r=0}^{K_{R}-1} \sum_{y_{0}=0}^{K_{Y}-1} \ldots \sum_{y_{K_{D}-1}=0}^{K_{Y}-1} u(a(r),\by) \Pr(R = r, Y(0) = y_{0}, \ldots Y(K_{D}-1) = y_{K_{D}-1} \mid \bX = \bx) \\
    &=
    \sum_{r=0}^{K_{R}-1} \sum_{y_{0}=0}^{K_{Y}-1} \ldots \sum_{y_{K_{D}-1}=0}^{K_{Y}-1}
    \Bigg(u_{y_{a(r)}}^{a(r)} 
        + \sum_{\substack{d^\prime=0 \\ d^\prime \ne a(r)}}^{K_{D}-1}\tilde{u}_{y_{d^\prime}}^{d^\prime}
    \Bigg)
    \Pr(R = r, Y(0) = y_{0}, \ldots Y(K_{D}-1) = y_{K_{D}-1} \mid \bX = \bx) \\
    &=
    \sum_{r=0}^{K_{R}-1} \sum_{y=0}^{K_{Y}-1}
    u_{y}^{a(r)} \Pr(R = r, Y(a(r)) = y \mid \bX = \bx)
    +
    \sum_{r=0}^{K_{R}-1}
    \sum_{\substack{d^\prime=0 \\ d^\prime \ne a(r)}}^{K_{D}-1}
    \sum_{y=0}^{K_{Y}-1}
    \tilde{u}_{y}^{d^\prime}
    \Pr(R = r, Y(d^\prime) = y \mid \bX = \bx). 
\end{align*}

Observe that
\begin{align*}
&\Pr(R = r, Y(a(r)) = y \mid \bX = \bx) \\
&=\sum_{z=0}^{1}\Pr(Z = z, R = r, Y(a(r)) = y \mid \bX = \bx) \\
&=\sum_{z=0}^{1}\Pr(Y(a(r)) = y \mid R = r, \bX = \bx, Z = z) \Pr(R = r, Z = z \mid \bX = \bx)\\
&=\sum_{z=0}^{1}\Pr(Y(a(r)) = y \mid D = a(r), R = r, \bX = \bx, Z = z) \Pr(R = r, Z = z \mid \bX = \bx)\\
&=\Pr(Y = y \mid D = a(r), R = r, \bX = \bx) \Pr(R = r \mid \bX = \bx)
\end{align*}
where the third equality follows from Assumption~\ref{assum:indep}.
Similarly,
\begin{align*}
&\Pr(R = r, Y(d^\prime) = y \mid \bX = \bx) \\
&=\sum_{z=0}^{1}\Pr(Z = z, R = r, Y(d^\prime) = y \mid \bX = \bx) \\
&=\sum_{z=0}^{1}\Pr(Y(d^\prime) = y \mid R = r, \bX = \bx, Z = z) \Pr(R = r, Z = z \mid \bX = \bx)\\
&=\sum_{z=0}^{1}\Pr(Y(d^\prime) = y \mid D = d^\prime, R = r, \bX = \bx, Z = z) \Pr(R = r, Z = z \mid \bX = \bx)\\
&=\Pr(Y = y \mid D = d^\prime, R = r, \bX = \bx) \Pr(R = r \mid \bX = \bx).
\end{align*}

Accordingly,
\begin{align*}
    &U(u; R \mid \bX = \bx)\\
    &=
    \sum_{r=0}^{K_{R}-1} \sum_{y=0}^{K_{Y}-1}
    u_{y}^{a(r)} \Pr(Y = y \mid D = a(r), R = r, \bX = \bx) \Pr(R = r \mid \bX = \bx)\\
    &+
    \sum_{r=0}^{K_{R}-1} 
    \sum_{\substack{d^\prime=0 \\ d^\prime \ne a(r)}}^{K_{D}-1}
    \sum_{y=0}^{K_{Y}-1} 
    \tilde{u}_{y}^{d^\prime}
    \Pr(Y = y \mid D = d^\prime, R = r, \bX = \bx) \Pr(R = r \mid \bX = \bx).
\end{align*}
\end{proof}

\subsection{Identification of the Difference in Expected Utility}

\begin{corollary}[Identification of the difference in expected utility]
Consider $u \in \mathcal{U}^{\textsc{Add}}$.
    Under Assumptions~\ref{assum:single_blinded}, \ref{assum:indep}, and \ref{assum:decision_pos},
    we can identify the difference in expected utility between decision-making systems as follows:
\begin{align*}
    &U(u; D(1)\mid \bX = \bx)
    -
    U(u; D(0)\mid \bX = \bx)
    \\
    &=
    \sum_{d=0}^{K_{D}-1} \sum_{y=0}^{K_{Y}-1}
    u_{y}^{d}
    \{\Pr(D = d, Y = y \mid \bX = \bx, Z = 1)
    -\Pr(D = d, Y = y \mid \bX = \bx, Z = 0)\} \\
    &\quad+\sum_{d=0}^{K_{D}-1} \sum_{\substack{d^\prime=0 \\ d^\prime \ne d}}^{K_{D}-1}\sum_{y=0}^{K_{Y}-1}\sum_{r=0}^{K_{R}-1}
    \tilde{u}_{y}^{d^\prime}
    \{\Pr(Y = y \mid R = r, D = d^\prime, \bX = \bx, Z = 1) \Pr(R = r, D = d \mid \bX = \bx, Z = 1)\\
    &\hspace{10em}-\Pr(Y = y \mid R = r, D = d^\prime, \bX = \bx, Z = 0) \Pr(R = r, D = d \mid \bX = \bx, Z = 0)\}
\end{align*}
and
\begin{align*}
    &U(u; D(z)\mid \bX = \bx)
    -
    U(u; R\mid \bX = \bx)
    \\
    &=
    \sum_{r=0}^{K_{R}-1} 
    \sum_{d=0}^{K_{D}-1}
    \sum_{y=0}^{K_{Y}-1}
    \Big[
    u_{y}^{d}\Pr(Y = y \mid R = r, D = d, \bX = \bx, Z = z)\\
    &\hspace{8em}
    -u_{y}^{a(r)}\Pr(Y = y \mid R = r, D = a(r), \bX = \bx, Z = z)
    \Big]\Pr(R = r, D = d \mid \bX = \bx, Z = z)
    \\
    &\quad+
    \sum_{r=0}^{K_{R}-1} 
    \sum_{d=0}^{K_{D}-1}
    \sum_{\substack{d^\prime=0 \\ d^\prime \ne d}}^{K_{D}-1} \sum_{y=0}^{K_{Y}-1} 
    \tilde{u}_{y}^{d^\prime}
    \Pr(Y = y \mid R = r, D = d^\prime, \bX = \bx, Z = z)
    \Pr(R = r, D = d \mid \bX = \bx, Z = z) \\
    &\quad-
    \sum_{r=0}^{K_{R}-1} 
    \sum_{d=0}^{K_{D}-1}
    \sum_{\substack{d^\prime=0 \\ d^\prime \ne a(r)}}^{K_{D}-1} \sum_{y=0}^{K_{Y}-1}
    \tilde{u}_{y}^{d^\prime}
    \Pr(Y = y \mid R = r, D = d^\prime, \bX = \bx, Z = z)
    \Pr(R = r, D = d \mid \bX = \bx, Z = z).
\end{align*}
\end{corollary}
\begin{proof}
  \begin{align*}
    &U(u; D(1)\mid \bX = \bx)
    -
    U(u; D(0)\mid \bX = \bx)
    \\
    &=
    \sum_{d=0}^{K_{D}-1} \sum_{y_{0}=0}^{K_{Y}-1} \ldots \sum_{y_{K_{D}-1}=0}^{K_{Y}-1} u(d,\by) \{\Pr(D(1) = d, Y(0) = y_{0}, \ldots Y(K_{D}-1) = y_{K_{D}-1} \mid \bX = \bx) \\
    &\hspace{15em}-\Pr(D(0) = d, Y(0) = y_{0}, \ldots Y(K_{D}-1) = y_{K_{D}-1} \mid \bX = \bx)\}\\
    &=
    \sum_{d=0}^{K_{D}-1} \sum_{y=0}^{K_{Y}-1}
    u_{y}^{d}
    \{\Pr(D = d, Y = y \mid \bX = \bx, Z = 1)
    -\Pr(D = d, Y = y \mid \bX = \bx, Z = 0)\} \\
    &\quad+\sum_{d=0}^{K_{D}-1} \sum_{\substack{d^\prime=0 \\ d^\prime \ne d}}^{K_{D}-1}\sum_{y=0}^{K_{Y}-1}\sum_{r=0}^{K_{R}-1}
    \tilde{u}_{y}^{d^\prime}
    \{\Pr(Y = y \mid R = r, D = d^\prime, \bX = \bx, Z = 1) \Pr(R = r, D = d \mid \bX = \bx, Z = 1)\\
    &\hspace{10em}-\Pr(Y = y \mid R = r, D = d^\prime, \bX = \bx, Z = 0) \Pr(R = r, D = d \mid \bX = \bx, Z = 0)\}
\end{align*}
which is straightforward from the proof of Theorem~\ref{thm:risk_additive}.

\begin{align*}
    &U(u; D(z)\mid \bX = \bx)
    -
    U(u; R\mid \bX = \bx)
    \\
    &=
    \sum_{d=0}^{K_{D}-1} \sum_{y_{0}=0}^{K_{Y}-1} \ldots \sum_{y_{K_{D}-1}=0}^{K_{Y}-1} u(d,\by) \Pr(D(z) = d, Y(0) = y_{0}, \ldots Y(K_{D}-1) = y_{K_{D}-1} \mid \bX = \bx) \\
    &\quad-
    \sum_{r=0}^{K_{R}-1} \sum_{y_{0}=0}^{K_{Y}-1} \ldots \sum_{y_{K_{D}-1}=0}^{K_{Y}-1} u(a(r),\by) \Pr(R = r, Y(0) = y_{0}, \ldots Y(K_{D}-1) = y_{K_{D}-1} \mid \bX = \bx) \\
    &=
    \sum_{r=0}^{K_{R}-1} \sum_{d=0}^{K_{D}-1} \sum_{y_{0}=0}^{K_{Y}-1} \ldots \sum_{y_{K_{D}-1}=0}^{K_{Y}-1} u(d,\by) \Pr(R = r, D(z) = d, Y(0) = y_{0}, \ldots Y(K_{D}-1) = y_{K_{D}-1} \mid \bX = \bx) \\
    &\quad-
    \sum_{r=0}^{K_{R}-1} \sum_{d=0}^{K_{D}-1} \sum_{y_{0}=0}^{K_{Y}-1} \ldots \sum_{y_{K_{D}-1}=0}^{K_{Y}-1} u(a(r),\by) \Pr(R = r, D(z) = d, Y(0) = y_{0}, \ldots Y(K_{D}-1) = y_{K_{D}-1} \mid \bX = \bx)  \\
    &=
    \sum_{r=0}^{K_{R}-1} \sum_{d=0}^{K_{D}-1} \sum_{y_{0}=0}^{K_{Y}-1} \ldots \sum_{y_{K_{D}-1}=0}^{K_{Y}-1} 
    \{u(d,\by)-u(a(r),\by)\}
    \Pr(R = r, D(z) = d, Y(0) = y_{0}, \ldots Y(K_{D}-1) = y_{K_{D}-1} \mid \bX = \bx) \\
    &=
    \sum_{r=0}^{K_{R}-1} \sum_{d=0}^{K_{D}-1} \sum_{y=0}^{K_{Y}-1}
    \Big[
    u_{y}^{d}\Pr(R = r, D(z) = d, Y(d) = y \mid \bX = \bx)\\
    &\hspace{8em}
    -u_{y}^{a(r)}\Pr(R = r, D(z) = d, Y(a(r)) = y \mid \bX = \bx)
    \Big]\\
    &\quad+
    \sum_{r=0}^{K_{R}-1} \sum_{d=0}^{K_{D}-1}
    \sum_{\substack{d^\prime=0 \\ d^\prime \ne d}}^{K_{D}-1} \sum_{y=0}^{K_{Y}-1}
    \tilde{u}_{y}^{d^\prime}
    \Pr(R = r, D(z) = d, Y(d^\prime) = y \mid \bX = \bx)\\
    &\quad-
    \sum_{r=0}^{K_{R}-1} \sum_{d=0}^{K_{D}-1}
    \sum_{\substack{d^\prime=0 \\ d^\prime \ne a(r)}}^{K_{D}-1} \sum_{y=0}^{K_{Y}-1}
    \tilde{u}_{y}^{d^\prime}
    \Pr(R = r, D(z) = d, Y(d^\prime) = y \mid \bX = \bx).
\end{align*}
The last equality follows from additive utility.

By Assumption~\ref{assum:single_blinded}, consistency, and
Assumption~\ref{assum:indep}, for any $e\in\cD$,
\begin{align*}
    &\Pr(R = r, D(z) = d, Y(e) = y \mid \bX = \bx)\\
    &=\Pr(Y = y \mid R = r, D = e, \bX = \bx, Z = z)
    \Pr(R = r, D = d \mid \bX = \bx, Z = z).
\end{align*}
Accordingly,
\begin{align*}
    &U(u; D(z)\mid \bX = \bx)
    -
    U(u; R\mid \bX = \bx)
    \\
    &=
    \sum_{r=0}^{K_{R}-1} \sum_{d=0}^{K_{D}-1} \sum_{y=0}^{K_{Y}-1}
    \Big[
    u_{y}^{d}\Pr(Y = y \mid R = r, D = d, \bX = \bx, Z = z)\\
    &\hspace{8em}
    -u_{y}^{a(r)}\Pr(Y = y \mid R = r, D = a(r), \bX = \bx, Z = z)
    \Big]\Pr(R = r, D = d \mid \bX = \bx, Z = z)\\
    &\quad+
    \sum_{r=0}^{K_{R}-1} \sum_{d=0}^{K_{D}-1}
    \sum_{\substack{d^\prime=0 \\ d^\prime \ne d}}^{K_{D}-1} \sum_{y=0}^{K_{Y}-1}
    \tilde{u}_{y}^{d^\prime}
    \Pr(Y = y \mid R = r, D = d^\prime, \bX = \bx, Z = z)
    \Pr(R = r, D = d \mid \bX = \bx, Z = z)\\
    &\quad-
    \sum_{r=0}^{K_{R}-1} \sum_{d=0}^{K_{D}-1}
    \sum_{\substack{d^\prime=0 \\ d^\prime \ne a(r)}}^{K_{D}-1} \sum_{y=0}^{K_{Y}-1}
    \tilde{u}_{y}^{d^\prime}
    \Pr(Y = y \mid R = r, D = d^\prime, \bX = \bx, Z = z)
    \Pr(R = r, D = d \mid \bX = \bx, Z = z)
\end{align*}
\end{proof}

\begin{remark}\label{remark:utility_difference}
Under a binary decision setting, or under a multi-valued decision with a restricted form of
additive utility satisfying $\tilde u_y^{d}=\tilde u_y$ for all
$d=0,\ldots,K_{D}-1$ (i.e., regret depends only on the counterfactual outcome and not on the
alternative decision), the difference in expected utility between the Human-with-AI and
Human-alone decision rules is identified under Assumption~\ref{assum:single_blinded} alone.

To see this, recall that
\begin{align*}
 &U(u; D(1)\mid \bX = \bx)
    -
    U(u; D(0)\mid \bX = \bx)
    \\
    &=
    \sum_{d=0}^{K_{D}-1} \sum_{y=0}^{K_{Y}-1}
    u_{y}^{d}
    \{\Pr(D(1) = d, Y(d) = y \mid \bX = \bx)
    -\Pr(D(0) = d, Y(d) = y \mid \bX = \bx)\} \\
    &\quad+
    \sum_{d=0}^{K_{D}-1}
    \sum_{\substack{d^\prime=0 \\ d^\prime \ne d}}^{K_{D}-1}
    \sum_{y=0}^{K_{Y}-1}
    \tilde{u}_{y}^{d^\prime} \{\Pr(D(1) = d, Y(d^\prime) = y \mid \bX = \bx)
    -\Pr(D(0) = d, Y(d^\prime) = y \mid \bX = \bx)\}
\end{align*}
Under Assumption~\ref{assum:single_blinded} and consistency, the first term is directly
identified from observed data as
\[
\sum_{d,y} u_y^d\{\Pr(D=d,Y=y\mid\bX=\bx,Z=1)-\Pr(D=d,Y=y\mid\bX=\bx,Z=0)\}.
\]

\paragraph{Binary decision case.}
When $K_{D}=2$, the second term can also be identified using the law of total probability with Assumption~\ref{assum:single_blinded} (cf. Theorem~1 of \citealp{ben2024does}). For example,
\begin{align*}
&\Pr(D(1)=0,Y(1)=y\mid\bX=\bx)-\Pr(D(0)=0,Y(1)=y\mid\bX=\bx) \\
&=\Pr(D=1,Y=y\mid\bX=\bx,Z=0)-\Pr(D=1,Y=y\mid\bX=\bx,Z=1),
\end{align*}
since
\begin{align*}
\sum_{d=0}^1 \Pr(D(1)=d,Y(1)=y\mid\bX=\bx)
=
\sum_{d=0}^1 \Pr(D(0)=d,Y(1)=y\mid\bX=\bx).
\end{align*}
An analogous identity holds for $\Pr(D(1)=1,Y(0)=y)-\Pr(D(0)=1,Y(0)=y)$.
Hence, in the binary decision case, the utility difference is identified without requiring
unconfoundedness of the decision.

\paragraph{Failure for $K_{D}\ge 3$.}
This argument does not extend to settings with three or more decision levels.
For instance, when $K_{D}=3$,
\begin{align*}
&\Pr(D(1)=2,Y(0)=y\mid\bX=\bx)-\Pr(D(0)=2,Y(0)=y\mid\bX=\bx) \\
&=\sum_{d=0}^1
\{\Pr(D(0)=d,Y(0)=y\mid\bX=\bx) 
-\Pr(D(1)=d,Y(0)=y\mid\bX=\bx)\},
\end{align*}
where only the component involving $d=0$ is identified from observed data; the term for
$d=1$ is not identified without additional assumptions.

\paragraph{Identification under a special case of utility.}
Now suppose $\tilde u_y^{d}=\tilde u_y$ for all $d$.
In this case, identification requires only the aggregate counterfactual contrast
\begin{align*}
\sum_{d=0}^{K_{D}-1}
    \sum_{\substack{d^\prime=0 \\ d^\prime \ne d}}^{K_{D}-1}
\{\Pr(D(1)=d,Y(d^\prime)=y\mid\bX=\bx)-\Pr(D(0)=d,Y(d^\prime)=y\mid\bX=\bx)\},
\end{align*}
rather than each pairwise term.
Using the law of total probability and Assumption~\ref{assum:single_blinded},
\begin{align*}
\sum_{d=0}^{K_{D}-1}
    \sum_{\substack{d^\prime=0 \\ d^\prime \ne d}}^{K_{D}-1}
&\{\Pr(D(1)=d,Y(d^\prime)=y\mid\bX=\bx)-\Pr(D(0)=d,Y(d^\prime)=y\mid\bX=\bx)\} \\
&=-
\sum_{d=0}^{K_{D}-1}
\{\Pr(D(1)=d,Y(d)=y\mid\bX=\bx)-\Pr(D(0)=d,Y(d)=y\mid\bX=\bx)\}.
\end{align*}
Substituting this expression yields
\begin{align*}
&U(u;D(1)\mid\bX=\bx)-U(u;D(0)\mid\bX=\bx) \\
&=
\sum_{d=0}^{K_{D}-1}\sum_{y=0}^{K_{Y}-1}
(u_y^d-\tilde u_y)
\{\Pr(D=d,Y=y\mid\bX=\bx,Z=1)-\Pr(D=d,Y=y\mid\bX=\bx,Z=0)\}.
\end{align*}
\end{remark}

\subsection{Proof of Theorem~\ref{thm:utility_generic}}
\label{proof:utility_generic}
Write $O_i=(Y_i,D_i,R_i,Z_i,\bX_i)$ and $O=(Y,D,R,Z,\bX)$.
For fixed $(ykdrz)$, define
\begin{equation*}
  \theta_{ykdrz}
  :=
  \E\left[
  \bbone\{D^\ast=d,R=r\}e(z,\bX)m_y^{Y}(z,r,k,\bX)
  \right].
\end{equation*}
This is the building block for both the observed-decision utility term
$(k=d)$ and the counterfactual utility term $(k=d^\prime)$.  Its
uncentered influence-function term is
\begin{align*}
  &\eta_{ykdrz}(Y,D,R,Z,\bX)\\
  &= \bbone\{D^\ast=d,R=r\}\Bigg\{
    m_y^{Y}(z,r,k,\bX)e(z,\bX)
    +\frac{\bbone\{Z=z,D=k\}}{m_k^{D}(z,r,\bX)}
    \big(\bbone\{Y=y\}-m_y^{Y}(z,r,k,\bX)\big)\\
  &\hspace{1.5in}
    +m_y^{Y}(z,r,k,\bX)\big(\bbone\{Z=z\}-e(z,\bX)\big)
    \Bigg\}.
\end{align*}
The expectation of this term equals $\theta_{ykdrz}$ because,
conditional on $(\bX,R)$, the outcome residual has mean zero among
observations with $(Z,D)=(z,k)$ and
$\E[\bbone\{Z=z\}-e(z,\bX)\mid\bX,R]=0$.
Thus,
\begin{equation*}
  \widehat{\theta}_{ykdrz}
  =
  \frac{1}{n}\sum_{i=1}^{n}
  \widehat{\eta}_{ykdrz}(O_i)
\end{equation*}
estimates
$\Pr(Y=y\mid D=k,R=r,Z=z,\bX=\bx)\Pr(D^\ast=d,R=r,Z=z\mid\bX=\bx)$
after averaging over $\bX$.

The remainder bias for this building block is
\begin{align*}
  &\E[\widehat{\eta}_{ykdrz}-\eta_{ykdrz}]\\
  &=
  \E\Bigg[
  \bbone\{D^\ast=d,R=r\}
  \Bigg\{
  \left(\frac{m_k^D(z,r,\bX)}{\hat m_k^D(z,r,\bX)}-1\right)
  e(z,\bX)
  \big(m_y^Y(z,r,k,\bX)-\hat m_y^Y(z,r,k,\bX)\big)
  \Bigg\}
  \Bigg].
\end{align*}
Therefore, by Assumptions~\ref{assum:decision_pos} and
\ref{assum:rate_condition},
\begin{equation*}
  \left|\E[\widehat{\eta}_{ykdrz}-\eta_{ykdrz}]\right|
  \le
  C
  \left\|m_k^{D}(z,r,\cdot)-\hat m_k^{D}(z,r,\cdot)\right\|_2
  \left\|m_y^{Y}(z,r,k,\cdot)-\hat m_y^{Y}(z,r,k,\cdot)\right\|_2
  =
  o_p(n^{-1/2}).
\end{equation*}

Applying Proposition 2 in \cite{kennedy_semiparametric_2023} and
linearity over the finite sums in the utility functional,
\begin{align*}
  &\widehat{\overline{U}(u; D^\ast)}-\overline{U}(u; D^\ast)\\
  &=
  \frac{1}{n}\sum_{i=1}^{n}
  \sum_{y=0}^{K_Y-1}
  \sum_{d=0}^{K_D-1}
  \sum_{r=0}^{K_R-1}
  \sum_{z=0}^{1}
  \Bigg(
  u_y^d\eta_{yddrz}(O_i)
  +\sum_{\substack{d^\prime=0\\ d^\prime\ne d}}^{K_D-1}
  \tilde u_y^{d^\prime}
  \eta_{yd^\prime drz}(O_i)
  \Bigg)\\
  &\quad
  -\overline{U}(u;D^\ast)+o_p(n^{-1/2}).
\end{align*}
This leads to
\begin{equation*}
  \sqrt{n}\left(\widehat{\overline{U}(u;D^\ast)}-\overline{U}(u;D^\ast)\right)
  \xrightarrow{d}N(0,V),
\end{equation*}
where
\begin{equation*}
  V=
  \E\Bigg[\Bigg\{
  \sum_{y=0}^{K_Y-1}
  \sum_{d=0}^{K_D-1}
  \sum_{r=0}^{K_R-1}
  \sum_{z=0}^{1}
  \Bigg(
  u_y^d\eta_{yddrz}(O)
  +\sum_{\substack{d^\prime=0\\ d^\prime\ne d}}^{K_D-1}
  \tilde u_y^{d^\prime}\eta_{yd^\prime drz}(O)
  \Bigg)
  -\overline{U}(u;D^\ast)
  \Bigg\}^2\Bigg].
\end{equation*}
  
\subsection{Evaluating Human Decision}
\label{utility_human}
We evalute the expected utility of the human decision $D(z)$ for $z \in \{0,1\}$.
By Theorem~\ref{thm:risk_additive}, the expected utility can be expressed as
\begin{align*}
    &\overline{U}(u; D(z))\\
    &=\E\Bigg[\sum_{d=0}^{K_{D}-1}  
    \sum_{y=0}^{K_{Y}-1}
    \sum_{r=0}^{K_{R}-1}
    u_{y}^{d}
   \Pr(Y = y, D = d, R = r \mid Z = z, \bX = \bx)\\
    &\qquad+
    \sum_{d=0}^{K_{D}-1}
    \sum_{\substack{d^\prime=0 \\ d^\prime \ne d}}^{K_{D}-1}
    \sum_{y = 0}^{K_{Y}-1}
    \sum_{r=0}^{K_{R}-1}
    \tilde{u}_{y}^{d^\prime}
    \Pr(Y = y\mid D = d^\prime, \bX = \bx, Z = z, R = r)\\
    &\qquad\hspace{10em}\times\Pr(D = d\mid R = r, \bX = \bx, Z = z) \Pr(R = r\mid \bX = \bx)\Bigg].
\end{align*}
using the law of total expectation and that $R \indep Z \mid \bX$ in
Assumption~\ref{assum:single_blinded} (b).  The displayed AIPW terms
below use the deterministic-recommendation specialization
$m^{R}_{r}(\bX)=\bbone\{R=r\}$, which holds in our application because
$R$ is a known function of $\bX$.  If $R$ is stochastic, the same
construction can instead include the nuisance model
$m^{R}_{r}(\bX) := \Pr(R = r\mid \bX = \bx)$ and its corresponding
augmentation term.

With the above expression, we propose an AIPW estimator using the
following two sets of uncentered influence function estimates: one for
the models of conditional expectations of compound outcomes,
$\bbone\{Y = y, D = d, R = r\}$, and the other for the models of the
product of analogous conditional expectations for counterfactual
utilities,
\begin{equation*}
    \widehat{\overline{U}(u; D(z))} = 
    \frac{1}{n} \sum_{i=1}^{n}\sum_{d=0}^{K_{D}-1}  
    \sum_{y=0}^{K_{Y}-1}
    \sum_{r=0}^{K_{R}-1} \Bigg(
    u_{y}^{d}\widehat{\varphi}_{ydrz}(Y_i,D_i,R_i,Z_i,\bX_i)
    + \sum_{\substack{d^\prime=0 \\ d^\prime \ne d}}^{K_{D}-1}
    \tilde{u}_{y}^{d^\prime}
    \widehat{\widetilde{\varphi}}_{ydd^\prime rz}(Y_i,D_i,R_i,Z_i,\bX_i)
    \Bigg).
  \end{equation*}
where
\begin{align*}
  \widehat{\varphi}_{ydrz}(Y,D,R,Z,\bX)
  =
  \bbone\{R = r\}&\Bigg\{
    \hat{m}_d^{D}(z,r,\bX)\hat{m}_y^{Y}(z,r,d,\bX)\\
  &\quad+ \frac{\bbone\{Z = z, D = d\}}{\hat{e}(z,\bX)}\Big(\bbone\{Y = y\} - \hat{m}_y^{Y}(z,r,d,\bX)\Big)\\
  &\quad+ \hat{m}_y^{Y}(z,r,d,\bX) \frac{\bbone\{Z = z\}}{\hat{e}(z,\bX)}\Big(\bbone\{D = d\} - \hat{m}_d^{D}(z,r,\bX)\Big)
    \Bigg\}, \\
  \widehat{\widetilde{\varphi}}_{ydd^\prime rz}(Y,D,R,Z,\bX)
  =
  \bbone\{R = r\}&\Bigg\{
    \hat{m}_d^{D}(z,r,\bX)\hat{m}_y^{Y}(z,r,d^\prime,\bX)\\
  &\quad+ \frac{\hat{m}_{d}^{D}(z,r,\bX)}{\hat{m}_{d^\prime}^{D}(z,r,\bX)} \frac{\bbone\{Z = z, D = d^\prime\}}{\hat{e}(z,\bX)}\Big(\bbone\{Y = y\} - \hat{m}_y^{Y}(z,r,d^\prime,\bX)\Big)\\
  &\quad+ \hat{m}_y^{Y}(z,r,d^\prime,\bX) \frac{\bbone\{Z = z\}}{\hat{e}(z,\bX)}\Big(\bbone\{D = d\} - \hat{m}_d^{D}(z,r,\bX)\Big)
    \Bigg\}.
\end{align*}

We now discuss the asymptotic properties of the proposed AIPW estimator.
Assumption~\ref{assum:decision_pos} bounds the decision probabilities
away from zero for all decision levels in each $(z,r)$ stratum.  In
particular, it covers both $m_d^{D}(z,r,\bX)$ and
$m_{d^\prime}^{D}(z,r,\bX)$, which appear in the counterfactual
component of the estimator.

\begin{corollary}
  [Asymptotic normality of the AIPW estimator for human decision]
  \label{thm:aipw_normality}
  Under Assumptions~\ref{assum:single_blinded}, \ref{assum:indep}, \ref{assum:decision_pos}, and \ref{assum:rate_condition},
  for each $z \in \{0,1\}$, we have
  \begin{equation*}
    \sqrt{n}(\widehat{\overline{U}(u; D(z))} - \overline{U}(u; D(z))) 
    \xrightarrow{d} 
    N(0, V_z)
  \end{equation*}
  where 
  \begin{equation*}
    V_z = \E\Bigg[\Bigg\{\sum_{d=0}^{K_{D}-1}  
    \sum_{y=0}^{K_{Y}-1}
    \sum_{r=0}^{K_{R}-1}\Bigg(
    u_{y}^{d}\varphi_{ydrz}(Y,D,R,Z,\bX)
    + \sum_{\substack{d^\prime=0 \\ d^\prime \ne d}}^{K_{D}-1}
    \tilde{u}_{y}^{d^\prime}
    \widetilde{\varphi}_{ydd^\prime rz}(Y,D,R,Z,\bX)
    \Bigg) - \overline{U}(u; D(z))\Bigg\}^2\Bigg].
    \end{equation*}
\end{corollary}
\begin{proof}
    We define $\beta_z$ as
  \begin{align*}
    \beta_z &:= 
    \E\Bigg[\sum_{d=0}^{K_{D}-1}  
    \sum_{y=0}^{K_{Y}-1}
    \sum_{r=0}^{K_{R}-1}
    u_{y}^{d}\E[W \mid Z = z, \bX] \\
    &\qquad\qquad+ 
    \sum_{d=0}^{K_{D}-1}  
    \sum_{y=0}^{K_{Y}-1}
    \sum_{r=0}^{K_{R}-1}
    \sum_{\substack{d^\prime=0 \\ d^\prime \ne d}}^{K_{D}-1}\tilde{u}_{y}^{d^\prime}
    \E[\Pr(Y = y\mid D = d^\prime, \bX = \bx, Z = z, R = r)\\
    &\hspace{15em} \times\Pr(D = d\mid R = r, \bX = \bx, Z = z) \Pr(R = r\mid \bX = \bx)]
    \Bigg] \\
    &= \E\Bigg[\sum_{d=0}^{K_{D}-1}  
    \sum_{y=0}^{K_{Y}-1}
    \sum_{r=0}^{K_{R}-1}
    \Bigg(
    u_{y}^{d}m^{Y}_{y} (z,r,d,\bX) m^{D}_{d}(z,r,\bX) m^{R}_{r}(\bX)
    + \sum_{\substack{d^\prime=0 \\ d^\prime \ne d}}^{K_{D}-1}\tilde{u}_{y}^{d^\prime}
    m^{Y}_{y} (z,r,d^\prime,\bX) m^{D}_{d}(z,r,\bX) m^{R}_{r}(\bX)
    \Bigg)
    \Bigg]
  \end{align*}
  where $W := \bbone\{Y = y, D = d, R = r\}$ denote a compound outcome.
  We propose AIPW estimators for the two terms, $m^{Y}_{y} (z,r,d,\bX) m^{D}_{d}(z,r,\bX) m^{R}_{r}(\bX)$
  and $m^{Y}_{y} (z,r,d^\prime,\bX) m^{D}_{d}(z,r,\bX) m^{R}_{r}(\bX)$, separately.
  
  Let the compound outcome model be
  \begin{equation*}
    m(z, \bx) := \E[W \mid Z = z, \bX = \bx] = m^{R}_{r}(\bx) m_d^{D}(z,r,\bx)m_y^{Y}(z,r,d,\bx).
  \end{equation*}
  Then, the uncentered influence-function term is given by
  \begin{equation*}
    \varphi_{ydrz}(W, Z, \bX)_{} := m(z, \bX) 
    + \frac{\bbone\{Z = z\}}{e(z, \bX)}(W - m(z, \bX)).
  \end{equation*}
  Plugging in $W$ and $m(z, \bx)$, we have
  \begin{align*}
    \varphi_{ydrz}(Y,D,R,Z,\bX)
    =
    \bbone\{R = r\}&\Bigg\{
      m_d^{D}(z,r,\bX)m_y^{Y}(z,r,d,\bX)\\
    &\quad+ \frac{\bbone\{Z = z, D = d\}}{e(z,\bX)}(\bbone\{Y = y\} - m_y^{Y}(z,r,d,\bX))\\
    &\quad+ m_y^{Y}(z,r,d,\bX) \frac{\bbone\{Z = z\}}{e(z,\bX)}(\bbone\{D = d\} - m_d^{D}(z,r,\bX))
      \Bigg\}
  \end{align*}
  where we use $m^{R}_{r}(\bX) = \bbone\{R = r\}$ with a slight abuse of notation, since it is a deterministic function of $\bX$.
  The AIPW estimator for the first term is given by
  \begin{equation*}
    \widehat{\E}[\E[W \mid Z = z, \bX]]
    =
    \frac{1}{n} \sum_{i=1}^{n} \widehat{\varphi}_{ydrz}(Y_i,D_i,R_i,Z_i,\bX_i).
  \end{equation*}

  We now examine the remainder bias of this AIPW estimator.
  \begin{align*}
    &\E[\widehat{\varphi}_{ydrz} - \varphi_{ydrz}]\\
    &= \E\Bigg[m_r^{R}(\bX)
    \Bigg\{
      \hat{m}_d^{D}(z,r,\bX)\hat{m}_y^{Y}(z,r,d,\bX) - m_d^{D}(z,r,\bX)m_y^{Y}(z,r,d,\bX)\\
    &\qquad\qquad\qquad\quad+ 
    \Bigg(\frac{e(z,\bX)}{\hat{e}(z,\bX)} - 1 \Bigg)
    m_d^{D}(z,r,\bX)(m_y^{Y}(z,r,d,\bX) - \hat{m}_y^{Y}(z,r,d,\bX))\\
    &\qquad\qquad\qquad\quad+ 
    \hat{m}_y^{Y}(z,r,d,\bX)
    \Bigg(\frac{e(z,\bX)}{\hat{e}(z,\bX)} - 1 \Bigg)
    (m_d^{D}(z,r,\bX) - \hat{m}_d^{D}(z,r,\bX))\\
    &\qquad\qquad\qquad\quad+ m_d^{D}(z,r,\bX)(m_y^{Y}(z,r,d,\bX) - \hat{m}_y^{Y}(z,r,d,\bX))\\
    &\qquad\qquad\qquad\quad+ \hat{m}_y^{Y}(z,r,d,\bX)(m_d^{D}(z,r,\bX) - \hat{m}_d^{D}(z,r,\bX))
      \Bigg\} \\
    &= \E\Bigg[m_r^{R}(\bX)
    \Bigg\{
    \Bigg(\frac{e(z,\bX)}{\hat{e}(z,\bX)} - 1 \Bigg)
    m_d^{D}(z,r,\bX)(m_y^{Y}(z,r,d,\bX) - \hat{m}_y^{Y}(z,r,d,\bX))\\
    &\qquad\qquad\qquad\quad+ 
    \hat{m}_y^{Y}(z,r,d,\bX)
    \Bigg(\frac{e(z,\bX)}{\hat{e}(z,\bX)} - 1 \Bigg)
    (m_d^{D}(z,r,\bX) - \hat{m}_d^{D}(z,r,\bX))
      \Bigg\}
    \Bigg] \\
    &= \E\Bigg[m_r^{R}(\bX)
    \Bigg\{
    \frac{e(z,\bX)-\hat{e}(z,\bX)}{\hat{e}(z,\bX)}
    (m_d^{D}(z,r,\bX)m_y^{Y}(z,r,d,\bX) - \hat{m}_d^{D}(z,r,\bX)\hat{m}_y^{Y}(z,r,d,\bX))
      \Bigg\}
    \Bigg] 
  \end{align*}  
The absolute bias is bounded by
  \begin{equation*}
    \lvert \E[\widehat{\varphi}_{ydrz} - \varphi_{ydrz}] \rvert
    \le 
    C_1(\lVert m_d^{D}(z,r,\cdot) - \hat{m}_d^{D}(z,r,\cdot) \rVert_2
    + \lVert m_y^{Y}(z,r,d,\cdot) - \hat{m}_y^{Y}(z,r,d,\cdot) \rVert_2
    )\times
    \lVert e(z,\cdot) - \hat{e}(z,\cdot) \rVert_2 
  \end{equation*}
By Assumptions~\ref{assum:single_blinded} (c) and \ref{assum:rate_condition}, this is $o_p(n^{-1/2})$.

  Similarly, we can show that the uncentered influence-function term for the second term is given by
  \begin{align*}
    \widetilde{\varphi}_{ydd^\prime rz}(Y,D,R,Z,\bX)
    =
    \bbone\{R = r\}&\Bigg\{
      m_d^{D}(z,r,\bX)m_y^{Y}(z,r,d^\prime,\bX)\\
    &\quad+ \frac{m_{d}^{D}(z,r,\bX)}{m_{d^\prime}^{D}(z,r,\bX)}
   \frac{\bbone\{Z = z, D = d^\prime\}}{e(z,\bX)}(\bbone\{Y = y\} - m_y^{Y}(z,r,d^\prime,\bX))\\
    &\quad+ m_y^{Y}(z,r,d^\prime,\bX) \frac{\bbone\{Z = z\}}{e(z,\bX)}(\bbone\{D = d\} - m_d^{D}(z,r,\bX))
      \Bigg\}
  \end{align*}
  This can be shown by using the product rule illustrated in \cite{kennedy_semiparametric_2023}:
  \begin{align*}
    &\textsc{IF}(m^{R}_{r}(\bX)m_d^{D}(z,r,\bX), m_y^{Y}(z,r,d^\prime,\bX)) \\
    &= m^{R}_{r}(\bX)m_d^{D}(z,r,\bX)\textsc{IF}(m_y^{Y}(z,r,d^\prime,\bX))
    + m_y^{Y}(z,r,d^\prime,\bX)\textsc{IF}(m^{R}_{r}(\bX)m_d^{D}(z,r,\bX)) \\
    &\quad - m^{R}_{r}(\bX)m_d^{D}(z,r,\bX), m_y^{Y}(z,r,d^\prime,\bX)
  \end{align*}
  Note that the uncentered influence-function term for the outcome model is given by
  \begin{align*}
    &\textsc{IF}(m_y^{Y}(z,r,d^\prime,\bX)) \\
    &= 
    m_y^{Y}(z,r,d^\prime,\bX) + \frac{\bbone\{Z = z, D = d^\prime, R = r\}}{\Pr(Z = z, D = d^\prime, R = r \mid \bX)}
    (\bbone\{Y = y\} - m_y^{Y}(z,r,d^\prime,\bX)) \\
    &= m_y^{Y}(z,r,d^\prime,\bX) + \frac{\bbone\{Z = z, D = d^\prime, R = r\}}{m_r^{R}(\bX)m_{d^\prime}^{D}(z,r,\bX)e(z,\bX)}
    (\bbone\{Y = y\} - m_y^{Y}(z,r,d^\prime,\bX)) 
  \end{align*}
  where the last equality uses Assumption~\ref{assum:single_blinded} (b).
  That of the product of AI model and decision model is given by
  \begin{equation*}
    \textsc{IF}(m^{R}_{r}(\bX)m_d^{D}(z,r,\bX))
    = 
    m^{R}_{r}(\bX)m_d^{D}(z,r,\bX) + \frac{\bbone\{Z = z\}}{e(z,\bX)}
    (\bbone\{D = d, R = r\} - m^{R}_{r}(\bX)m_d^{D}(z,r,\bX))
  \end{equation*}
  using the compound outcome $\bbone\{D = d, R = r\}$.
Thus, the AIPW estimator for the second term is given by
  \begin{align*}
    &\widehat{\E}[\E[\Pr(Y = y\mid D = d^\prime, \bX, Z = z, R = r) 
    \Pr(D = d\mid R = r, \bX, Z = z) \Pr(R = r\mid \bX)]]\\
    &=
    \frac{1}{n} \sum_{i=1}^{n} \widehat{\widetilde{\varphi}}_{ydd^\prime rz}(Y_i,D_i,R_i,Z_i,\bX_i).
  \end{align*} 

We now examine the remainder bias in this second AIPW estimator.
  \begin{align*}
    &\E[\widehat{\widetilde{\varphi}}_{ydd^\prime rz} - \widetilde{\varphi}_{ydd^\prime rz}]\\
    &= \E\Bigg[m_r^{R}(\bX)
    \Bigg\{
      \underbrace{\hat{m}_d^{D}(z,r,\bX)\hat{m}_y^{Y}(z,r,d^\prime,\bX) - m_d^{D}(z,r,\bX)m_y^{Y}(z,r,d^\prime,\bX)}_{T_1}\\
    &\qquad\qquad\qquad\quad+ 
    \underbrace{\frac{\hat{m}_{d}^{D}(z,r,\bX)}{\hat{m}_{d^\prime}^{D}(z,r,\bX)}\Bigg(\frac{e(z,\bX)}{\hat{e}(z,\bX)} - 1 \Bigg)
    m_{d^\prime}^{D}(z,r,\bX)(m_y^{Y}(z,r,d^\prime,\bX) - \hat{m}_y^{Y}(z,r,d^\prime,\bX))}_{T_2}\\
    &\qquad\qquad\qquad\quad+ 
    \underbrace{\hat{m}_y^{Y}(z,r,d^\prime,\bX)
    \Bigg(\frac{e(z,\bX)}{\hat{e}(z,\bX)} - 1 \Bigg)
    (m_d^{D}(z,r,\bX) - \hat{m}_d^{D}(z,r,\bX))}_{T_3}\\
    &\qquad\qquad\qquad\quad+\underbrace{\frac{\hat{m}_{d}^{D}(z,r,\bX)}{\hat{m}_{d^\prime}^{D}(z,r,\bX)} m_{d^\prime}^{D}(z,r,\bX)
    (m_y^{Y}(z,r,d^\prime,\bX) - \hat{m}_y^{Y}(z,r,d^\prime,\bX))}_{T_4}\\
    &\qquad\qquad\qquad\quad+\underbrace{\hat{m}_y^{Y}(z,r,d^\prime,\bX)(m_d^{D}(z,r,\bX) - \hat{m}_d^{D}(z,r,\bX))}_{T_5}
      \Bigg\}
  \end{align*} 
  We first combine $T_1$, $T_4$ and $T_5$ as
  \begin{align*}
    &T_1 + T_4 + T_5 \\
    &=
    \Bigg(\frac{\hat{m}_{d}^{D}(z,r,\bX)}{\hat{m}_{d^\prime}^{D}(z,r,\bX)}
    m_{d^\prime}^{D}(z,r,\bX) - m_d^{D}(z,r,\bX)\Bigg)
    (m_y^{Y}(z,r,d^\prime,\bX) - \hat{m}_y^{Y}(z,r,d^\prime,\bX))\\
    &=
    \Bigg\{\frac{m_{d^\prime}^{D}(z,r,\bX)}{\hat{m}_{d^\prime}^{D}(z,r,\bX)}
    (\hat{m}_{d}^{D}(z,r,\bX) - m_d^{D}(z,r,\bX))
    - \frac{m_{d}^{D}(z,r,\bX)}{\hat{m}_{d^\prime}^{D}(z,r,\bX)}
    (\hat{m}_{d^\prime}^{D}(z,r,\bX) - m_{d^\prime}^{D}(z,r,\bX))\Bigg\}\\
    &\qquad \times
    (m_y^{Y}(z,r,d^\prime,\bX) - \hat{m}_y^{Y}(z,r,d^\prime,\bX))
  \end{align*}
  Next, combining $T_2$ and $T_3$, we have
  \begin{align*}
    &T_2 + T_3 \\
    &=
    \Bigg(\frac{e(z,\bX)}{\hat{e}(z,\bX)} - 1 \Bigg)
    \Bigg(\frac{\hat{m}_{d}^{D}(z,r,\bX)}{\hat{m}_{d^\prime}^{D}(z,r,\bX)}
    m_{d^\prime}^{D}(z,r,\bX) - m_d^{D}(z,r,\bX)\Bigg)
    (m_y^{Y}(z,r,d^\prime,\bX) - \hat{m}_y^{Y}(z,r,d^\prime,\bX))\\
    &\quad+
    \Bigg(\frac{e(z,\bX)}{\hat{e}(z,\bX)} - 1 \Bigg)
    (m_d^{D}(z,r,\bX)m_y^{Y}(z,r,d^\prime,\bX) - \hat{m}_d^{D}(z,r,\bX)\hat{m}_y^{Y}(z,r,d^\prime,\bX)) \\
    &=\Bigg\{\frac{m_{d^\prime}^{D}(z,r,\bX)}{\hat{m}_{d^\prime}^{D}(z,r,\bX)}
    (\hat{m}_{d}^{D}(z,r,\bX) - m_d^{D}(z,r,\bX))
    - \frac{m_{d}^{D}(z,r,\bX)}{\hat{m}_{d^\prime}^{D}(z,r,\bX)}
    (\hat{m}_{d^\prime}^{D}(z,r,\bX) - m_{d^\prime}^{D}(z,r,\bX))\Bigg\}\\
    &\qquad \times
    \frac{e(z,\bX)-\hat{e}(z,\bX)}{\hat{e}(z,\bX)}(m_y^{Y}(z,r,d^\prime,\bX) - \hat{m}_y^{Y}(z,r,d^\prime,\bX))\\
    &\quad+
    \frac{e(z,\bX)-\hat{e}(z,\bX)}{\hat{e}(z,\bX)} 
    (m_d^{D}(z,r,\bX)m_y^{Y}(z,r,d^\prime,\bX) - \hat{m}_d^{D}(z,r,\bX)\hat{m}_y^{Y}(z,r,d^\prime,\bX))
  \end{align*}
  Therefore, we can write
  \begin{align*}
    &\E[\widehat{\widetilde{\varphi}}_{ydd^\prime rz} - \widetilde{\varphi}_{ydd^\prime rz}]\\
    &= \E\Bigg[m_r^{R}(\bX)
    \Bigg[
    \Bigg\{\frac{m_{d^\prime}^{D}(z,r,\bX)}{\hat{m}_{d^\prime}^{D}(z,r,\bX)}
    (\hat{m}_{d}^{D}(z,r,\bX) - m_d^{D}(z,r,\bX))
    - \frac{m_{d}^{D}(z,r,\bX)}{\hat{m}_{d^\prime}^{D}(z,r,\bX)}
    (\hat{m}_{d^\prime}^{D}(z,r,\bX) - m_{d^\prime}^{D}(z,r,\bX))\Bigg\}\\
    &\qquad \times
    (m_y^{Y}(z,r,d^\prime,\bX) - \hat{m}_y^{Y}(z,r,d^\prime,\bX))\\
    &\quad+
    \Bigg\{\frac{m_{d^\prime}^{D}(z,r,\bX)}{\hat{m}_{d^\prime}^{D}(z,r,\bX)}
    (\hat{m}_{d}^{D}(z,r,\bX) - m_d^{D}(z,r,\bX))
    - \frac{m_{d}^{D}(z,r,\bX)}{\hat{m}_{d^\prime}^{D}(z,r,\bX)}
    (\hat{m}_{d^\prime}^{D}(z,r,\bX) - m_{d^\prime}^{D}(z,r,\bX))\Bigg\}\\
    &\qquad \times
    \frac{e(z,\bX)-\hat{e}(z,\bX)}{\hat{e}(z,\bX)}(m_y^{Y}(z,r,d^\prime,\bX) - \hat{m}_y^{Y}(z,r,d^\prime,\bX))\\
    &\quad+
    \frac{e(z,\bX)-\hat{e}(z,\bX)}{\hat{e}(z,\bX)}
    (m_d^{D}(z,r,\bX)m_y^{Y}(z,r,d^\prime,\bX) - \hat{m}_d^{D}(z,r,\bX)\hat{m}_y^{Y}(z,r,d^\prime,\bX))
    \Bigg]
    \Bigg]
  \end{align*}

  Thus, the absolute bias is bounded by
  \begin{align*}
    \lvert \E[\widehat{\widetilde{\varphi}}_{ydd^\prime rz} - \widetilde{\varphi}_{ydd^\prime rz}] \rvert 
    \le C_2 &(
    \lVert m_y^{Y}(z,r,d^\prime,\cdot) - \hat{m}_y^{Y}(z,r,d^\prime,\cdot) \rVert_2
   \times \lVert m_d^{D}(z,r,\cdot) - \hat{m}_d^{D}(z,r,\cdot) \rVert_2 \\
   &\quad +
   \lVert m_y^{Y}(z,r,d^\prime,\cdot) - \hat{m}_y^{Y}(z,r,d^\prime,\cdot) \rVert_2
   \times \lVert m_{d^\prime}^{D}(z,r,\cdot) - \hat{m}_{d^\prime}^{D}(z,r,\cdot) \rVert_2 \\
    &\quad +
    \lVert m_d^{D}(z,r,\cdot) - \hat{m}_d^{D}(z,r,\cdot) \rVert_2
    \times 
   \lVert e(z,\cdot) - \hat{e}(z,\cdot) \rVert_2 \\
   &\quad +
   \lVert m_y^{Y}(z,r,d^\prime,\cdot) - \hat{m}_y^{Y}(z,r,d^\prime,\cdot) \rVert_2
    \times 
   \lVert e(z,\cdot) - \hat{e}(z,\cdot) \rVert_2
    )
\end{align*}
This is $o_p(n^{-1/2})$ by Assumptions~\ref{assum:single_blinded} (c), \ref{assum:decision_pos}, and \ref{assum:rate_condition}.

Applying Proposition 2 in \cite{kennedy_semiparametric_2023},
we can then write
\begin{align*}
  \hat{\beta}_z - \beta_z
  &=
  \frac{1}{n} \sum_{i=1}^{n} \sum_{d=0}^{K_{D}-1}  
    \sum_{y=0}^{K_{Y}-1}
    \sum_{r=0}^{K_{R}-1}\Bigg(
    u_{y}^{d}\varphi_{ydrz}(Y_i,D_i,R_i,Z_i,\bX_i)
    + \sum_{\substack{d^\prime=0 \\ d^\prime \ne d}}^{K_{D}-1}
    \tilde{u}_{y}^{d^\prime}
    \widetilde{\varphi}_{ydd^\prime rz}(Y_i,D_i,R_i,Z_i,\bX_i)
    \Bigg)\\
    &\quad
    - \beta_z
    + o_p(n^{-1/2}) 
\end{align*}
where
\begin{equation*}
    \hat{\beta}_z = 
    \frac{1}{n} \sum_{i=1}^{n}\sum_{d=0}^{K_{D}-1}  
    \sum_{y=0}^{K_{Y}-1}
    \sum_{r=0}^{K_{R}-1} \Bigg(
    u_{y}^{d}\widehat{\varphi}_{ydrz}(Y_i,D_i,R_i,Z_i,\bX_i)
    + \sum_{\substack{d^\prime=0 \\ d^\prime \ne d}}^{K_{D}-1}
    \tilde{u}_{y}^{d^\prime}
    \widehat{\widetilde{\varphi}}_{ydd^\prime rz}(Y_i,D_i,R_i,Z_i,\bX_i)
    \Bigg).
  \end{equation*}
  This leads to the desired result,
  \begin{equation*}
    \sqrt{n}(\hat{\beta}_z - \beta_z) \xrightarrow{d} N(0, V_z)
  \end{equation*}
  where 
  \begin{equation*}
    V_z = \E\Bigg[\Bigg\{\sum_{d=0}^{K_{D}-1}  
    \sum_{y=0}^{K_{Y}-1}
    \sum_{r=0}^{K_{R}-1}\Bigg(
    u_{y}^{d}\varphi_{ydrz}(Y,D,R,Z,\bX)
    + \sum_{\substack{d^\prime=0 \\ d^\prime \ne d}}^{K_{D}-1}
    \tilde{u}_{y}^{d^\prime}
    \widetilde{\varphi}_{ydd^\prime rz}(Y,D,R,Z,\bX)
    \Bigg) - \beta_z\Bigg\}^2\Bigg].
  \end{equation*}

\end{proof}
  
\subsection{Evaluating AI Recommendation}
\label{utility_ai}
Now, we turn to the expected utility of following the AI recommendation $R$, which is
given by
\begin{align*}
    &\overline{U}(u; R)\\
    &=\E\Bigg[
    \sum_{y=0}^{K_{Y}-1}
    \sum_{r=0}^{K_{R}-1}
    \sum_{z=0}^{1}
    u_{y}^{a(r)}
   \Pr(Y = y \mid D = a(r), R = r, Z = z, \bX = \bx)\\
    &\qquad\hspace{10em}\times\Pr(R = r, Z = z\mid \bX = \bx)\\
    &\qquad+
    \sum_{y = 0}^{K_{Y}-1}
    \sum_{r=0}^{K_{R}-1}
    \sum_{\substack{d^\prime=0 \\ d^\prime \ne a(r)}}^{K_{D}-1}
    \sum_{z=0}^{1}
    \tilde{u}_{y}^{d^\prime}
    \Pr(Y = y\mid D = d^\prime, R = r, Z = z, \bX = \bx)\\
    &\qquad\hspace{10em}\times\Pr(R = r, Z = z\mid \bX = \bx)\Bigg].
\end{align*}

Following the same strategy as the one used for the evaluation of
human decision above, we propose an AIPW estimator using the following
two sets of uncentered influence function estimates: one for
$\Pr(Y = y \mid D = a(r), R = r, Z = z, \bX = \bx)\Pr(R = r, Z = z\mid
\bX = \bx)$ and the other for
$\Pr(Y = y\mid D = d^\prime, R = r, Z = z, \bX = \bx)\Pr(R = r, Z =
z\mid \bX = \bx)$.

\begin{equation*}
    \widehat{\overline{U}(u; R)} = 
    \frac{1}{n} \sum_{i=1}^{n}\sum_{y=0}^{K_{Y}-1}
    \sum_{r=0}^{K_{R}-1}
    \sum_{z=0}^{1}
    \Bigg(
    u_{y}^{a(r)} \widehat{\psi}_{yrz}(Y_{i},D_{i},R_{i},Z_{i},\bX_{i})
    + \sum_{\substack{d^\prime=0 \\ d^\prime \ne a(r)}}^{K_{D}-1}
    \tilde{u}_{y}^{d^\prime}
    \widehat{\widetilde{\psi}}_{yrd^\prime z}(Y_{i},D_{i},R_{i},Z_{i},\bX_{i})
    \Bigg),
  \end{equation*}
where
\begin{align*}
    &\widehat{\psi}_{yrz}(Y,D,R,Z,\bX)\\
    &= \bbone\{R = r\} \Bigg\{
      \hat{m}_y^{Y}(z,r,a(r),\bX)\hat{e}(z, \bX)+ \frac{\bbone\{Z = z, D = a(r)\}}{\hat{m}_{a(r)}^{D}(z,r,\bX)}
      (\bbone\{Y = y\} - \hat{m}_y^{Y}(z,r,a(r),\bX)) \\
      &\qquad\qquad\qquad\quad+ \hat{m}_y^{Y}(z,r,a(r),\bX)(\bbone\{Z = z\} - \hat{e}(z,\bX))
      \Bigg\},
  \end{align*}
  and 
  \begin{align*}
    &\widehat{\widetilde{\psi}}_{yrd^\prime z}(Y,D,R,Z,\bX)\\
    &= \bbone\{R = r\} \Bigg\{
      \hat{m}_y^{Y}(z,r,d^\prime,\bX)\hat{e}(z, \bX)+ \frac{\bbone\{Z = z, D = d^\prime\}}{\hat{m}_{d^\prime}^{D}(z,r,\bX)}
      (\bbone\{Y = y\} - \hat{m}_y^{Y}(z,r,d^\prime,\bX)) \\
      &\qquad\qquad\qquad\quad+ \hat{m}_y^{Y}(z,r,d^\prime,\bX)(\bbone\{Z = z\} - \hat{e}(z,\bX))
      \Bigg\}.
  \end{align*}

  The following corollary establishes the asymptotic normality of this AIPW estimator.
  \begin{corollary}
  [Asymptotic normality of the AIPW estimator for AI recommendation]\label{thm:utility_ai}
  Under Assumptions~\ref{assum:single_blinded}, \ref{assum:indep}, \ref{assum:decision_pos}, and \ref{assum:rate_condition},
  we have
  \begin{equation*}
    \sqrt{n}(\widehat{\overline{U}(u; R)} - \overline{U}(u; R)) \xrightarrow{d} N(0, V)
  \end{equation*}
  where
  \begin{equation*}
    V = \E\Bigg[\Bigg\{\sum_{y=0}^{K_{Y}-1}
    \sum_{r=0}^{K_{R}-1}
    \sum_{z=0}^{1}
    \Bigg(
    u_{y}^{a(r)} \psi_{yrz}(Y_{i},D_{i},R_{i},Z_{i},\bX_{i})
    + \sum_{\substack{d^\prime=0 \\ d^\prime \ne a(r)}}^{K_{D}-1}
    \tilde{u}_{y}^{d^\prime}
    \widetilde{\psi}_{yrd^\prime z}(Y_{i},D_{i},R_{i},Z_{i},\bX_{i})
    \Bigg) - \overline{U}(u; R)\Bigg\}^2\Bigg].
  \end{equation*}
\end{corollary}

\begin{proof}
  Recall that the uncentered influence-function term for the outcome model is given by
  \begin{align*}
    &\textsc{IF}(m_y^{Y}(z,r,a(r),\bX)) \\
    &= m_y^{Y}(z,r,a(r),\bX) + \frac{\bbone\{Z = z, D = a(r), R = r\}}{m_r^{R}(\bX)m_{a(r)}^{D}(z,r,\bX)e(z,\bX)}
    (\bbone\{Y = y\} - m_y^{Y}(z,r,a(r),\bX)) \\
    &\textsc{IF}(m_y^{Y}(z,r,d^\prime,\bX)) \\
    &= m_y^{Y}(z,r,d^\prime,\bX) + \frac{\bbone\{Z = z, D = d^\prime, R = r\}}{m_r^{R}(\bX)m_{d^\prime}^{D}(z,r,\bX)e(z,\bX)}
    (\bbone\{Y = y\} - m_y^{Y}(z,r,d^\prime,\bX)) 
  \end{align*}
  That of $\Pr(R = r, Z = z\mid \bX = \bx)$ is given by
  \begin{equation*}
    \textsc{IF}(\Pr(R = r, Z = z\mid \bX = \bx))
    = 
    m^{R}_{r}(\bX)e(z, \bX)
    + e(z,\bX)(\bbone\{R = r\} - m^{R}_{r}(\bX))
    + m^{R}_{r}(\bX)(\bbone\{Z = z\} - e(z,\bX))
  \end{equation*}
  where we use that $\Pr(R = r, Z = z\mid \bX = \bx) = \Pr(R = r \mid Z = z, \bX = \bx)e(z,\bX)$
  and $\Pr(R = r \mid Z = z, \bX = \bx) = m^{R}_{r}(\bX)$ in our setup.
  Let $\psi_{yrz}(Y,D,R,Z,\bX) := \textsc{IF}(m_y^{Y}(z,r,a(r),\bX)m^{R}_{r}(\bX)e(z, \bX))$
  and $\widetilde{\psi}_{yrd^\prime z}(Y,D,R,Z,\bX) := \textsc{IF}(m_y^{Y}(z,r,d^\prime,\bX)m^{R}_{r}(\bX)e(z, \bX))$.
  Using product rule, the centered influence-function term for the first term is given by
  \begin{align*}
    &\psi_{yrz}(Y,D,R,Z,\bX)\\
    &= 
    \textsc{IF}(m_y^{Y}(z,r,a(r),\bX))m^{R}_{r}(\bX)e(z, \bX)
    + \textsc{IF}(m^{R}_{r}(\bX)e(z, \bX))m_y^{Y}(z,r,a(r),\bX) \\
    &\quad- m_y^{Y}(z,r,a(r),\bX)m^{R}_{r}(\bX)e(z, \bX) \\
    &= \bbone\{R = r\} \Bigg\{
      m_y^{Y}(z,r,a(r),\bX)e(z, \bX)+ \frac{\bbone\{Z = z, D = a(r)\}}{m_{a(r)}^{D}(z,r,\bX)}
      (\bbone\{Y = y\} - m_y^{Y}(z,r,a(r),\bX)) \\
      &\qquad\qquad\qquad\quad+ m_y^{Y}(z,r,a(r),\bX)(\bbone\{Z = z\} - e(z,\bX))
      \Bigg\}
  \end{align*}
  and similarly
  \begin{align*}
    &\widetilde{\psi}_{yrd^\prime z}(Y,D,R,Z,\bX)\\
    &= \bbone\{R = r\} \Bigg\{
      m_y^{Y}(z,r,d^\prime,\bX)e(z, \bX)+ \frac{\bbone\{Z = z, D = d^\prime\}}{m_{d^\prime}^{D}(z,r,\bX)}
      (\bbone\{Y = y\} - m_y^{Y}(z,r,d^\prime,\bX)) \\
      &\qquad\qquad\qquad\quad+ m_y^{Y}(z,r,d^\prime,\bX)(\bbone\{Z = z\} - e(z,\bX))
      \Bigg\}
  \end{align*}
  where we use $m^{R}_{r}(\bX) = \bbone\{R = r\}$ with a slight abuse of notation, since it is a deterministic function of $\bX$.
  Thus, we use the following two AIPW estimators for each term:
  \begin{align*}
    &\widehat{\E}[\Pr(Y = y \mid D = a(r), R = r, Z = z, \bX = \bx)\Pr(R = r, Z = z\mid \bX = \bx)]\\
    &=
    \frac{1}{n} \sum_{i=1}^{n} \widehat{\psi}_{yrz}(Y,D,R,Z,\bX)
  \end{align*}
  and
  \begin{align*}
    &\widehat{\E}[\Pr(Y = y\mid D = d^\prime, R = r, Z = z, \bX = \bx)\Pr(R = r, Z = z\mid \bX = \bx)]\\
    &=
    \frac{1}{n} \sum_{i=1}^{n} \widehat{\widetilde{\psi}}_{yrd^\prime z}(Y,D,R,Z,\bX)
  \end{align*}
  We examine the remainder bias in each term.
  \begin{align*}
    &\E[\widehat{\psi}_{yrz}(Y,D,R,Z,\bX) - \psi_{yrz}(Y,D,R,Z,\bX)] \\
    &= \E \Bigg[
      m^{R}_{r}(\bX) \Bigg\{
        \hat{m}_y^{Y}(z,r,a(r),\bX)\hat{e}(z, \bX)
        - m_y^{Y}(z,r,a(r),\bX)e(z, \bX) \\
      &\qquad\qquad\qquad\quad+ 
      \Bigg(\frac{m_{a(r)}^{D}(z,r,\bX)}{\hat{m}_{a(r)}^{D}(z,r,\bX)} - 1 \Bigg)
      e(z,\bX)
      (m_y^{Y}(z,r,a(r),\bX) - \hat{m}_y^{Y}(z,r,a(r),\bX)) \\
      &\qquad\qquad\qquad\quad+ 
      \hat{m}_y^{Y}(z,r,a(r),\bX)
      (e(z,\bX) - \hat{e}(z,\bX))\\
      &\qquad\qquad\qquad\quad+ e(z,\bX)
      (m_y^{Y}(z,r,a(r),\bX) - \hat{m}_y^{Y}(z,r,a(r),\bX)) 
      \Bigg\}
      \Bigg]  \\
    &= \E \Bigg[
      m^{R}_{r}(\bX) \Bigg\{
      \Bigg(\frac{m_{a(r)}^{D}(z,r,\bX)}{\hat{m}_{a(r)}^{D}(z,r,\bX)} - 1 \Bigg)
      e(z,\bX)
      (m_y^{Y}(z,r,a(r),\bX) - \hat{m}_y^{Y}(z,r,a(r),\bX)) 
      \Bigg\}
      \Bigg] 
  \end{align*}

  The absolute bias is bounded by
  \begin{equation*}
    \lvert \E[\widehat{\psi}_{yrz} - \psi_{yrz}] \rvert
    \le 
    C_1\lVert m_{a(r)}^{D}(z,r,\cdot) - \hat{m}_{a(r)}^{D}(z,r,\cdot) \rVert_2
    \times \lVert m_y^{Y}(z,r,a(r),\cdot) - \hat{m}_y^{Y}(z,r,a(r),\cdot) \rVert_2
  \end{equation*}
By Assumptions~\ref{assum:decision_pos} and \ref{assum:rate_condition}, this is $o_p(n^{-1/2})$.

Lastly, the remainder bias in the second term is given by
  \begin{align*}
    &\E[\widehat{\widetilde{\psi}}_{yrd^\prime z}(Y,D,R,Z,\bX) - \widetilde{\psi}_{yrd^\prime z}(Y,D,R,Z,\bX)] \\
    &= \E \Bigg[
      m^{R}_{r}(\bX) \Bigg\{
        \hat{m}_y^{Y}(z,r,d^\prime,\bX)\hat{e}(z, \bX)
        - m_y^{Y}(z,r,d^\prime,\bX)e(z, \bX) \\
      &\qquad\qquad\qquad\quad+ 
      \Bigg(\frac{m_{d^\prime}^{D}(z,r,\bX)}{\hat{m}_{d^\prime}^{D}(z,r,\bX)} - 1 \Bigg)
      e(z,\bX)
      (m_y^{Y}(z,r,d^\prime,\bX) - \hat{m}_y^{Y}(z,r,d^\prime,\bX)) \\
      &\qquad\qquad\qquad\quad+ 
      \hat{m}_y^{Y}(z,r,d^\prime,\bX)
      (e(z,\bX) - \hat{e}(z,\bX))\\
      &\qquad\qquad\qquad\quad+ e(z,\bX)
      (m_y^{Y}(z,r,d^\prime,\bX) - \hat{m}_y^{Y}(z,r,d^\prime,\bX)) 
      \Bigg\}
      \Bigg]  \\
    &= \E \Bigg[
      m^{R}_{r}(\bX) \Bigg\{
      \Bigg(\frac{m_{d^\prime}^{D}(z,r,\bX)}{\hat{m}_{d^\prime}^{D}(z,r,\bX)} - 1 \Bigg)
      e(z,\bX)
      (m_y^{Y}(z,r,d^\prime,\bX) - \hat{m}_y^{Y}(z,r,d^\prime,\bX)) 
      \Bigg\}
      \Bigg] 
  \end{align*}
  The absolute bias is bounded by
  \begin{equation*}
    \lvert \E[\widehat{\widetilde{\psi}}_{yrd^\prime z} - \widetilde{\psi}_{yrd^\prime z}] \rvert
    \le
    C_2\lVert m_{d^\prime}^{D}(z,r,\cdot) - \hat{m}_{d^\prime}^{D}(z,r,\cdot) \rVert_2
    \times \lVert m_y^{Y}(z,r,d^\prime,\cdot) - \hat{m}_y^{Y}(z,r,d^\prime,\cdot) \rVert_2
  \end{equation*}
By Assumptions~\ref{assum:decision_pos} and \ref{assum:rate_condition}, this is $o_p(n^{-1/2})$.

  Applying Proposition 2 in \cite{kennedy_semiparametric_2023},
  we can then write
  \begin{align*}
    &\widehat{\overline{U}(u; R)} - \overline{U}(u; R) \\
    &=
    \frac{1}{n} \sum_{i=1}^{n}\sum_{y=0}^{K_{Y}-1}
    \sum_{r=0}^{K_{R}-1}
    \sum_{z=0}^{1}
    \Bigg(
    u_{y}^{a(r)} \psi_{yrz}(Y_i,D_i,R_i,Z_i,\bX_i)
    + \sum_{\substack{d^\prime=0 \\ d^\prime \ne a(r)}}^{K_{D}-1}
    \tilde{u}_{y}^{d^\prime}
    \widetilde{\psi}_{yrd^\prime z}(Y_i,D_i,R_i,Z_i,\bX_i)
    \Bigg)\\
    &\quad
    - \overline{U}(u; R)
    + o_p(n^{-1/2}) 
  \end{align*}
  This leads to the desired result,
  \begin{equation*}
    \sqrt{n}(\widehat{\overline{U}(u; R)} - \overline{U}(u; R)) \xrightarrow{d} N(0, V)
  \end{equation*}
  where 
  \begin{equation*}
    V = \E\Bigg[\Bigg\{\sum_{y=0}^{K_{Y}-1}
    \sum_{r=0}^{K_{R}-1}
    \sum_{z=0}^{1}
    \Bigg(
    u_{y}^{a(r)} \psi_{yrz}(Y,D,R,Z,\bX)
    + \sum_{\substack{d^\prime=0 \\ d^\prime \ne a(r)}}^{K_{D}-1}
    \tilde{u}_{y}^{d^\prime}
    \widetilde{\psi}_{yrd^\prime z}(Y,D,R,Z,\bX)
    \Bigg) - \overline{U}(u; R)\Bigg\}^2\Bigg].
  \end{equation*} 
\end{proof}

\section{Utility Specification under Binary Case}

\subsection{Additive Counterfactual Utility with Cost of Decision, Outcome, and Regret}
\label{additive_loss}

\begin{table}[H]
\centering 
    \begin{tabular}{l|c|c|c|}
  \multicolumn{2}{c}{} & \multicolumn{2}{c}{\textbf{Decision}} \\\cline{3-4}
  \multicolumn{2}{c|}{} 
    & \multirow{2}{*}{Release $(D^\ast = 0)$} 
    & \multirow{2}{*}{Cash bail $(D^\ast = 1)$} \\
  \multicolumn{2}{c|}{} & & \\ \cline{2-4}

  & \multirow{2}{*}{\shortstack{Safe\\$(Y(0) = 0,\ Y(1) = 0)$}} 
    & \multirow{2}{*}{\shortstack[l]{$-r_{\nocrime}^\cash$\\$=-r_{\nocrime}^\ROR c_{\crime}^\cash$}} 
    & \multirow{2}{*}{$-c^\cash-r_{\nocrime}^\ROR$} \\
  & & & \\ \cline{2-4}

  \multirow{3}{*}{\textbf{Principal}}
  & \multirow{2}{*}{\shortstack{Backlash\\$(Y(0) = 0,\ Y(1) = 1)$}} 
    & \multirow{2}{*}{$0$} 
    & \multirow{2}{*}{$-c^\cash -c_\crime^\cash -r_{\nocrime}^\ROR$} \\
  & & & \\ \cline{2-4}

  \textbf{Strata}
  & \multirow{2}{*}{\shortstack{Preventable\\$(Y(0) = 1,\ Y(1) = 0)$}} 
    & \multirow{2}{*}{\shortstack[l]{$-c_\crime^\ROR -r_{\nocrime}^\cash$\\$=-1 -r_{\nocrime}^\ROR c_\crime^\cash$}}
    & \multirow{2}{*}{$-c^\cash$} \\
  & & & \\ \cline{2-4}

    & \multirow{2}{*}{\shortstack{Hopeless\\$(Y(0) = 1,\ Y(1) = 1)$}} 
    & \multirow{2}{*}{\shortstack[l]{$-c_\crime^\ROR$\\$=-1$}} 
    & \multirow{2}{*}{$-c^\cash -c_\crime^\cash$} \\
  & & & \\ \cline{2-4}
\end{tabular}
\caption{Additive counterfactual utilities in the application study. $c^\cash$: cost of cash bail. $c_\crime^\cash$: cost of an undesirable event under cash bail. $r_{\nocrime}^\ROR$: regret of an absence of undesirable event under ROR. $r_{\nocrime}^\cash$: regret of an absence of undesirable event under cash bail.} 
\label{tab:binary_empirical}
\end{table}

\subsection{Alternative Specification}
\label{alternative_loss}
In this section, we consider alternative specification of additive counterfactual utility with binary decision and outcome, for a better interpretability.
We begin by standardizing one of the eight parameters: the utility of the absence of an undesirable event under an ROR decision. For the remaining three utility parameters for the realized outcome under the observed decision, $u_{y}^{d}$, we consider two additive costs: $c_{y} \ge 0$, representing the cost of an undesirable outcome, and $c_{d} \ge 0$, representing the cost of cash bail.
\begin{align*}
    u_{\nocrime}^{\ROR} &= 1\\
    u_{\crime}^{\ROR} &= 1 - c_{y} \\
    u_{\nocrime}^{\cash} &= 1 - c_{d}  \\
    u_{\crime}^{\cash} &= 1 - c_{d} - c_{y} 
\end{align*}

Now, turning to the four counterfactual utility parameters for the counterfactual outcome under the alternative decision, $\tilde{u}_{y}^{1-d}$, we assume a similar additive utility structure with two costs, $c_{y}$ and $c_{d}$. To allow the counterfactual utility parameters to differ from the observed ones in both magnitude and sign, we introduce a discount factor, $\gamma \ge 0$.
\begin{align*}
    \tilde{u}_{\nocrime}^{\ROR} = -\gamma u_{0}^{0} &= -\gamma \\
    \tilde{u}_{\crime}^{\ROR} = -\gamma u_{0}^{1} &= -\gamma(1 - c_{y}) \\
    \tilde{u}_{\nocrime}^{\cash} = -\gamma u_{1}^{0} &= -\gamma(1 - c_{d})  \\
    \tilde{u}_{\crime}^{\cash} = -\gamma u_{1}^{1} &= -\gamma(1 - c_{d} - c_{y}) 
\end{align*}

Recall that in the binary case, the existing risk score framework imposes additional constraints requiring that utilities be equal between the \textit{Safe} and \textit{Backlash} strata and between the \textit{Preventable} and \textit{Hopeless} strata; that is, $\tilde{u}_{\nocrime}^{\cash} = \tilde{u}_{\crime}^{\cash}$ and $u_{\nocrime}^{\cash} = u_{\crime}^{\cash}$. 
In other words, if we were to impose the utility structure with three parameters---two costs and a discount factor---the risk score framework implicitly assumes a zero cost of undesirable outcome, i.e., $c_{y} = 0$.
Yet, this would make all utilities, for each combination of the decision and principal strata, equal to $1 - c_d - \gamma$.

\begin{table}[H]
\centering 
    \begin{tabular}{l|c|c|c|}
  \multicolumn{2}{c}{} & \multicolumn{2}{c}{\textbf{Decision}} \\\cline{3-4}
  \multicolumn{2}{c|}{} 
    & \multirow{2}{*}{Release $(D^\ast = 0)$} 
    & \multirow{2}{*}{Cash bail $(D^\ast = 1)$} \\
  \multicolumn{2}{c|}{} & & \\ \cline{2-4}

  & \multirow{2}{*}{\shortstack{Safe\\$(Y(0) = 0,\ Y(1) = 0)$}} 
    & \multirow{2}{*}{$1 - \gamma(1 - c_{d})$} 
    & \multirow{2}{*}{$1 - c_{d} - \gamma$} \\
  & & & \\ \cline{2-4}

  \multirow{3}{*}{\textbf{Principal}}
  & \multirow{2}{*}{\shortstack{Backlash\\$(Y(0) = 0,\ Y(1) = 1)$}} 
    & \multirow{2}{*}{$1 - \gamma(1 - c_{d} - c_{y})$} 
    & \multirow{2}{*}{$1 - c_{d} - c_{y} - \gamma$} \\
  & & & \\ \cline{2-4}

  \textbf{Strata}
  & \multirow{2}{*}{\shortstack{Preventable\\$(Y(0) = 1,\ Y(1) = 0)$}} 
    & \multirow{2}{*}{$1 - c_{y} - \gamma(1 - c_{d})$} 
    & \multirow{2}{*}{$1 - c_{d} - \gamma(1 - c_{y})$} \\
  & & & \\ \cline{2-4}

    & \multirow{2}{*}{\shortstack{Hopeless\\$(Y(0) = 1,\ Y(1) = 1)$}} 
    & \multirow{2}{*}{$1 - c_{y} - \gamma(1 - c_{d} - c_{y})$} 
    & \multirow{2}{*}{$1 - c_{d} - c_{y} - \gamma(1 - c_{y})$} \\
  & & & \\ \cline{2-4}
\end{tabular}
\caption{Additive counterfactual utilities in the application study. $c_y$ represents cost of undesirable outcome; $u_d$, cost of cash bail; $\gamma$, a discount factor.} 
\label{tbl:add_utility_app}
\end{table}

Table~\ref{tbl:add_utility_app} summarizes the additive counterfactual utilities using three parameters: $c_y$, $c_d$, and $\gamma$. For example, if $\gamma = 1$, meaning that the decision maker weighs the utility of the realized outcome under the observed decision the same as the counterfactual utility, then the utility under release is symmetric to that under cash bail (same magnitude but opposite sign). In this case, the utility of the safe stratum is also the same as that of the hopeless stratum.

\section{Additional Results}

\subsection{Human versus PSA}

\begin{figure}[H]
    \centering
    \includegraphics[width=\linewidth]{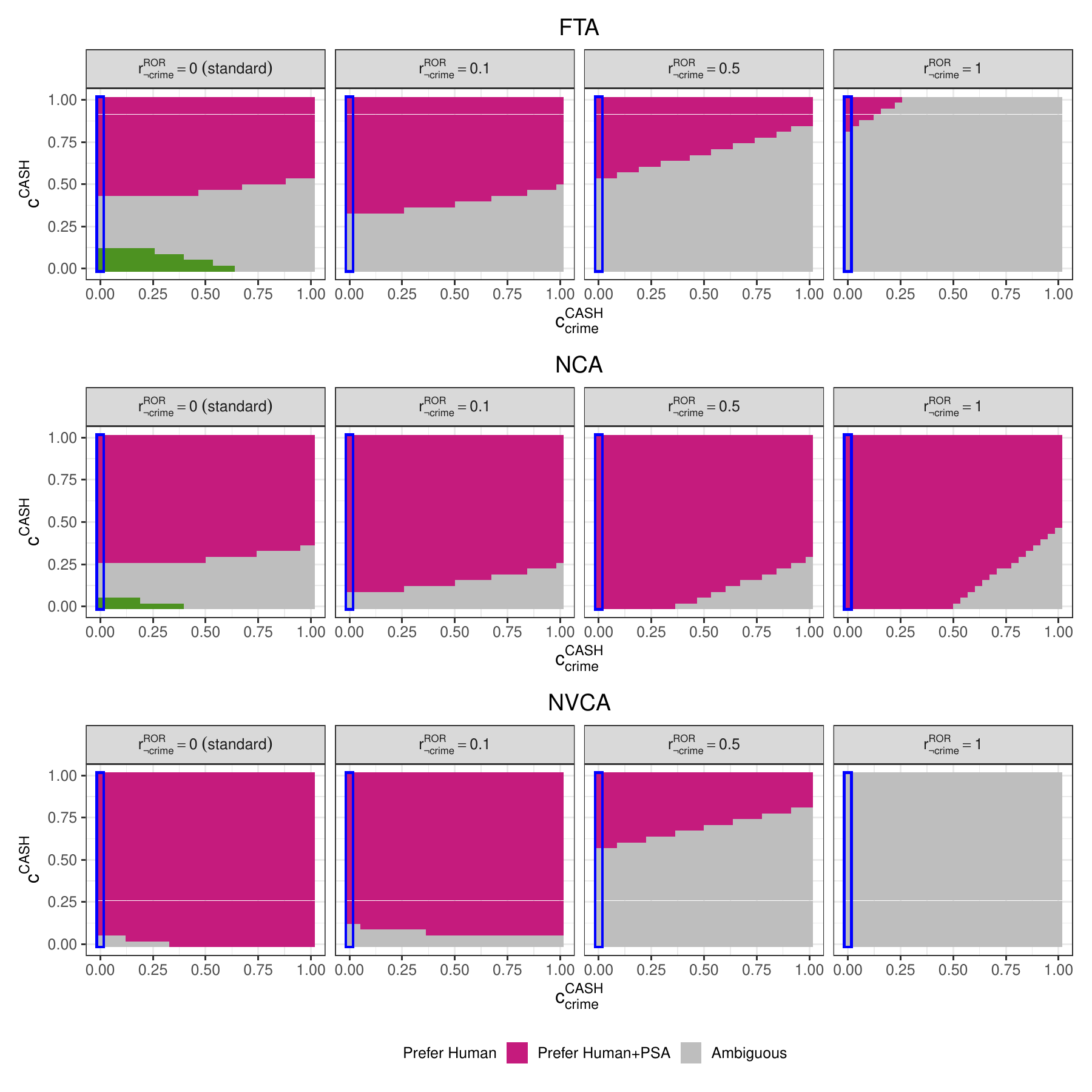}
    \caption{Estimated preference for human decisions over PSA-guided decisions. Each panel corresponds to a different regret parameter $r_{\nocrime}^{\ROR}$; $x$-axis = cost of crime under cash bail $c_{\crime}^{\cash}$; $y$-axis = cost of cash bail $c_d$; $c_{\crime}^{\ROR} = 1$ and $r_{\nocrime}^{\cash} = r_{\nocrime}^{\ROR} c_{\crime}^{\cash}$. Blue region at $c_{\crime}^{\cash} = 0$: risk score framework. First panel ($r_{\nocrime}^{\ROR} = 0$): standard decision framework.}
\end{figure}

\subsection{Estimated Utility}

\begin{figure}[H]
    \centering
    \includegraphics[width=\linewidth, trim=0 0 0 0.4in, clip]{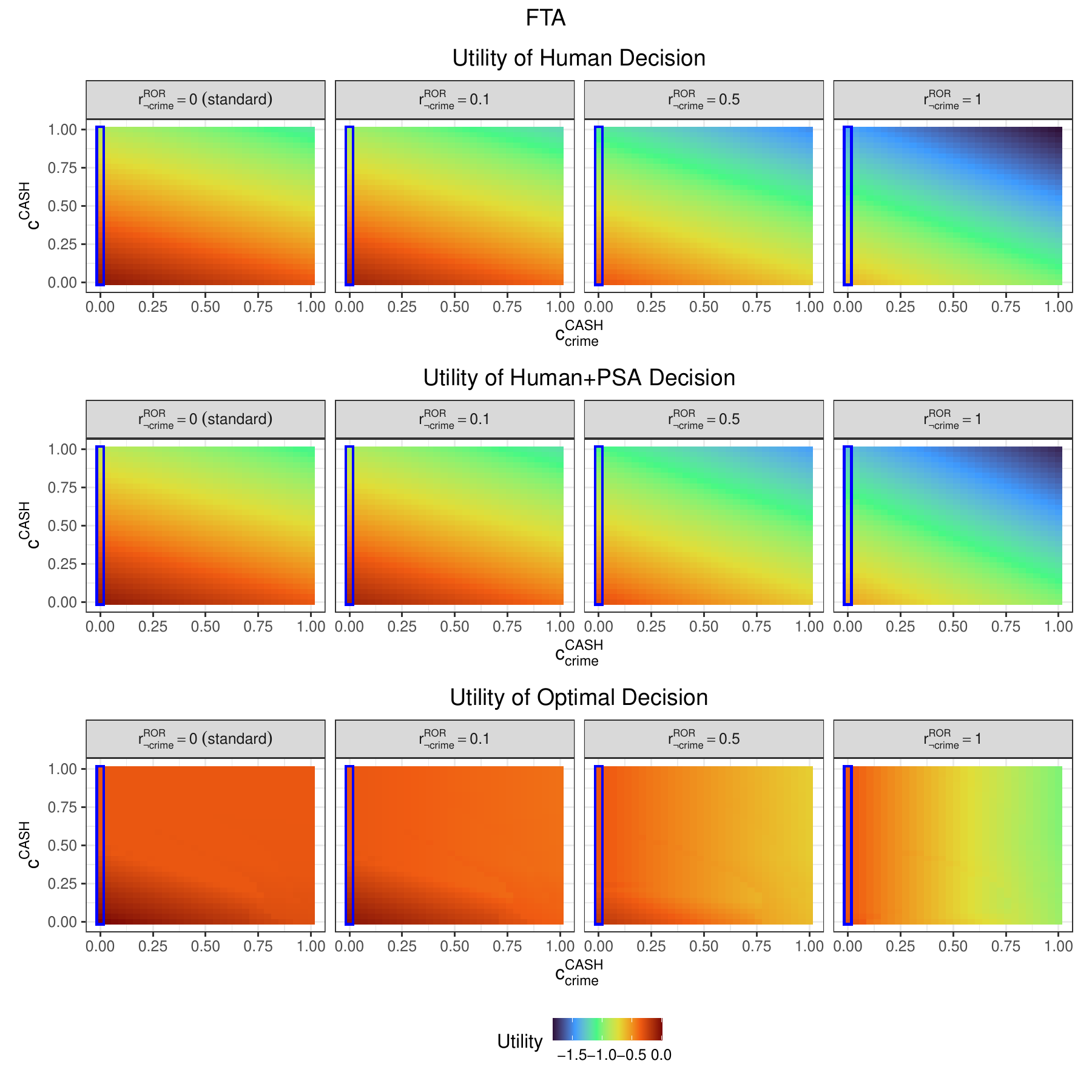}
    \caption{Estimated utility of different decision-making regimes (FTA outcome). 
    The optimal decision tree policy has maximum depth 2 and minimum leaf size 10.
    Each panel corresponds to a different regret parameter $r_{\nocrime}^{\ROR}$; $x$-axis = cost of crime under cash bail $c_{\crime}^{\cash}$; $y$-axis = cost of cash bail $c_d$; $c_{\crime}^{\ROR} = 1$ and $r_{\nocrime}^{\cash} = r_{\nocrime}^{\ROR} c_{\crime}^{\cash}$. Blue region at $c_{\crime}^{\cash} = 0$: risk score framework. First panel ($r_{\nocrime}^{\ROR} = 0$): standard decision framework.}
\end{figure}

\begin{figure}[H]
    \centering
    \includegraphics[width=\linewidth, trim=0 0 0 0.4in, clip]{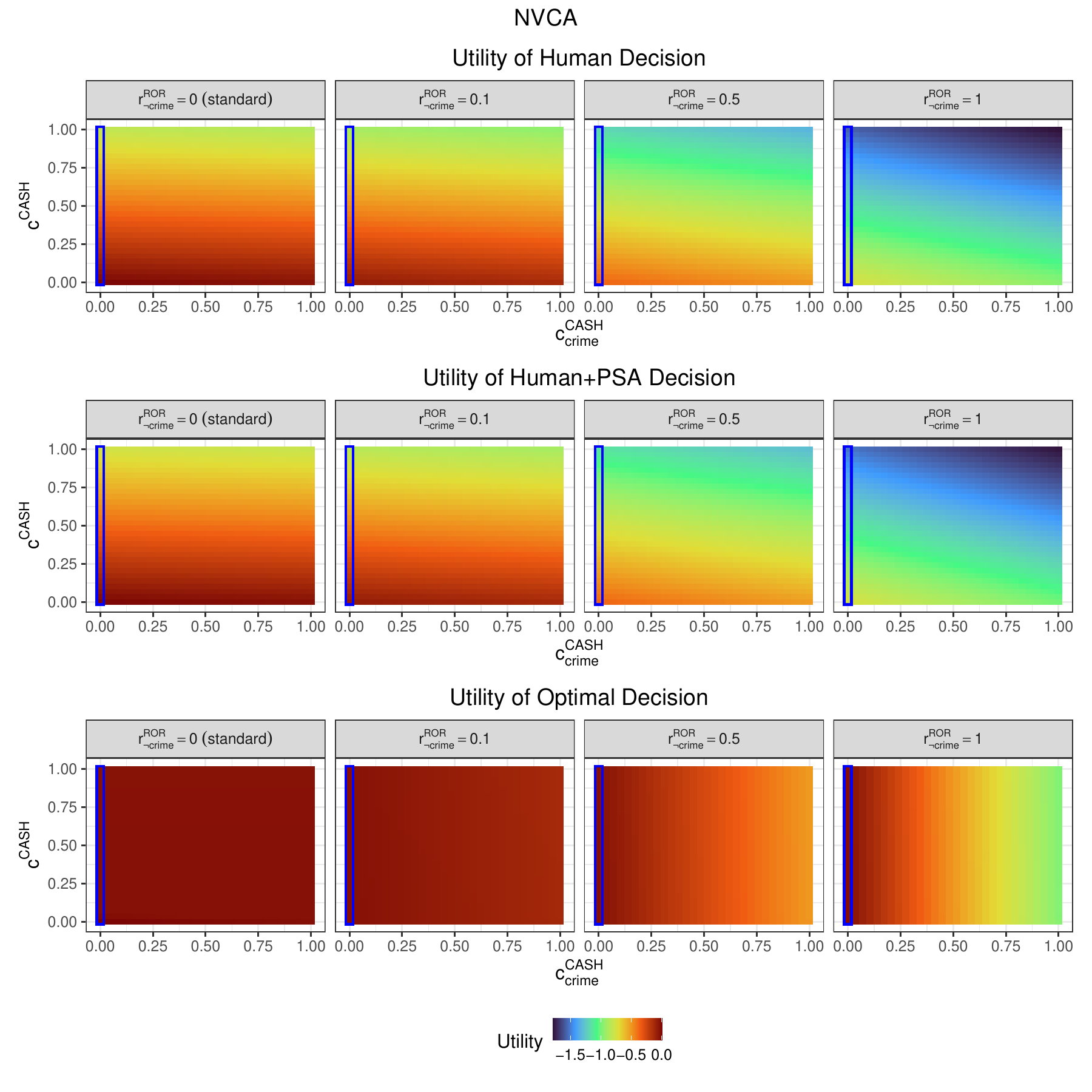}
    \caption{Estimated utility of different decision-making regimes (NVCA outcome). 
    The optimal decision tree policy has maximum depth 2 and minimum leaf size 10.
    Each panel corresponds to a different regret parameter $r_{\nocrime}^{\ROR}$; $x$-axis = cost of crime under cash bail $c_{\crime}^{\cash}$; $y$-axis = cost of cash bail $c_d$; $c_{\crime}^{\ROR} = 1$ and $r_{\nocrime}^{\cash} = r_{\nocrime}^{\ROR} c_{\crime}^{\cash}$. Blue region at $c_{\crime}^{\cash} = 0$: risk score framework. First panel ($r_{\nocrime}^{\ROR} = 0$): standard decision framework.}
\end{figure}

\subsection{Results with $u_0^0 = 1$}
\label{app:u00_1}

\begin{figure}[H]
    \centering
    \includegraphics[width=\linewidth]{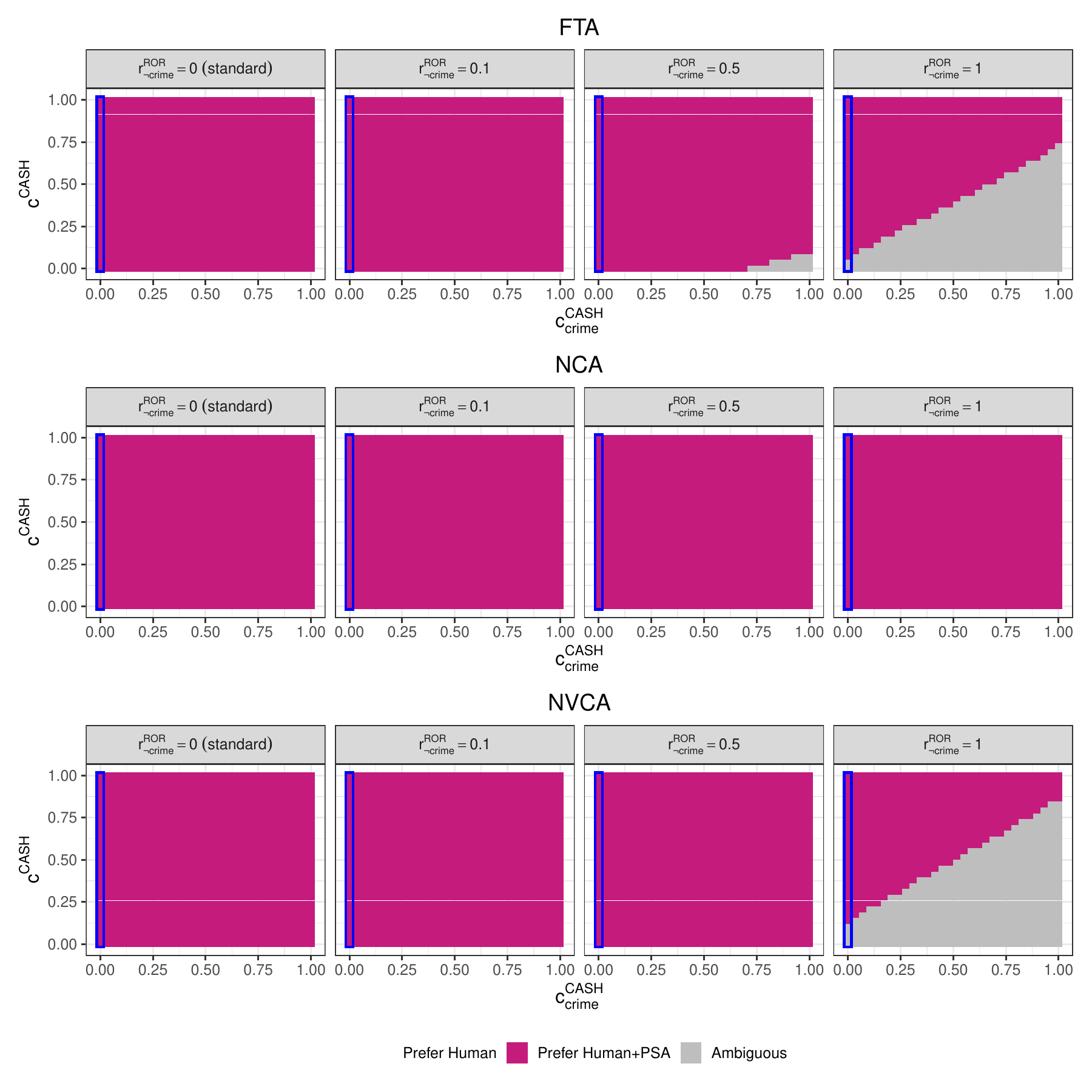}
    \caption{Estimated preference for human decisions over PSA-guided decisions ($u_0^0 = 1$). Each panel corresponds to a different regret parameter $r_{\nocrime}^{\ROR}$; $x$-axis = cost of crime under cash bail $c_{\crime}^{\cash}$; $y$-axis = cost of cash bail $c_d$; $c_{\crime}^{\ROR} = 1$ and $r_{\nocrime}^{\cash} = r_{\nocrime}^{\ROR} c_{\crime}^{\cash}$. Blue region at $c_{\crime}^{\cash} = 0$: risk score framework. First panel ($r_{\nocrime}^{\ROR} = 0$): standard decision framework.}
\end{figure}

\begin{figure}[H]
    \centering
    \includegraphics[width=\linewidth]{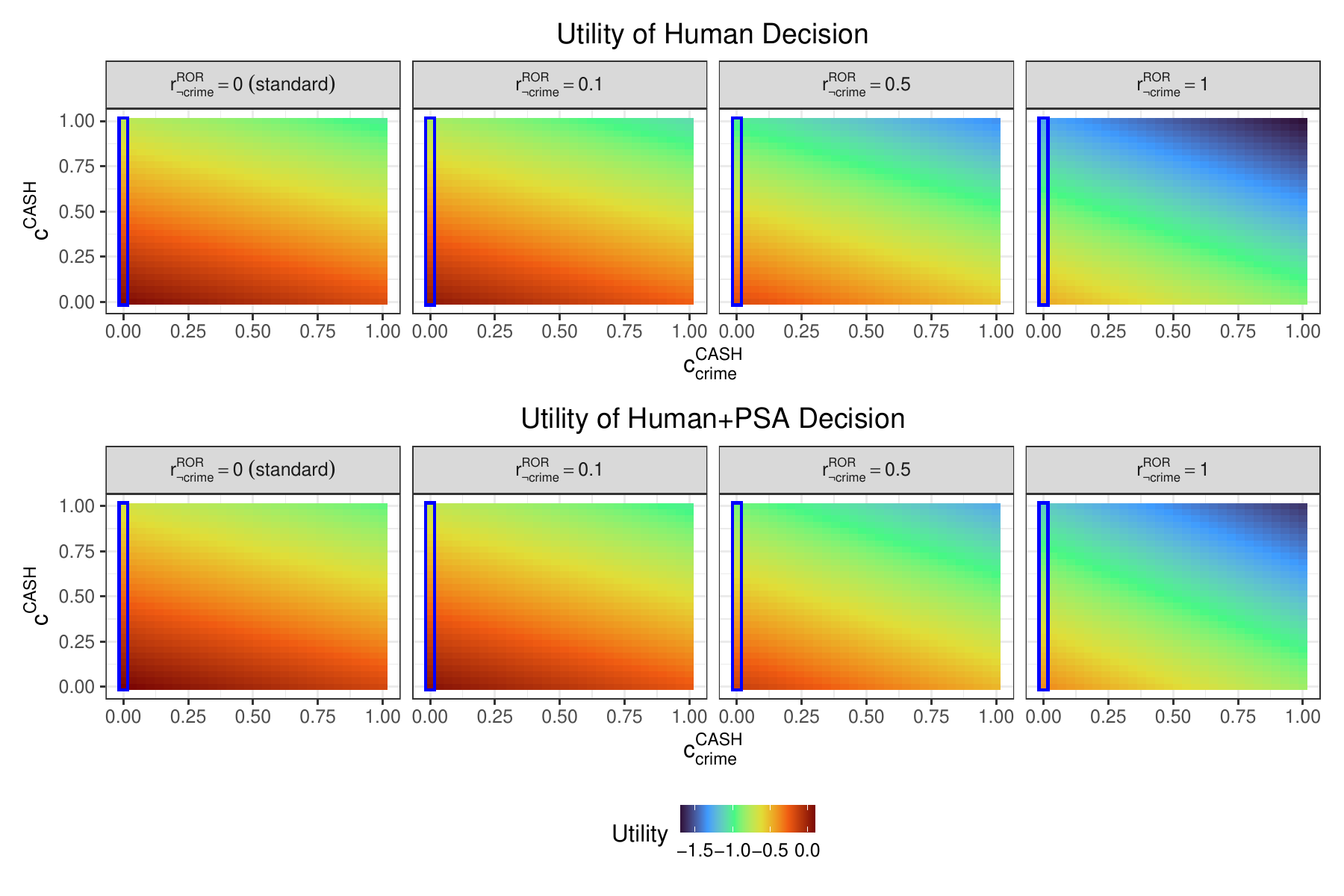}
    \caption{Estimated utility of different decision-making regimes (FTA outcome, $u_0^0 = 1$). Each panel corresponds to a different regret parameter $r_{\nocrime}^{\ROR}$; $x$-axis = cost of crime under cash bail $c_{\crime}^{\cash}$; $y$-axis = cost of cash bail $c_d$; $c_{\crime}^{\ROR} = 1$ and $r_{\nocrime}^{\cash} = r_{\nocrime}^{\ROR} c_{\crime}^{\cash}$. Blue region at $c_{\crime}^{\cash} = 0$: risk score framework. First panel ($r_{\nocrime}^{\ROR} = 0$): standard decision framework.}
\end{figure}

\begin{figure}[H]
    \centering
    \includegraphics[width=\linewidth]{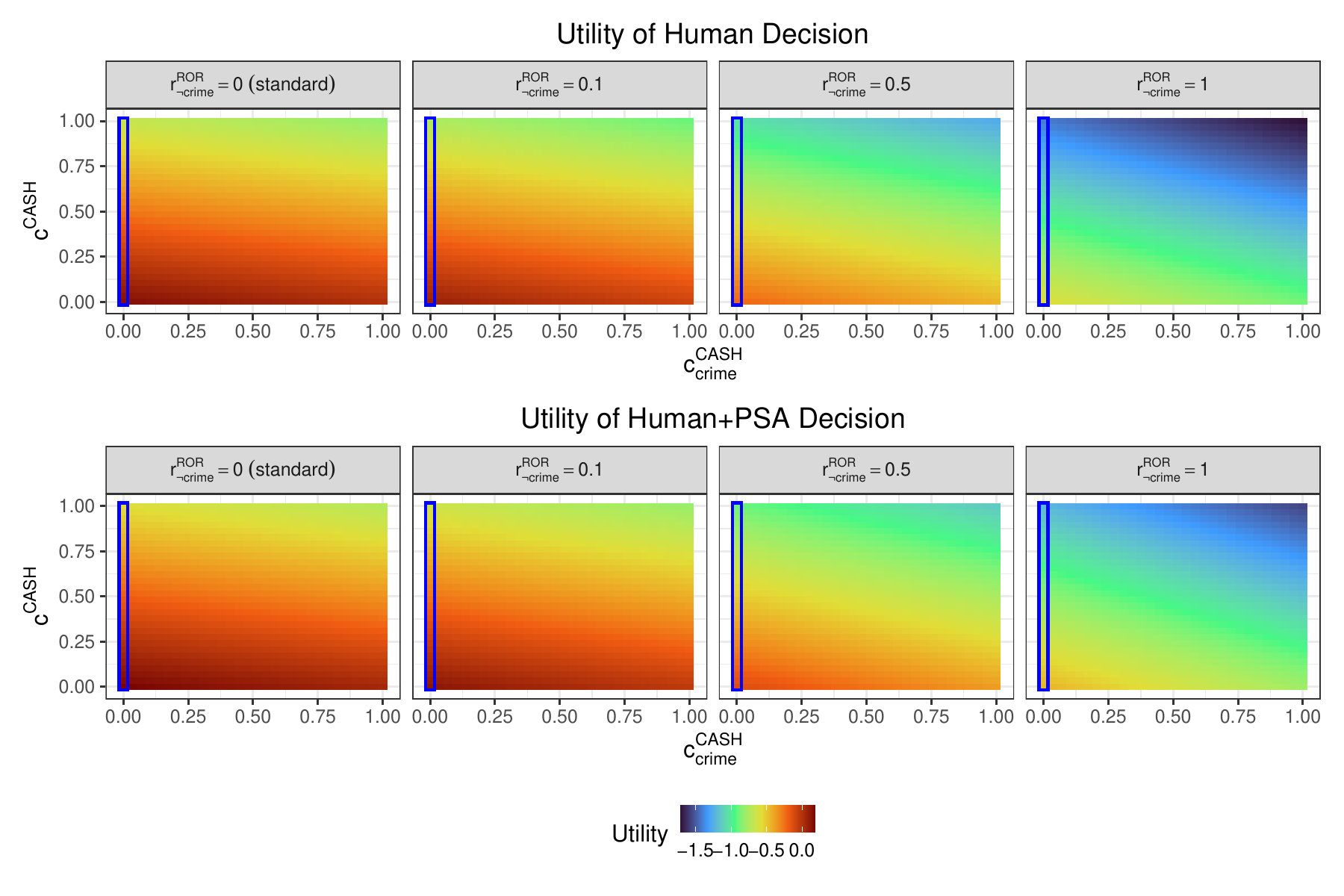}
    \caption{Estimated utility of different decision-making regimes (NCA outcome, $u_0^0 = 1$). Each panel corresponds to a different regret parameter $r_{\nocrime}^{\ROR}$; $x$-axis = cost of crime under cash bail $c_{\crime}^{\cash}$; $y$-axis = cost of cash bail $c_d$; $c_{\crime}^{\ROR} = 1$ and $r_{\nocrime}^{\cash} = r_{\nocrime}^{\ROR} c_{\crime}^{\cash}$. Blue region at $c_{\crime}^{\cash} = 0$: risk score framework. First panel ($r_{\nocrime}^{\ROR} = 0$): standard decision framework.}
\end{figure}

\begin{figure}[H]
    \centering
    \includegraphics[width=\linewidth]{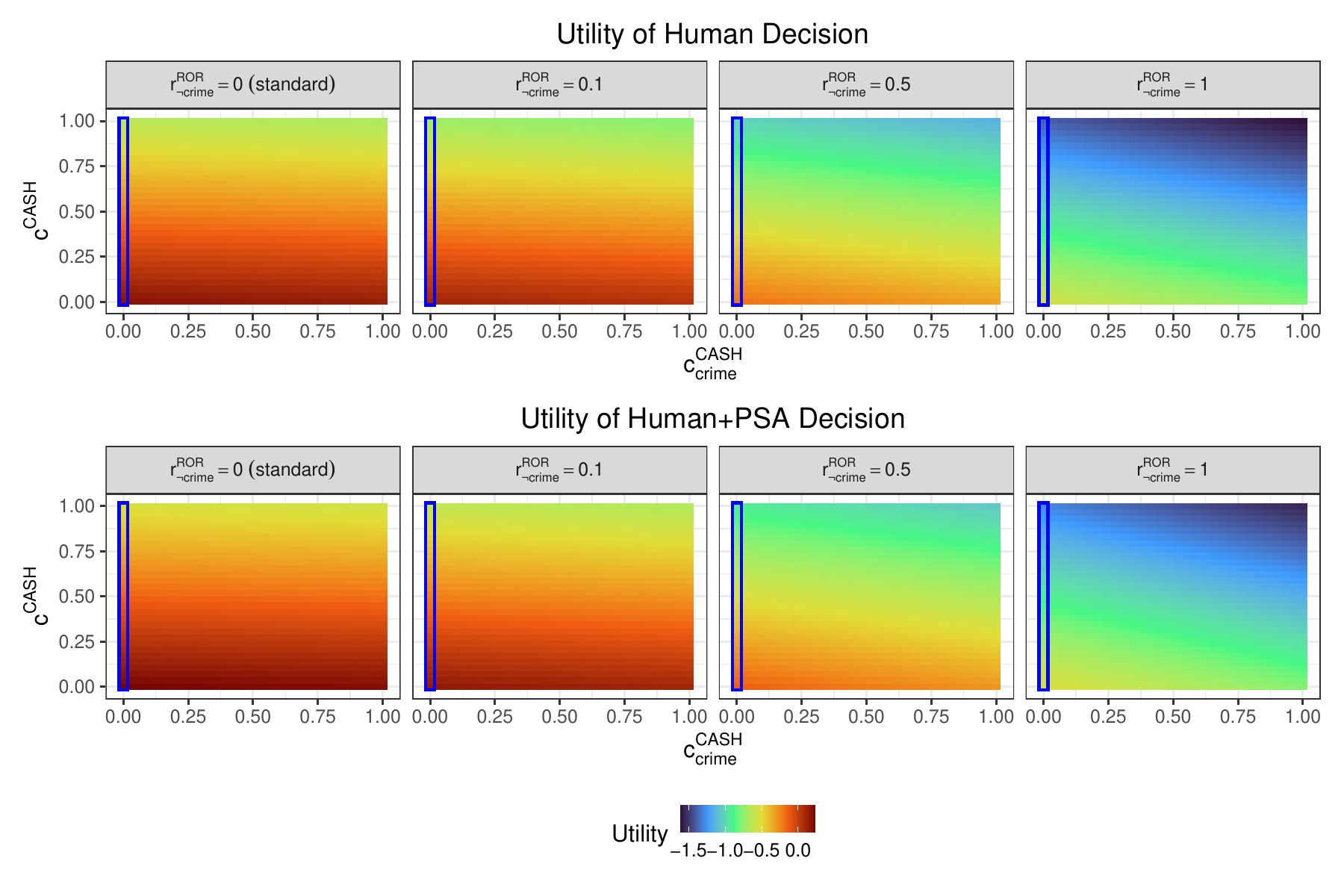}
    \caption{Estimated utility of different decision-making regimes (NVCA outcome, $u_0^0 = 1$). Each panel corresponds to a different regret parameter $r_{\nocrime}^{\ROR}$; $x$-axis = cost of crime under cash bail $c_{\crime}^{\cash}$; $y$-axis = cost of cash bail $c_d$; $c_{\crime}^{\ROR} = 1$ and $r_{\nocrime}^{\cash} = r_{\nocrime}^{\ROR} c_{\crime}^{\cash}$. Blue region at $c_{\crime}^{\cash} = 0$: risk score framework. First panel ($r_{\nocrime}^{\ROR} = 0$): standard decision framework.}
\end{figure}

\end{document}